\documentclass[11pt,letterpaper]{article}
\usepackage[letterpaper,margin=1in]{geometry}

\usepackage[utf8]{inputenc}
\usepackage[T1]{fontenc}
\usepackage{microtype}
\usepackage{amssymb}
\usepackage{amsmath}
\usepackage[fleqn,thmmarks]{ntheorem}
\usepackage{mathtools}
\usepackage{fancyhdr}
\usepackage{ifthen}
\usepackage{cite}
\usepackage{listings}
\RequirePackage{caption}
\usepackage{tikz}
\usetikzlibrary{decorations.pathreplacing,arrows.meta,calc}\usepackage{titling}
\usepackage{booktabs}
\usepackage{authblk}

\usepackage{algorithm}
\usepackage{algpseudocode}

\usepackage{aliascnt}
\usepackage{bm}

\RequirePackage{hyperref}
\RequirePackage[nameinlink]{cleveref}

\definecolor{ETHblue}{HTML}{1F407A}
\definecolor{ETHlightblue}{HTML}{1269B0}
\definecolor{ETHcyan}{HTML}{007A96}
\definecolor{ETHred}{HTML}{A8322D}
\definecolor{ETHgreen}{HTML}{485A2C}
\definecolor{ETHlightgreen}{HTML}{72791C}\usetikzlibrary{decorations.pathreplacing,arrows.meta,calc}
\definecolor{ETHviolette}{HTML}{91056A}
\definecolor{ETHgray}{HTML}{6F6F6F}
\definecolor{ETHochery}{HTML}{956013}

\hypersetup{
  linktocpage,
  colorlinks,
  linkcolor=ETHblue,
  citecolor=ETHblue,
  urlcolor=ETHblue,
  filecolor=ETHblue,
  plainpages=false
}

\pretitle{\centering\LARGE\bfseries}
\posttitle{\par\vskip 0.5em}

\preauthor{\centering}
\postauthor{\par\vskip 0.5em}

\predate{}
\postdate{}
\date{}

\newcommand{\qedsymb}{\ensuremath{\Box}}

\theoremstyle{plain}
\theoremseparator{.}

\newtheorem{theorem}{Theorem}[section] \newaliascnt{lemma}{theorem} \newtheorem{lemma}[lemma]{Lemma} \aliascntresetthe{lemma} \newaliascnt{corollary}{theorem} \newtheorem{corollary}[corollary]{Corollary} \aliascntresetthe{corollary} \newaliascnt{observation}{theorem}  \aliascntresetthe{observation}  \newaliascnt{conjecture}{theorem}  \aliascntresetthe{conjecture} \newaliascnt{definition}{theorem}  \aliascntresetthe{definition}

\makeatletter
\newtheoremstyle{proofplain}
  {\item[\theorem@headerfont\hskip\labelsep ##1\theorem@separator]}
  {\item[\theorem@headerfont\hskip\labelsep ##3\theorem@separator]}
\makeatother

\theoremstyle{proofplain}
\theoremseparator{.}
\theorembodyfont{\normalfont}
\theoremheaderfont{\normalfont\itshape}
\theoremsymbol{\qedsymb}
\newtheorem{proof}{Proof}

\crefname{chapter}{Chapter}{Chapters}\crefname{section}{Section}{Sections}\crefname{subsection}{Subsection}{Subsections}
\crefname{definition}{Definition}{Definitions}\crefname{example}{Example}{Examples}\crefname{figure}{Figure}{Figures}
\crefname{table}{Table}{Tables}\crefname{theorem}{Theorem}{Theorems}\crefname{lemma}{Lemma}{Lemmata}
\crefname{corollary}{Corollary}{Corollaries}\crefname{observation}{Observation}{Observations}\crefname{fact}{Fact}{Facts}
\crefname{conjecture}{Conjecture}{Conjectures}\crefname{equation}{}{}\crefname{enumi}{}{}

\newcommand{\eps}{\ensuremath{\varepsilon}}

\usepackage[
    createShortEnv,
    commandRef=Cref,
    conf={no link to proof, restate}
]{proof-at-the-end}

\author[1]{Kostas Lakis}
\author[1]{Johannes Lengler}
\author[1]{Adeline Pittet}
\affil[1]{Department of Computer Science, ETH Zurich, Switzerland}

\title{Distributed Colouring with \texorpdfstring{\boldmath{$4/3\chi$}}{4/3 chi} Colours\\ for Hyperbolic Random Graphs}

\begin{document}

\maketitle

\begin{abstract}

We study distributed vertex colouring on Hyperbolic Random Graphs (HRGs), a geometric random graph model capturing key structural features of real-world networks. This provides a natural setting for analysing distributed algorithms beyond worst-case general graphs. We introduce \emph{Sequential Radial Colouring}, a CONGEST algorithm using only efficient local computation. The algorithm achieves a near-optimal palette, colouring HRGs with $\tfrac{4}{3}\chi$ colours and running in $O((\log\log n)^2)$ rounds a.a.s. We also give a variant that speeds this up to $O(\log\log n)$ rounds a.a.s., at the price of using $O(\chi\log\log n)$ colours. Finally, for every constant $\varepsilon>0$, it runs in $O(1)$ rounds a.a.s. when $\chi^{1+\varepsilon}$ colours are used. This greatly reduces the number of colours over the previous constant-round algorithm of Maus and Ruff~\cite{mausruff2026distributed} by a factor of at least $n^{1/6}$.

Our analysis contains a phase in which we consider a classical randomised colouring protocol on a (large) clique of the graph. We also delve deeper into this part of the analysis and improve upon previous results for colouring a clique $\mathcal C$, bounding the number of rounds required as a function of the additive slack $s = |\Psi| - \chi$, where $\Psi$ is the set of colours used. In particular, constant-round colouring is possible if and only if $s=|\mathcal C|^{1+\Omega(1)}$, while \(s=|\mathcal C|/\log |\mathcal C|\) already gives the optimal \(\Theta(\log\log |\mathcal C|)\) round complexity.
\end{abstract}
\thispagestyle{empty}



\section{Introduction}

Vertex colouring is one of the most central problems in the area of distributed computation. 
In the LOCAL model, each vertex of a graph is considered as an entity with unlimited computational capabilities. Algorithms proceed in synchronous rounds, wherein any two adjacent vertices can communicate information of arbitrary size. In the CONGEST model each edge can only transmit $O(\log n)$ bits per round.
A vast body of literature has been dedicated to distributed colouring in the LOCAL model, mostly for a number of colours close to the maximum degree $\Delta=\Delta(G)$ of the underlying graph $G$, ideally with $\Delta+1$ colours. Results also exist for the CONGEST model, but usually for restricted graph classes.
We discuss some of these results in~\Cref{sec:related_work_colouring} but we also refer the reader to~\cite{grebik2025descriptive} which includes a recent survey. The rough picture is that these problems are solvable in $\log^{O(1)}n$ rounds deterministically and $(\log \log n)^{O(1)}$ rounds using randomness.

Through the classical worst-case lens, it is natural to parametrize the number of colours in terms of $\Delta$ since colouring with fewer colours is hard for worst-case instances. However, many real networks have a heavy-tailed degree distribution, which means that $\Delta$ is very large, for example $n^{\Omega(1)}$ if the degree sequence follows a power law. For such instances, the chromatic number $\chi=\chi(G)$ may be asymptotically much smaller than $\Delta$. In this case, guarantees in terms of $\Delta$ are rather weak.

For this reason, Maus and Ruff~\cite{mausruff2026distributed} have started to study the problem of distributed colouring in models of complex networks, specifically in the model of \emph{Hyperbolic Random Graphs (HRG)}~\cite{krioukov2010hyperbolic}. This choice is motivated by the fact that this model has been found to be a good qualitative and empirical match for real networks with heavy-tailed degree distributions, including social, technological, and biological networks~\cite{boguna2010sustaining,serrano2008self,papadopoulos2012popularity}. In a nutshell, the HRG model samples vertices as points in a hyperbolic disk (with intensity depending on $\alpha$, a parameter controlling the degree distribution) and connects pairs of vertices with hyperbolic distance at most the radius of the disk. Vertices closer to the center have exponentially higher degrees. See~\Cref{sec:hyperbolic_random_graphs} for a definition and~\Cref{fig:both-figures} for an example.

The approach of~\cite{mausruff2026distributed} is related to the average-case analysis that was popular for a short time in the 70s~\cite{Karp1976probabilistic}. However, the results obtained for graph algorithms during this period failed to transfer into practice because the graph models were not representative of real networks. Notably, this has changed in the recent years. HRG and its generalisation of \emph{Geometric Inhomogeneous Random Graphs (GIRG)}~\cite{bringmann2019geometric} are statistically a much better match for real-world networks~\cite{faloutsos1999power,gugelmann2012random,blasius2016hyperbolic,boguna2010sustaining,papadopoulos2012popularity,krioukov2010hyperbolic,boguna2021network,papadopoulos2014network}. More importantly, theoretical and empirical runtime results for HRG and GIRG have been shown to transfer to real-world networks for several fundamental algorithms. The most prominent examples are local routing algorithms~\cite{boguna2010sustaining,blasius2021strongly,papadopoulos2010greedy,bringmann2022greedy} and algorithms for shortest paths via bidirectional breadth-first search~\cite{blasius2024external,cerf2024balanced}. Other examples include computation of the diameter, vertex cover, enumeration of maximal cliques, or the Louvain algorithm for graph clustering~\cite{blasius2024external}.

Arguably, for real-world networks with heavy-tailed degrees it is important to use the CONGEST model instead of the LOCAL model. The reason is that average distances are ultra-small in such networks. In the models, a $(1-o(1))$ fraction of all pairs of vertices in the giant component have distance $O(\log \log n)$, and this includes \emph{all} pairs of vertices of degree larger than $\log^c n$ for sufficiently large $c$~\cite{bringmann2025average}. Those small distances are not an artifact of the models, but are instead quite realistic. For example, the Facebook network with $\approx 1.6\cdot 10^9$ vertices has a staggeringly small average distance of $4.57$ hops between vertices in its giant component~\cite{edunov2016three}, with hardly any distance being larger than $7$~\cite{backstrom2012four}. As a result, such networks permit an abusive LOCAL algorithm: In $O(\log\log n)$ rounds, the vertices broadcast the subgraph induced by vertices of degree at least $\log^c n$. Once this subgraph is known, its vertices can locally compute a $\chi$-colouring; the remaining task of colouring vertices of degree less than $\log^c n=\chi^{o(1)}$ is straightforward.

In the HRG model, Maus and Ruff~\cite{mausruff2026distributed} recently gave two main results. They gave an algorithm in the LOCAL model which uses $(1+o(1))\chi$ colours and requires $O(1)$ rounds. However, as they frankly admit, this uses the above somewhat abusive method on a large subgraph with constant diameter, and it is not clear how to transfer the result to the CONGEST model. In the main content of their paper they describe efficient (in fact requiring constantly many or even 2 rounds in some cases) algorithms using substantially fewer than $\Delta$ colours in the CONGEST model, in particular $\Delta^{1-\delta}$ colours suffice for a suitable constant $\delta>0$. However, the number of colours required by their algorithm is still substantially larger than the chromatic number \(\chi\), by more than a factor \(n^{1/6}\). 

\textbf{Our contribution.} We show in this paper that colouring with $O(\chi)$ colours in $O((\log \log n)^2)$ rounds is possible for HRG in the CONGEST model. The hidden constant for the number of colours depends on the tail of the degree distribution and lies in the interval $(\sqrt{4/3},4/3)$. 
Moreover, we give faster algorithms if the number of colours is slightly larger, in particular we show that a constant number of rounds suffices if we have $\chi^{1+\eps}$ colours, for any constant $\eps>0$. Compared to the previous constant-round CONGEST algorithm of Maus and Ruff, this decreases the number of colours by a factor of at least \(n^{1/6}\), for \(\eps\) chosen sufficiently small.

Apart from the improved algorithms, we also provide some structural insights into why it is possible to use much fewer than $\Delta$ colours. Instead of $\Delta$, we propose the \emph{degeneracy} $\kappa$ as a more natural bottleneck for HRG. The degeneracy $\kappa=\kappa(G)$ of a graph $G$ is the largest minimum degree in any induced subgraph of $G$. Equivalently, a graph has degeneracy $\kappa$ if there is an ordering of the vertices such that any vertex has at most $\kappa$ preceding neighbours in this ordering. If the ordering is known, a sequential greedy algorithm can trivially colour the vertices in that order with $\kappa+1$ colours. For distributed algorithms, there are two challenges: computing the ordering of vertices and overcoming the sequential nature of the greedy algorithm. Keep in mind that the degeneracy, clique number and chromatic number are of the same order $\Theta(n^{1 - \alpha})$ in HRG~\cite{baguley2025hrg}. For the first point, we show that in HRG the degrees give an almost optimal ordering: every vertex has at most $(1+o(1))\kappa$ neighbours of larger degree. Since $\kappa \le (\tfrac{4}{3} + o(1))\chi$ for HRG a.a.s.~\cite{baguley2025hrg}, this implies that we need only $O(\chi)$ colours. Indeed, we show that our algorithms still work with $(1+\eps)\kappa$ colours for any constant $\eps>0$. While $\kappa \le (\tfrac{4}{3} + o(1))\chi$ is only an upper bound, we conjecture that there is indeed a non-trivial constant factor gap between $\kappa$ and $\chi$. In fact, such a constant factor gap was shown in~\cite{baguley2025hrg} between $\kappa$ and the clique number $\omega = \omega(G)$ for HRG, and we conjecture that $\chi = (1+o(1))\omega$ in HRG. Hence, we believe that it is hard to colour an HRG with $(1+o(1))\chi$ colours, but we also discuss in detail what our results imply depending on whether this conjecture is true or false (\Cref{subsec:proof-near-chromatic-theorem}). 

Even though a good ordering of vertices is given for free in HRG by the degrees, it is still non-trivial to exploit the ordering to obtain a fast algorithm in the CONGEST model. In order to explain our algorithm, we first revisit the \emph{Random Colour Trial (RCT)} algorithms proposed in~\cite{mausruff2026distributed}.
In the basic RCT, each node picks a uniformly random candidate colour from its palette of available colours, i.e., the whole colour palette initialized at the start without the colours used by its already
permanently coloured neighbours, and sends this colour to its neighbours while receiving their
colours. If no neighbour tries to get coloured with the same candidate colour, the node permanently
adopts the colour, otherwise it retries in the next iteration with a
new random candidate colour from its updated palette.
This process repeats until all nodes have a valid colour.
Note that an iteration of RCT actually requires two rounds of communication, as the vertices must inform their neighbours if they have permanently adopted the colour they chose in the current iteration.

The paper~\cite{mausruff2026distributed} also introduced two new variants, RCTID and RCTDEG, which assign a priority order among vertices according to random IDs or according to the degrees respectively, where in the latter case ties are broken arbitrarily. In those variants, a node adopts its colour permanently if there is no conflict with neighbours of higher priority, while it ignores neighbours of lower priority. While the basic RCT and RCTID a.a.s.\! find a colouring efficiently with $\eps \Delta$ colours for any constant $\eps >0$, RCTDEG can even find a colouring with $n^{-\delta} \Delta$ colours for an explicit constant $\delta >0$. However, in all cases, the number of colours needed for the algorithms to succeed is more than $n^{1/6}\chi$. Moreover, all three algorithms are shown to fail with $n^{\delta'} \chi$ colours for some constant $\delta' > 0$. So, while they can beat the maximum degree $\Delta$, they cannot approach $\chi$ colours.

\textbf{Our improvement.} The reason why RCTDEG is unable to approach $\chi$ colours is the absence of recourse, i.e.\! the fact that vertices permanently stick to a colour once they have been able to claim it in a round. As a consequence (as shown in~\cite{mausruff2026distributed}), many low-degree vertices claim colours prematurely, which then blocks the whole colour palette for some vertices if only $n^{\delta'}\chi$ colours are available. In order to avoid this problem we \emph{delay} colouring of lower-degree vertices.\footnote{Another natural idea would be that vertices of lower degree can always be ``kicked out'' of their colour by higher-degree vertices in subsequent rounds. However, this turns out to be more difficult to analyse.} More precisely, we partition the set of degrees into intervals $[L_i, U_i]$, which are disjoint and decreasing, so $U_{i+1} = L_i-1$. Our algorithm, \emph{Sequential Radial Colouring (SRC)}, then proceeds in phases: in the $i$th phase only vertices with degrees in the interval $[L_i, U_i]$ are active, and those simply run RCTDEG during that phase. Each phase lasts long enough such that the active vertices manage to colour themselves.

For a colour set $\Psi$, the intervals are chosen such that each vertex with degree in $[L_i,U_i]$ has fewer than $|\Psi|$ neighbours with degrees in $[L_i,\infty)$. Thus, the decomposition itself depends on the palette size, with larger palettes permitting coarser intervals and near-optimal palettes requiring finer ones. Note that such intervals can only exist for $|\Psi| > \kappa$, since their existence implies in particular that every vertex has fewer than $|\Psi|$ neighbours of larger degree. Indeed, we show that $(1+\eps)\kappa$ colours suffice for any constant $\eps >0$. After fixing the degree-tail parameter, \(\kappa\) is a.a.s. strictly smaller than \(\frac{4}{3}\chi\) by a fixed margin, so choosing \(\eps>0\) sufficiently small makes the \((1+\eps)\kappa\)-colouring use fewer than \(\frac{4}{3}\chi\) colours.
We suspect that $\kappa$ is a structural obstacle for efficient distributed colouring. 

With $(1+\eps)\kappa$ colours, our decomposition yields $O(\log \log n)$ intervals and we show that RCTDEG needs $O(\log \log n)$ rounds per interval. This leads to an overall runtime of $O((\log \log n)^2)$ for the SRC algorithm.  For larger values $|\Psi| = \Theta(\kappa \log \log n)$, we still use the same intervals as before. But now we assign a distinct set of $O(\kappa)$ colours to each interval\footnote{This works immediately if the hidden constant is sufficiently large. Otherwise, we use every colour palette for constantly many intervals and colour those sequentially.} and colour them in parallel, so the algorithm terminates in $O(\log \log n)$ rounds. We call this variant \emph{Parallel Radial Colouring (PRC)}. 
Finally, if $|\Psi|=\Theta(\kappa^{1+\eps})$, then a constant number of intervals suffices to cover the vertex set. On each such interval, the multiplicative slack of $\kappa^\eps$ between $|\Psi|$ and the number of preceding neighbours ensures that RCTDEG terminates in a constant number of rounds. Hence, SRC colours the graph in a constant number of rounds. In summary, we prove the following theorem.

\begin{theorem}\label{thm:summary}
Consider colouring an HRG with degeneracy $\kappa$ with $|\Psi|$ colours. For any $\eps >0$, a.a.s.\!
    \begin{enumerate}
        \item If $|\Psi| \ge (1+\eps)\kappa$ then SRC finds a colouring in $O((\log \log n)^2)$ rounds.
        \item If $|\Psi| =\Omega(\kappa \log \log n)$ then PRC finds a colouring in $O(\log \log n)$ rounds.
        \item If $|\Psi| \ge \kappa^{1+\eps}$ then SRC finds a colouring in $O(1)$ rounds.
    \end{enumerate}
Both SRC and PRC run in the CONGEST model and only require efficient computation per node.
\end{theorem}

The first interval is chosen such that we obtain a set $\mathcal C$ that is close to a clique.\footnote{We omit some technical details here since there may be small deviations from the clique structure, but a.a.s.\! the induced subgraph contains at most $(1+o(1))\omega$ vertices and contains a clique of size $(1-o(1))\omega$.} There is some priority order between the vertices depending on the edges outwards from the clique, but the precise ordering is irrelevant for our proof. That is, our proof would also work if the graph were itself just this clique.

We believe that the colouring process on the clique is an interesting process in its own right. For this reason, we analyse this part in more detail in terms of the excess number of colours $s = |\Psi|-|\mathcal C|$. We show that the number of rounds is constant if and only if $s = |\mathcal C|^{1 + \Omega(1)}$, and that it is $O(\log(|\mathcal C|/s) +\log\log n)$ for general $1 \le s = o(|\mathcal C|)$. The case $s=1$ recovers the classical $O(\log |\mathcal C|)$ bound by Johansson~\cite{johansson1999simple}. For slack $s=|\mathcal C|/\log n$, which corresponds to a palette of size $(1+o(1))|\mathcal C|$, the upper and lower bounds coincide at $\Theta(\log\log n)$ rounds. 

\begin{table}[t]
\centering
\renewcommand{\arraystretch}{1.25}
\begin{tabular}{lll}
\toprule
\textbf{Palette size} & \textbf{Upper bound} & \textbf{Lower bound} \\
\midrule
$|\Psi|=\left(1+1/\log n\right)|\mathcal C|$
&
$O(\log\log n)$
&
$\Omega(\log\log n)$
\\
$|\Psi|=|\mathcal C|^{1+\varepsilon_n}$, $|\mathcal C|^{\varepsilon_n}\to\infty$
&
$O(1/\varepsilon_n)$
&
$\Omega\!\left(\log(1+1/\varepsilon_n)\right)$
\\
\bottomrule
\end{tabular}
\caption{Upper and lower bounds for clique colouring with RCTDEG. 
Note that the number of rounds is $O(1)$ if and only if the number of colours is $|\mathcal C|^{1+\Omega(1)}$.}
\label{tab:first-interval-round-bounds}
\end{table}

\subsection{Properties of the HRG model}\label{sec:related_works_HRG}

The HRG model was introduced in~\cite{krioukov2010hyperbolic} and has been studied intensively in its own right. It was later generalized to a broader class of random graph models called Geometric Inhomogeneous Random Graphs (GIRG)~\cite{bringmann2019geometric}, and many of the following results have been obtained in this broader framework.

\textbf{Degree distribution and clustering.} The degree distribution of a random vertex is a power-law~\cite{gugelmann2012random} with exponent $\tau = 2\alpha + 1 \in (2, 3)$. This means that the probability for a random vertex to have degree $k$ is $\Theta(k^{-\tau})$. Likewise, the probability to have degree \emph{at least} $k$ is $\Theta(k^{1-\tau})$. The fact that $\tau \in (2, 3)$ guarantees that the expected degree is constant (i.e.\ the graph is sparse) and at the same time the variance diverges with $n$. For many real-world networks that have been analysed in this manner, the power-law parameter has been found to lie in the range $(2,3)$ as well~\cite{boccaletti2006complex,jeong2000large,broder2000graph,clauset2009power}. In particular, this implies a large maximum degree of $\Delta = \Theta(n^{1/(2\alpha)}) = \Omega(n^{1/2})$. 
Moreover, HRG were shown to exhibit a constant clustering coefficient~\cite{gugelmann2012random}, which is the probability that two random neighbours of a random vertex are themselves connected by an edge.

\textbf{Typical distances and diameter.} The giant component contains $\Theta(n)$ vertices and all other components are of at most polylogarithmic size~\cite{blasius_et_al:LIPIcs.ESA.2023.20,kiwi2019second,bringmann2019geometric}. In accordance with the empirical observation that real-world networks have very small distances, it has been shown in~\cite{bringmann2025average} for HRG that typical distances are of the order $O(\log \log n)$. That is, two vertices chosen uniformly at random from the largest connected component of the graph (the so-called \emph{giant}) have a.a.s.\! such a graph distance. The diameter of the giant (and in fact of any component) was shown to be $O(\log n)$~\cite{muller2019diameter,benjert_et_al:LIPIcs.STACS.2026.11}. Moreover, there exists a constant $c>0$ such that a.a.s.\! all vertices of degree at least $(\log n)^c$ belong to the giant.

\textbf{Algorithmic properties.} The HRG model exhibits sublinear treewidth and 
small separators, of the order $n^{1 - \alpha}$~\cite{blasius2016hyperbolic}. Note that this is the same order as the largest clique in the graph, as seen in the next paragraph. Additionally, samples from the HRG model can be obtained in expected linear time~\cite{bringmann2015sampling}. 
A main motivation for introducing HRG was to explain why local routing algorithms empirically work well when real networks are embedded into a hyperbolic space~\cite{boguna2010sustaining,blasius2021strongly,papadopoulos2010greedy,bringmann2022greedy}. Other algorithms with the reputation of working empirically well on real-world networks, but without worst-case guarantees, were studied empirically on the GIRG model (generalizing HRG) in~\cite{blasius2024external}. Among those, bidirectional breadth-first search has also been theoretically analysed on GIRG~\cite{blasius2024external,cerf2024balanced}.

\textbf{Largest clique, chromatic number, and degeneracy.} Consider the set $\mathcal{C}$ of vertices which satisfy $r \le R / 2$. We will refer to this set as the \emph{core}. Since the hyperbolic distance to the origin is $r$, by triangle inequality these vertices form a clique (because vertices with distance at most $R$ are connected), and thus $|\mathcal{C}| \le \omega \le \chi$. It has been shown~\cite{baguley2025hrg} that $|\mathcal{C}| \le \omega \le ((4/3)^{\alpha/2} + o(1))|\mathcal{C}| = \Theta(n^{1 - \alpha})$. They also show that $\chi \le ((4/3)^{\alpha} + o(1))|\mathcal{C}|$, therefore $\chi = \Theta(n^{1 - \alpha})$. However, since $1/2 < \alpha < 1$, the maximum degree $\Delta$ is of much larger order, namely $\Theta(n^{1/(2\alpha)})$. The previous upper bound on the chromatic number comes directly from upper-bounding the degeneracy $\kappa$ by $((4/3)^{\alpha} + o(1))|\mathcal{C}|$. Interestingly, it was also shown that there exists a constant $c > 1$ such that $\kappa \ge c\omega$.




\begin{figure}[t]
\centering

\begin{minipage}[t]{0.495\textwidth}
    \centering
    \includegraphics[width=0.85\textwidth]{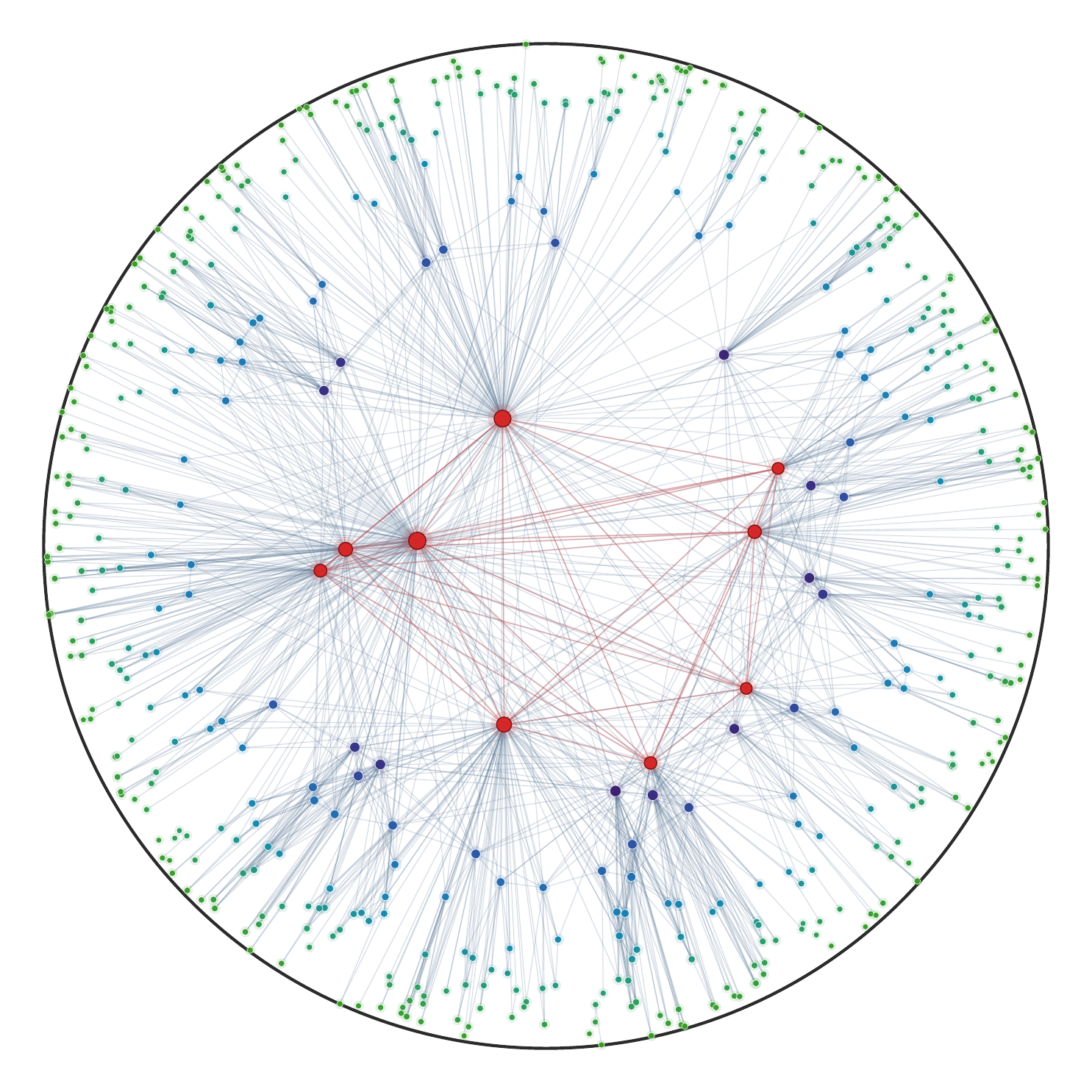}
\end{minipage}
\hfill
\begin{minipage}[t]{0.495\textwidth}
    \centering
\begin{tikzpicture}[
    scale=0.82,
    vertex/.style={circle,fill=black,inner sep=1.2pt},
    prev/.style={circle,fill=blue!75!black,inner sep=1.8pt},
    internal/.style={circle,fill=green!60!black,inner sep=1.8pt},
    future/.style={circle,fill=gray!60,inner sep=1.6pt},
    edge/.style={black!70, line width=0.5pt},
    band/.style={dashed, black!45, line width=0.5pt}
]

\coordinate (v) at (1.75,0.25);
\def\rinner{1.25}
\def\router{2.35}
\def\rbig{3.9}

\fill[blue!8] (0,0) circle (\rinner);
\fill[green!10, even odd rule] (0,0) circle (\router) (0,0) circle (\rinner);
\draw[band] (0,0) circle (\rinner);
\draw[band] (0,0) circle (\router);
\draw[line width=0.8pt] (0,0) circle (\rbig);

\draw[dashed, black!18] (0,0) circle (3.0);

\coordinate (p1) at (0.15,0.90);
\coordinate (p2) at (-0.55,0.40);
\coordinate (p3) at (-0.20,-0.75);

\foreach \p in {p1,p2,p3}{
    \draw[edge] (v)--(\p);
    \node[prev] at (\p) {};
}

\coordinate (i1) at (1.10,1.70);
\coordinate (i2) at (2.45,1.00);
\coordinate (i3) at (2.35,-0.85);
\coordinate (i4) at (1.05,-1.55);

\foreach \p in {i1,i4}{
    \draw[edge] (v)--(\p);
    \node[internal] at (\p) {};
}



\foreach \x/\y in {
    -0.6/1.9, 0.7/2.0, 2.0/1.9, -1.5/0.8, -1.7/-0.8,
    -0.9/-1.8, 0.2/-2.0, 2.0/-1.9, 3.0/0.2, -2.7/1.8,
    -3.0/0.2, -2.5/-1.9, -0.8/3.0, 0.8/3.1, 2.3/2.9,
    3.0/1.0, 3.15/-0.2, 2.75/-2.35, 0.6/-3.0, -1.2/-3.0
}{
    \node[vertex] at (\x,\y) {};
}

\node[font=\scriptsize, blue!60!black] at (-0.1,0.0) {processed};
\node[font=\scriptsize, green!50!black] at (-0.55,1.55) {active band};


\node[circle,fill=black,inner sep=2.4pt] at (v) {};
\node[font=\small, anchor=west] at ($(v)+(0.08,-0.02)$) {$v$};

\path[use as bounding box]
    (-3.55,-4.25) rectangle (3.55,3.55);

\end{tikzpicture}
\end{minipage}
\captionsetup{
        width=0.98\textwidth,
        justification=justified,
        singlelinecheck=false
    }
\caption{
        \textbf{Left:} An example hyperbolic random graph for
        \(n=500\), \(C=1\), and \(\alpha=0.6\). The red vertices
        form a large clique, referred to as the \emph{core}.
        \textbf{Right:} Blue vertices lie in already processed inner
        bands and contribute to \(d_{\mathrm{pre}}(v)\), while green
        vertices lie in the active band and contribute to
        \(d_{\mathrm{int}}(v)\). The sizes of the bands in the illustration are not to scale. There are other edges incident to $v$, which are not depicted.
    }
    \label{fig:both-figures}
\end{figure}
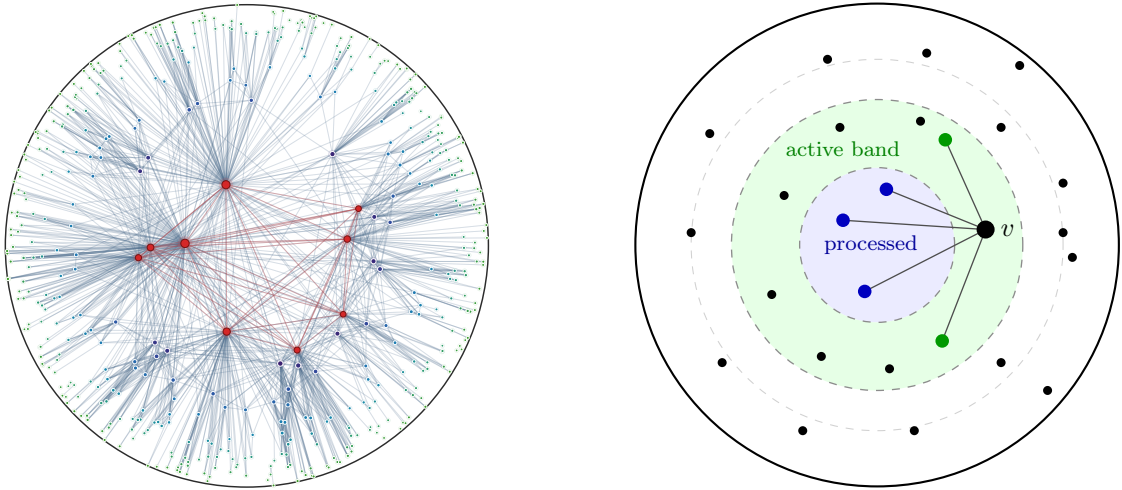

\subsection{Formal definition of Hyperbolic Random Graphs}\label{sec:hyperbolic_random_graphs}

We now formally introduce the HRG model. We adopt the Poissonized version as in~\cite{kiwi2019second} and~\cite{mausruff2026distributed}, as opposed to a version with fixed number of vertices. See~\Cref{fig:both-figures} for an example.

\textbf{Vertex set.} Given $n \in \mathbb{N}, \alpha \in (1/2, 1), C \in \mathbb{R}$, we set $R = 2\log n + C$ and define the vertex set $V$ as the result of an inhomogeneous Poisson point process~\cite{kingman1992poisson} on the hyperbolic disk with radius $R$. The intensity of the process is non-zero only at radius $r \in [0, R)$ and angle $\theta \in [0, 2\pi)$ and is given by
\begin{align*}
    f(r, \theta) = \frac{1}{2\pi}e^{(R - C)/2}\frac{\alpha\sinh{\alpha r}}{\cosh{\alpha R} - 1}.
\end{align*}

The intensity of the Poisson point process above scales roughly exponentially with the radial coordinate, i.e.\ most points are near the boundary. Intuitively, the radial coordinate of a vertex can be thought of as its ``importance'' (and controls the degree) and the angular coordinate as its ``type''.

The expected number of points is exactly $n$. Note also that the numbers of vertices in disjoint regions are independent Poisson random variables. We will frequently identify a vertex with the pair $(r,\theta)$ of its radial and angular coordinates.

\textbf{Edge set.} There is an edge between two vertices $u = (r, \theta), v = (r', \theta')$ if and only if the hyperbolic distance $d$ between $u$ and $v$ is at most $R$. By the hyperbolic law of cosines, the distance satisfies: $\cosh{d} = \cosh{(r)}\cosh{(r')} - \sinh{(r)}\sinh{(r')}\cos{(\theta - \theta')}$.

\section{Roadmap of the analysis}

We give a high-level overview of the analysis and explain how the geometry of HRGs enables colouring with a near-optimal number of colours. The central idea is to exploit the strong relation between degree and radial position. Vertices close to the centre have very large degree and may lose too many available colours if lower-degree neighbours are coloured too early. Sequential Radial Colouring (SRC) therefore uses degree as priority order. Moreover, it delays colouring of lower-degree vertices by splitting the vertices into batches by degree and then processing those batches sequentially. More precisely, the vertices are grouped into radial pseudo-regions, which are processed from the centre outwards; within the currently active region, all uncoloured vertices run RCTDEG in parallel.

The decomposition is chosen so that, when a band becomes active, all relevant colouring obstructions have the right scale. The analysis uses three degree parameters, two of which are illustrated in~\Cref{fig:both-figures}. The \emph{previous degree} $d_{\mathrm{pre}}(v)$ counts neighbours in previous bands; these neighbours have already fixed their colours, and we pessimistically assume that they use distinct colours.\footnote{In particular, we effectively show that the algorithm solves the list-colouring problem.} The \emph{internal degree} $d_{\mathrm{int}}(v)$ counts neighbours inside the active band, which may further restrict the palette of $v$ during the active phase. The \emph{higher-priority degree} $\deg^+(v)$ counts neighbours of degree at least $\deg(v)$; it upper-bounds the number of neighbours that can block $v$ in one round of RCTDEG, and is bounded by 
\begin{align}\label{eq:previous-and-internal-degree}
    d_{\mathrm{pre}}(v) \le \deg^+(v) \le d_{\mathrm{pre}}(v)+d_{\mathrm{int}}(v),
\end{align}
regardless of our choice of bands. Note that $\max_{v\in V} \deg^+(v) \ge \kappa$, because the former is the maximum number of preceding neighbours in an ordering induced by degrees. We may assume worst-case tie breaking here, but the number of neighbours of the same degree as $v$ is negligibly small. Hence, there is at least one vertex for which the right hand side of~\eqref{eq:previous-and-internal-degree} is at least $\kappa$. This is why our approach is inherently limited to colour palettes larger than $\kappa$. However, we will show that we can choose the band widths small enough that the right hand side of~\eqref{eq:previous-and-internal-degree} is at most $(1+\eps)\kappa$ for all vertices $v\in V$. 

Below we give some intuition on how the three parts of Theorem~\ref{thm:summary} are proved. Throughout the analysis, there are two sources of randomness: the generation of the HRG and the random choices of the colouring algorithm. In our proofs, we typically first show structural pseudo-randomness results about \emph{typical} HRGs, meaning that a.a.s.\! such a graph satisfies the required structural conditions. Typical examples are the maximum and minimum values that $d_{\mathrm{int}}(v)$ and $d_{\mathrm{pre}}(v)$ take over all vertices in a band. We systematically collect such structural statements in \Cref{sec:preliminaries} and, for the two SRC regimes, at the beginnings of \Cref{sec:sequential-radial-colouring,sec:seq-col-constant-rounds}. In the second part, starting from~\Cref{subsec:colouring-pseudo-clique} we then analyse the algorithmic randomness using those pseudo-randomness conditions.

\subsection{Sequential Radial Colouring}

We subdivide the disk into radial bands whose relevant boundaries lie between a lower cutoff near the core and an upper cutoff near the boundary of the disk. In this range, degrees concentrate and the expected degree is strictly decreasing in the radius. Thus, although vertices only observe their degrees locally, they can estimate their radial positions by inverting the expected-degree function, recovering the true radius up to an additive $o(1)$ error; each pseudo-region therefore differs from its true radial counterpart only within an $o(1)$ boundary neighbourhood.  This allows the algorithm to implement a radial decomposition using only degree information.


\textbf{Band-decomposition for $\bm{|\Psi|\geq (1+\varepsilon)\kappa}$.} In the near-degeneracy regime $|\Psi|=(1+\varepsilon)\kappa$, SRC uses the finest decomposition (see~\Cref{fig:band-decomposition} for an illustration). After the pseudo-clique, the intermediate part consists of \(O(\log\log n)\) fine pseudo-bands of width \(1/\log\log n\) close to the core, followed by \(O(\log\log n)\) coarse pseudo-bands whose widths increase essentially geometrically as the distance from the centre increases. The widths are chosen so that
\begin{align*}
d_{\mathrm{pre}}(v)
\le
\begin{cases}
(1+o(1))\kappa, & \text{fine region},\\
\chi, & \text{coarse region},
\end{cases}
\qquad
d_{\mathrm{int}}(v)
\le
\begin{cases}
\chi/\sqrt{\log\log n}, & \text{fine region},\\
\eta\chi, & \text{coarse region},
\end{cases}
\end{align*}
for a sufficiently small constant $\eta>0$. Consequently, upon activation, every vertex $v$ has at least $C_{\mathrm{int}}M_{\mathrm{int}}(v)$ available colours\footnote{Equivalently, after colours used in previous regions have been removed, $|\Psi|-d_{\mathrm{pre}}(v)\ge C_{\mathrm{int}}M_{\mathrm{int}}(v)$.}, for some fixed constant $C_{\mathrm{int}}>2$, where
\begin{align}
M_{\mathrm{int}}(v)
:=
\max\left\{d_{\mathrm{int}}(v),\frac{\chi}{\sqrt{\log\log n}}\right\}.
\label{eq:active-band-slack}
\end{align}

The maximum with \(\chi/\sqrt{\log\log n}\) keeps this scale polynomially large, namely \(n^{\Omega(1)}\), even as the internal degrees decrease near the disk boundary. During the active phase, internal neighbours can remove at most $d_{\mathrm{int}}(v)\le M_{\mathrm{int}}(v)$ further colours, even if they all choose distinct colours. Hence, throughout the phase, each active vertex retains at least \((C_{\mathrm{int}}-1)M_{\mathrm{int}}(v)\) available colours. Since $C_{\mathrm{int}}>2$, this is at least $cM_{\mathrm{int}}(v)$ colours for some constant $c>1$. After the intermediate bands, the pseudo-outer region begins at a cutoff where the relevant degrees are already small enough compared with the palette to give high constant-round success probability.

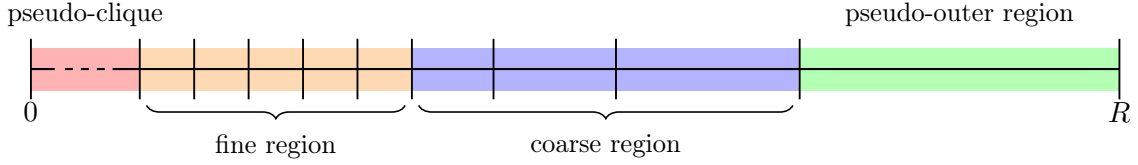
\begin{figure}[t]
    \centering
    \begin{tikzpicture}[
        x=0.9cm,
        y=1cm,
        bandtick/.style={line width=0.7pt},
        label/.style={font=\small},
        regionlabel/.style={font=\small}
    ]


        \def\axisend{16}

        \fill[red!30]
            (0,-0.28) rectangle (1.6,0.28);
        \fill[orange!30]
            (1.6,-0.28) rectangle (5.6,0.28);
        \fill[blue!30]
            (5.6,-0.28) rectangle (11.3,0.28);
        \fill[green!30]
            (11.3,-0.28) rectangle (\axisend,0.28);


        \draw[dashed,line width=0.8pt]
            (0.5,0) -- (1.2,0);

        \draw[-,line width=0.8pt]
            (1.2,0) -- (\axisend,0);

        \draw[-,line width=0.8pt]
            (0,0) -- (0.35,0);

        \node[anchor=north] at (0,-0.3)
            {$0$};
        \node[anchor=north] at (\axisend,-0.3)
            {$R$};

        \draw[bandtick] (0,-0.4) -- (0,0.4);
        \draw[bandtick] (1.6,-0.4) -- (1.6,0.4);

        \foreach \x in {2.4,3.2,4.0,4.8,5.6}
            \draw[bandtick] (\x,-0.4) -- (\x,0.4);

        \foreach \x in {6.8,8.6,11.3}
            \draw[bandtick] (\x,-0.4) -- (\x,0.4);

        \draw[bandtick]
            (\axisend,-0.4) -- (\axisend,0.4);

        \node[regionlabel,align=center] at (0.8,0.72)
            {pseudo-clique};

        \node[regionlabel,align=center] at (13.65,0.72)
            {pseudo-outer region};

        \draw[
            decorate,
            decoration={brace,mirror,amplitude=5pt},
            line width=0.6pt
        ]
            (1.7,-0.48) -- (5.5,-0.48)
            node[midway,below=7pt,label]
            {fine region};

        \draw[
            decorate,
            decoration={brace,mirror,amplitude=5pt},
            line width=0.6pt
        ]
            (5.7,-0.48) -- (11.2,-0.48)
            node[midway,below=7pt,label]
            {coarse region};

    \end{tikzpicture}
    \captionsetup{
        width=0.98\textwidth,
        justification=justified,
        singlelinecheck=false
    }
    \caption{The structure of the bands. Near the pseudo-clique, the fine pseudo-bands have equal width $O(1/\log\log n)$. Farther from the centre, the coarse pseudo-bands increase in width essentially geometrically until the pseudo-outer region begins.}
    \label{fig:band-decomposition}
\end{figure}

\textbf{Band-decomposition for $\bm{|\Psi|\ge \kappa^{1+\varepsilon}}$.} In the polynomial-slack regime $|\Psi|\ge \kappa^{1+\varepsilon}$, we use a
constant-band decomposition. The inner region consists of the core clique together with a small radial buffer of width \(c_{\mathrm{in}}\log\log n\), for a sufficiently large constant \(c_{\mathrm{in}}>0\). The remaining intermediate range up to the outer region is split into constantly many pseudo-bands, each of width\footnote{The last band might be truncated.} \(\varepsilon(1-\alpha)\alpha^{-1}\log n\). This width is large enough to have $O(1)$ intervals, but still small enough that the inner degree is small compared to the palette size.

For both decompositions the inner region is handled separately in the clique-colouring section, leaving the intermediate pseudo-bands as the main part of the SRC analysis. 

\textbf{Intermediate pseudo-bands for $\bm{|\Psi|\geq(1+\varepsilon)\kappa}$.} We analyse an intermediate pseudo-band through its \emph{internal uncoloured degree}. Let $D_0$ be the maximum internal uncoloured degree when the band becomes active, and for $i\ge1$ let $D_i$ be the corresponding maximum after round $i$ of the active phase. Initially, $D_0$ is bounded by the internal degree of the band. We show that, a.a.s., the internal uncoloured degree drops to $O(\log n)$ within $O(\log\log n)$ active rounds.

Condition on the set of uncoloured vertices at the beginning of round $i$. Recall that, throughout the active phase, every uncoloured vertex in the active pseudo-band has at least $c\cdot M_{\mathrm{int}}$ available colours, for some constant $c>1$. Moreover, it has at most $D_{i-1}$ uncoloured neighbours inside the active pseudo-band, and hence at most $D_{i-1}$ higher-priority neighbours that can block it. Therefore, the probability that a fixed vertex remains uncoloured in round $i$ is at most $D_{i-1}/(c \cdot M_{\mathrm{int}})$.

To bound the number of uncoloured neighbours of a fixed vertex, we sum the survival indicators of its neighbours in the active band. Each such indicator satisfies the one-vertex bound above, but the indicators are dependent. We overcome this by exposing the colour choices in degree-priority order and applying the sequential domination lemma\footnote{The aim is to apply a concentration bound to the sum of survival indicators. Since these indicators are dependent, this cannot be done directly; the stochastic domination step replaces them by independent Bernoulli variables to which Chernoff bounds apply.}: under any conditioning on the choices already exposed, the same probability bound remains valid, since it already assumes the worst case that higher-priority vertices forbid as many colours as possible. Thus, the number $U_i$ of uncoloured neighbours of any fixed vertex after round $i$ is stochastically dominated by a random variable $Z_i$ with
\begin{align*}
U_i
\preceq
Z_i 
\sim
\operatorname{Bin}\big(D_{i-1},\tfrac{D_{i-1}}{c\cdot M_{\mathrm{int}}}\big).
\end{align*}

The expectation of this dominating binomial variable is $(D_{i-1}^2)/(c\cdot M_{\mathrm{int}})$. Before using the quadratic recurrence, we first obtain a constant-factor decrease. Since $D_{i-1}\le M_{\mathrm{int}}$, this expectation is at most $D_{i-1}/c$ for $c>1$. A Chernoff bound and a union bound then give a fixed-factor decrease for $D_i$. Hence, after a constant number $t_0=O(1)$ of preliminary rounds, the maximum internal uncoloured degree satisfies $D_{t_0}\le\delta M_{\mathrm{int}}$ for a sufficiently small constant $\delta\in(0,1)$. An induction over the rounds $i\geq 0$ after the preliminary phase gives
\begin{align*}
D_{t_0+i}
\le
\delta M_{\mathrm{int}}\,2^{-(2^i-1)}.
\end{align*}



Thus, the maximum internal uncoloured degree drops doubly exponentially in the number of active rounds, and reaches $O(\log n)$ after $O(\log\log n)$ rounds. The domination by independent Bernoulli variables is what makes this recurrence hold with sufficiently high probability.

After that point, constantly many further rounds wrap up the proof. Indeed, in each such round the probability that a fixed remaining vertex stays uncoloured is already of order $\log n/n^{\Omega(1)}$. Choosing a sufficiently large constant number $q$ of additional rounds makes this probability $(\log n/n^{\Omega(1)})^q$, and a union bound over all vertices in the band then shows that a.a.s.\! all remaining vertices are coloured.

\textbf{Intermediate pseudo-bands for $\bm{|\Psi|\ge \kappa^{1+\varepsilon}}$.}
In this case we can only afford a constant number of pseudo-bands. Thus, each band is much larger, and the vertices inside of a single band are more inhomogeneous. Therefore, the bound~\eqref{eq:previous-and-internal-degree} is no longer strong enough. Instead, we take the prioritisation by degree into account and bound the \emph{higher-priority degree} $\deg^+(v)$ directly (we really mean in the entire graph, not just the band). For every vertex outside the inner region, \(\deg^+(u)=o(n^{1-\alpha})\) w.e.h.p. Moreover, the pseudo-band widths ensure that, when a band becomes active, colours forbidden by previous and internal neighbours amount to only \(o(|\Psi|)\). Thus, every active vertex retains \((1-o(1))|\Psi|\) available colours throughout the phase.

Conditioned on the state at the start of the round, an active vertex \(u\) remains uncoloured only if a higher-priority neighbour chooses the same colour. Since \(|\Psi|\ge\kappa^{1+\varepsilon}=\Theta(n^{(1-\alpha)(1+\varepsilon)})\),
\begin{align*}
\Pr(u\text{ remains uncoloured in given round})
\le
\frac{\deg^+(u)}{(1-o(1))|\Psi|}
=
o(n^{-\varepsilon(1-\alpha)}).
\end{align*}

Since the same conditional estimate holds in every round, choosing a sufficiently large constant number of rounds allows a union bound over all vertices in the band, and hence the whole band is coloured a.a.s. after a constant number of rounds.

\textbf{Pseudo-outer band in both SRC regimes.} For the near-degeneracy palette, the outer cutoff ensures degree $O(\log^2 n)$ in the pseudo-outer region. Since $|\Psi|\ge\kappa=\Theta(n^{1-\alpha})$, each active vertex has $(1-o(1))|\Psi|$ available colours, and its one-round failure probability is at most $O(\log^2 n/n^{1-\alpha})$. After constantly many rounds, a union bound over all vertices shows that all vertices in the pseudo-outer region are coloured a.a.s. For the polynomial-slack palette, the outer boundary ensures degree \(O(n^{1-\alpha})\), which is negligible compared with \(|\Psi|\ge \kappa^{1+\varepsilon}\). Thus, each vertex again keeps $(1-o(1))|\Psi|$ available colours, its one-round failure probability is $O(n^{-\varepsilon(1-\alpha)+o(1)})$, and any fixed number $t>1/(\varepsilon(1-\alpha))$ of rounds suffices by a union bound.

\subsection{Parallel Radial Colouring}

The second algorithm, Parallel Radial Colouring, trades number of colours used for runtime by assigning disjoint palettes to the pseudo-regions. Once the palettes are disjoint, colours chosen in one pseudo-region do not restrict colour choices in another region, so the previous-degree obstruction disappears and the regions can be coloured simultaneously. The local analysis is unchanged: the pseudo-clique and the intermediate pseudo-bands each require $O(\log\log n)$ rounds, while the pseudo-outer region requires only constantly many rounds. Thus the total runtime is $O(\log\log n)$ a.a.s., at the price of an additional $O(\log\log n)$ factor in the number of colours.

\subsection{Clique colouring}

With the band analysis out of the way, the remaining ingredient is the colouring of the pseudo-clique. We isolate the corresponding process as RCTDEG on a clique and prove matching upper and lower bounds for the two regimes: palettes of size \((1+1/\log n)|\mathcal C|\) result in \(\Theta(\log\log n)\) rounds,\footnote{More generally, if a clique \(\mathcal C\) is coloured with \(|\mathcal C|+s\) colours, where \(1\le s=o(n)\), then RCTDEG finishes in \(O(\log(|\mathcal C|/s)+\log\log n)\) rounds a.a.s.} while constant-round colouring is possible if and only if the multiplicative slack is of order \(|\mathcal C|^{\Omega(1)}\).

\textbf{Clique upper bounds.} The two upper bounds are proven for an arbitrary graph on $n$ vertices, which is why they can be invoked for the pseudo-clique. The proof closely follows the intermediate-band analysis for SRC with a near-degeneracy palette. Instead of tracking the internal uncoloured degree $D_i$, we track the number $U_i$ of uncoloured vertices remaining in the clique after round $i$. If the clique has size $n$ and the palette has additive slack $s_n$, then, at the beginning of round $i$, every uncoloured vertex has $s_n+U_{i-1}$ available colours. Conditioned on the state at the start of the round, a vertex with $j$ higher-priority uncoloured neighbours therefore remains uncoloured with probability at most $j/(s_n+U_{i-1})$. Summing over the priority order, the expected number of vertices surviving round $i$ is thus of order
\vspace{-0.7em}
\begin{align*}
\sum\nolimits_{j=0}^{U_{i-1}-1}\frac{j}{s_n+U_{i-1}}
\leq
\frac{U_{i-1}^2}{2(s_n+U_{i-1})}.
\end{align*}

Exposing choices in the priority order and applying a Chernoff bound then gives the following evolution. As long as $U_{i-1}$ is larger than a fixed fraction of the slack $s_n$, the expectation $U_{i-1}^2/(2(s_n+U_{i-1}))\leq U_{i-1}/2$ is still linear in $U_{i-1}$, so each round reduces the number of uncoloured vertices by a constant factor. This gives a geometric phase of length $O(\log(n/s_n))$, until $U_i=O(s_n)$. From then on the denominator is of order $s_n$, and the recurrence becomes quadratic, with expectation of order $U_{i-1}^2/s_n$. This drives $U_i$ down doubly exponentially, such that after $O(\log\log(n))$ further rounds $U_i=O(\log(n))$. A final pair-collision argument colours the remaining $O(\log n)$ vertices. Hence $(1+o(1))n$ colours suffice in $O(\log(n/s_n)+\log\log n)$ rounds.

For the polynomial-slack regime we use a different general upper bound. More precisely, for any graph on $n$ vertices, if $|\Psi|=\lceil n^{1+\varepsilon_n}\rceil$ with $\varepsilon_n>0$ and $n^{\varepsilon_n}\to\infty$, then RCTDEG finishes in $(1+o(1))/\varepsilon_n$ rounds a.a.s. Indeed, every uncoloured vertex has at least $|\Psi|-n\ge n(n^{\varepsilon_n}-1)$ available colours in every round, so the probability that a fixed vertex remains uncoloured for $t$ rounds is at most $(n^{\varepsilon_n}-1)^{-t}$. A union bound over all vertices gives completion after $(1+o(1))/\varepsilon_n$ rounds.

\textbf{Clique lower bounds.} The lower bounds show that the clique analysis is essentially tight. For palettes of size $\Theta(n)$, RCTDEG needs $\Omega(\log\log n)$ rounds a.a.s.; moreover, if $|\Psi|=n^{1+o(1)}$, then for every fixed number of rounds, $\omega(1)$ vertices remain uncoloured a.a.s. Thus a polynomial multiplicative slack is necessary for constant-round colouring.

The main idea is to prove a recursive lower bound on the number $U_i$ of uncoloured vertices after round $i$. If $|\Psi|=\Theta(nf(n))$, then the expected number of surviving vertices in the next round has the same quadratic scale as in the upper bound, namely roughly $U_i^2/|\Psi|$. Iterating this suggests a lower bound of order $n/(Df(n))^{2^i-1}$ for a suitable constant $D>0$. The difficulty is making this recursion hold with high probability, as there is no independence across vertices.


To overcome this dependence, we prove a lower-tail concentration bound for a single clique round. After conditioning on the previous history, the only remaining randomness comes from the independent colour choices made in that round. The number of vertices that remain uncoloured can then be expressed as a function of these choices by summing the colour excesses: if a colour is chosen by $m$ vertices, it contributes $m-1$ surviving vertices. This function is self-bounding. Changing the colour chosen by one vertex changes the number of surviving vertices by at most one. A concentration inequality for self-bounding functions therefore yields a Chernoff-type lower-tail estimate without requiring the individual survival events to be independent. We then iterate this estimate to obtain a general recursive lower bound on $U_i$. The two lower-bound corollaries follow by applying the recursion with $f(n)=\Theta(1)$ and with $f(n)=n^{o(1)}$ respectively.

\section{Discussion and open problems}\label{sec:discussion}

We have shown that, for every \(\alpha\in(1/2,1)\), HRGs can be coloured with \(((4/3)^\alpha+\eps)\chi\) colours in \(O((\log\log n)^2)\) rounds a.a.s. in the CONGEST model. 

Unlike for worst-case graph instances, the bottleneck for the number of colours is not the maximum degree, but rather the degeneracy of the graph. In the case of HRG, the degeneracy agrees up to lower-order terms with the parameter $\kappa_{\mathrm{deg}}$, defined as the maximum number of preceding neighbours of any vertex in the descending sorting of vertices by degree. We believe that this is an interesting parameter, and that it should be explored further which classes of graphs can efficiently be coloured in the CONGEST model with $O(\kappa_{\text{deg}})$ or even $(1+o(1))\kappa_{\text{deg}}$ colours. For models of real networks, it would be interesting to see whether the constraints of HRG can be relaxed. For example, can the results be generalised to all graphs in which the degree of a random vertex follows a power-law, possibly with further conditions?

Our algorithm proceeds in phases, where in each phase only vertices within a certain range of degrees are active. It is open whether the algorithm can be simplified further to work without phases. In particular, consider the very simple algorithm in which each vertex tentatively decides for a colour. If in any round a vertex has a conflict with a higher-degree neighbour (ties broken by random IDs), it redraws its colour, avoiding all colours currently used by higher-degree neighbours. Note that this decision may induce conflicts with lower-degree neighbours, which will have to resolve the conflict in the next round. This algorithm is guaranteed to find a proper colouring for any graph if at least $\kappa_{\text{deg}}+1$ colours are available, because after $i$ rounds, the $i$ highest-degree vertices are conflict-free and have thus terminated.\footnote{We assume here that the ordering defining $\kappa_{\text{deg}}$ uses the same tie-breaking IDs as the algorithm.} But it is open whether this algorithm is efficient for $O(\kappa_{\text{deg}})$ colours in HRG or on other classes of graphs.

\section{Further related work on distributed colouring}\label{sec:related_work_colouring}
We give here a selection of results for distributed colouring. For a more complete and general recent survey, see~\cite{grebik2025descriptive}. 

\paragraph{Details on the results by \texorpdfstring{Maus and Ruff~\cite{mausruff2026distributed}}{Maus and Ruff}.} We have already briefly summarized the important results from that paper, but give more details here. Recall that in the RCT algorithm, in case of a conflict both vertices discard their colour, while in RCTID conflicts are resolved by ID and in RCTDEG conflicts are resolved by degree, with the higher degree having higher priority. Maus and Ruff show that RCT succeeds after 2 rounds of communication if given $\eps \Delta$ colours for any constant $\eps > 0$.\footnote{The 2 rounds are actually necessary for a single step of the protocol, i.e.\ one round to communicate the attempted colour and another to inform neighbours whether the colour was permanently adopted or not. Vertices which did not manage to secure a colour in the first round simply sample a new one from the set of now available colours and do not even communicate it to their neighbours; with high probability, it will not be the same as theirs.}  They also show that RCT fails if only given $n^{\delta}\chi$ colours for some fixed constant $\delta$, i.e.\ the protocol is unable to approach $\chi$ colours. Interestingly, the same results are shown for RCTID, with the difference that the failure surfaces even with $\Delta \log^{-3}n$ colours. The underlying reason is that the neighbourhood of a large-degree vertex may exhaust its colour space, because the random-ID ordering places a constant fraction of its neighbours ahead of it.


The study of RCTDEG is motivated by exactly this failure. For this algorithm, it is shown that a constant number of rounds suffices if $\Delta^{1 - \delta_1 + o(1)}$ colours are allowed, but that the algorithm never terminates if instead there are $\Delta^{1 - \delta_2 - o(1)}$ colours. The constants $\delta_1$ and $\delta_2$ given in~\cite{mausruff2026distributed} match if $\alpha \le 3/4$, otherwise they don't match. This corresponds, up to lower order terms, to the cases $\Delta \ge n^{2/3}$ and $\Delta<n^{2/3}$, respectively. In any case, the number of colours necessary for the algorithm to succeed is at least $n^{1/6}\chi$.\footnote{In particular their algorithm requires $n^{1/(2 \alpha + 1/2)}$ colours if $\alpha < 3/4$ and $n^{1/(4\alpha - 1)}$ colours otherwise. The deviation from $\chi =\Theta(n^{1-\alpha})$ is minimized when $\alpha$ approaches $1/2$, in which case $\chi \approx n^{1/2}$ and their palette has size at least $n^{2/3} = n^{1/2 + 1/6}$.} For a larger number of colours, they show that even 2 rounds of the CONGEST protocol suffice.


\paragraph{Colouring with approximately \texorpdfstring{$\chi$}{χ} colours.}
The number of colours used to colour HRG in~\cite{mausruff2026distributed} is smaller than $\Delta$, but still a polynomial factor larger than $\chi$. Our algorithm can colour HRGs with $O(\chi n^{\eps})$ colours in constantly many rounds in the CONGEST model, for any $\eps > 0$.

The latter is already known to be possible for general graphs, but only in the (randomized) LOCAL model~\cite{barenboim2018fast}. Again for general graphs and constant $c, \chi$, the round complexity of $c$-colouring a $\chi$-colourable graph is $\Tilde{\Theta}(n^{1/\rho})$,\footnote{The tilde in our notation hides polylogarithmic factors \emph{in the argument}. So for example $\Tilde{O}(\log n)$ would mean $O(\log n (\log \log n)^c)$ for some constant $c$.} where $\rho = \lfloor \frac{c - 1}{\chi - 1} \rfloor$, as shown in~\cite{coiteux2024no}. In contrast, our algorithm only requires $O((\log \log n)^2)$ rounds to colour HRGs with $O(\chi)$ colours. An approximation ratio of $O(\frac{\log n}{\log \log n})$ can be guaranteed in polylogarithmically many rounds for general graphs using low diameter decompositions~\cite{linial1993low, barenboim2018fast, balliu2021improved}. The above general graph results only work in the LOCAL model, whereas our algorithm works in the CONGEST model. Moreover, our algorithm uses efficient local computation, rather than relying on the unbounded local computation permitted by the LOCAL and CONGEST models.

Results regarding colouring with roughly $\chi$ colours in the CONGEST model for other specific graph classes exist. In~\cite{halldorsson2014distributed}, the authors show how to colour interval graphs with $O(\chi)$ colours in an optimal $O(\log^* n)$ rounds.\footnote{The number of times one needs to iteratively take the logarithm of $n$ until we reach a number at most $1$ is $\log^*n$.} Later, this was improved to $(1 + \eps)\chi$ colours~\cite{halldorsson2020improved}, with an optimal dependency on $1/\eps$. These two algorithms only work in the CONGEST model if the vertices are aware of their interval representations.

Coming back to the LOCAL model, a similar $(1+\eps)$-approximation in $O(\eps^{-1} \log n)$ rounds was later achieved for chordal graphs~\cite{konrad2022distributed}. Related results for (re)colouring of interval and chordal graphs are derived in~\cite{bousquet2021distributed}. Another special case that has been studied is the one of unit disk graphs.\footnote{Unit disk graphs are graphs where each vertex is a point in Euclidean space and there is an edge if two vertices are within unit distance of each other. The results are stated in terms of the clique number \(\omega\), which is within a constant factor of \(\chi\) for unit disk graphs.}
There exists a location-oblivious distributed 3-approximation in an older distributed model~\cite{couture2007location}. Later, a constant round $4 \omega$-colouring algorithm was given in the location-aware setting, along with a different $(5.68 \omega + 1)$-colouring algorithm that works in $O(\log^* n)$ rounds which does not require location knowledge~\cite{esperet2023distributed}.

There exists a deterministic LOCAL algorithm that colours graphs of arboricity \(t\) with \(O(t^{1+\varepsilon})\) colours in \(O(\log t \log n)\) rounds~\cite{barenboim2010deterministic}. Since arboricity is within a factor two of degeneracy for every graph, and HRGs satisfy \(\kappa=\Theta(\chi)\) w.e.h.p.~\cite{baguley2025hrg}, this implies an \(O(\log^2 n)\)-round algorithm that colours an HRG with \(O(\chi^{1+\varepsilon})\) colours.

\paragraph{Colouring with \texorpdfstring{$\Delta$}{Δ} or \texorpdfstring{$\Delta + 1$}{Δ + 1} colours.} 
There have been many works studying the problems of distributed colouring with $\Delta$ (the maximum degree) or $\Delta + 1$ colours, with the focus being mainly on the latter. A $(\Delta + 1)$-colouring always exists and a $\Delta$-colouring exists for all graphs except cliques and odd cycles by Brooks' theorem. The current state of the art for deterministic LOCAL $\Delta$-colouring is a $\Tilde{O}(\log^{5/3} n)$-round algorithm employing a reduction to Maximal Independent Set~\cite{bourreau2026faster}. For constant $\Delta$, the bound becomes a tight $\Theta(\log n)$. For the randomized LOCAL model, the former paper shows that $\Tilde{O}((\log \log n)^{5/3})$ rounds suffice~\cite{bourreau2026faster}. Remarkably, the current bounds coincide with those for $(\Delta + 1)$-colouring, which had been established in earlier work~\cite{ghaffari2024near}. 

For the case of the CONGEST model, the state of the art for \((\Delta+1)\)-colouring without assumptions on \(\Delta\) is an \(O(\log^3 n)\)-round deterministic algorithm~\cite{ghaffari2022deterministic} and an \(O((\log\log n)^3)\)-round randomized algorithm~\cite{flin2023coloring}. For $\Delta$-colouring, there exists a randomized CONGEST algorithm~\cite{halldorsson2024distributed} which uses a number of rounds polynomial in $\log \log n$.

\paragraph{Colouring with \texorpdfstring{$\Delta - \sqrt{\Delta}$}{Δ - √Δ} colours.} 
An interesting body of works concerns itself with a somewhat intermediate regime between the previous two. Let $k_{\Delta}$ be the largest integer $k$ such that $(k + 1)(k + 2) \le \Delta$ and assume for the following that $\Delta$ is at least a large enough constant. Outside of distributed computing, deciding whether a graph $G$ is $c$-colourable is tractable (even linear time algorithms exist) if $c \ge \Delta - k_{\Delta}$ but NP-hard otherwise~\cite{molloy2014colouring}. Recently~\cite{bamas2018distributed, flin2026sublogarithmic}, a similar characterisation has been obtained for the round complexity of the distributed version of this problem. If $c \ge \Delta - k_\Delta + 1$,\footnote{Note that this threshold is by 1 larger than for NP-hardness.} then a randomized LOCAL algorithm can solve it in $\Tilde{O}((\log \log n)^4)$ rounds. When $\Delta \ge \log^{50}n$, the bound is improved to $O(\log^* n)$. On the other hand, for smaller $c$, at least $\Omega(n/\Delta)$ rounds are required.

\section{Preliminaries}\label{sec:preliminaries}
\paragraph{Notation and terminology.} 
The hyperbolic disk of radius $R$ is denoted by $\mathcal D_R$. We write $G \sim \mathcal{G}(n, \alpha, C)$ for a hyperbolic random graph $G$ sampled with parameters $n, \alpha$ and $C$. We use $B_x(r)$ to denote the ball of (hyperbolic) radius $r$ centred at point $x$. By $\mu(A)$ we denote the probability mass of a geometric region $A$ (almost always a ball or combination thereof) under the normalised vertex-position distribution. Thus the number of vertices in $A$ is Poisson with mean $n\mu(A)$. We use $d(r,s,\varphi)$ to denote the hyperbolic distance of two points at radial coordinates $r$ and $s$ which are separated by an angle of $\varphi$. We also define $\theta_{R}(r, s) :=\max\{\varphi\in[0,\pi]:d(r,s,\varphi)\le R\}$, i.e.\ the maximum angle at which two such vertices connect. We write \(\mathcal C:=V(G)\cap B_0(R/2)\) for the core clique and \(\sigma=\sigma(G):=|\mathcal C|\) for its size. Further, we write $\chi=\chi(G)$ for the chromatic number, and $\kappa=\kappa(G)$ for the degeneracy.

We use \emph{asymptotically almost surely} (a.a.s.) to denote probability $1-o(1)$, \emph{with high probability} (w.h.p.) for probability $1-O(1/n)$, and \emph{with extremely high probability} (w.e.h.p.) for probability $1-n^{-\omega(1)}$. Throughout the paper, in conditional probabilities $\Pr[\mathcal A | \mathcal B]$ we omit stating the condition $\Pr[\mathcal B] > 0$ explicitly. Further, whenever we take a union bound over all vertices, we do so on the event \(|V(G)|\le 2n\), which holds with probability \(1-e^{-\Omega(n)}\) since \(|V(G)|\sim\operatorname{Po}(n)\).

We now collect several auxiliary estimates used in the analysis, some of which are proved in existing works and some of which are proved in our appendix. We begin with the measure estimate for centred balls. This will be used throughout to estimate the number of vertices in radial regions.

\begin{lemma}[Measure of a centred ball {\cite[Equation~(2)]{baguley2025hrg}}]
\label{lem:centered-ball-measure}
Let $G\sim \mathcal G(n,\alpha,C)$ be an HRG with $\alpha\in(1/2,1)$. For every $0\le r\le R$, the measure of the centred ball $B_0(r)$ is
\begin{align*}
\mu(B_0(r))
=
\frac{\cosh(\alpha r)-1}{\cosh(\alpha R)-1}.
\end{align*}

In particular, if $r=r(n)=\omega(1)$, then $\mu(B_0(r))=(1+o(1))e^{-\alpha(R-r)}$. If instead $r=O(1)$, then $\mu(B_0(r))=O(e^{-\alpha (R-r)})$.
\end{lemma}

One of the properties of the Poisson point process is that the number of vertices in any given geometric area follows a Poisson distribution. In particular, if we fix the position of a vertex $u$, then its degree is also Poisson distributed, because it is given by the number of vertices that fall in the ball $B_u(R)$ of radius $R$ around $u$. The expectation depends on the radius of $u$. 



We denote this expectation by $\bar d(r) := \mathbb E[\deg(u)\mid r(u)=r]$. To estimate neighbourhoods restricted to radial regions, we use the angular connection threshold. For two vertices with radii $r$ and $s$ satisfying $r+s>R$, the largest angular separation that still gives an edge is $\theta_R(r, s)$, which by the hyperbolic law of cosines is $\theta_R(r, s)=\arccos\left((\cosh r\cosh s-\cosh R)/(\sinh r\sinh s)\right)$. If $r+s\le R$, then all angular positions are adjacent and we set $\theta_R(r, s):=\pi$.

\begin{lemma}[Angular threshold asymptotics {\cite[Lemma~6]{gugelmann2012random}}]
\label{lem:angular-threshold-asymptotics}
Let $G\sim \mathcal G(n,\alpha,C)$ be an HRG with $\alpha\in(1/2,1)$ and let $0\le r\le R$ and let $s> R-r$. Then, 
\begin{align*}
\theta_R(r, s)=2e^{(R-r-s)/2} \left(1+\Theta\left(e^{R-r-s}\right)\right).
\end{align*}
\end{lemma}

We prove a uniform expected-degree asymptotic on a radial range containing all boundaries used in the decomposition, keeping an explicit $o(1)$ term. Together with the explicit $o(1)$ multiplicative degree concentration on this range (Lemma~\ref{lem:degree-concentration-above-cutoff}), this yields the desired additive deviation between the true and estimated radii
(Lemma~\ref{lem:radius-estimate-accuracy}).

\begin{lemmaE}[Uniform expected-degree asymptotic][category=degree]
\label{lem:uniform-expected-degree-asymptotic}
Let $G\sim \mathcal G(n,\alpha,C)$ be an HRG with $\alpha\in(1/2,1)$ and let $a,b>0$ be fixed constants. Uniformly for all
$r\in[R/2-a,\,R-4\log\log n+b]$,
\begin{align*}
\bar d(r)
=
\left(\frac{2\alpha}{\pi(\alpha-1/2)}\pm o(\rho_n)\right)ne^{-r/2},
\end{align*}
where $\rho_n:=\left(48/(K\log n)\right)^{1/4}$ for any constant $K>0$\footnote{We keep \(K\) explicit because the later degree cutoff is chosen as \(K\log^2 n\), and the radius-estimation error is calibrated to this choice with an explicit constant.}.
\end{lemmaE}
\begin{proofE}
Consider a vertex $u$ with radius $r(u)=r\in[R/2-a,\,R-4\log\log n+b]$. Recall that $\bar d(r)$ denotes the expected degree of a vertex at radius $r$. Then $\bar d(r) = n\mu(B_u(R)\cap \mathcal D_R)$.

By rotational symmetry, we may assume that $u$ has angular coordinate $0$. For $r>0$, vertices with radius $s\le R-r$ are adjacent to $u$ for all angular coordinates, while for $s>R-r$ the admissible angular interval has length $2\theta_R(r,s)$. Hence, we have
\begin{align*}
\bar d(r)
&=
n\int_0^{R-r}
\frac{\alpha\sinh(\alpha s)}{\cosh(\alpha R)-1}\,ds
+
n\int_{R-r}^{R}
\frac{2\theta_R(r,s)}{2\pi}
\frac{\alpha\sinh(\alpha s)}{\cosh(\alpha R)-1}\,ds\\
&=
\frac{n}{\cosh(\alpha R)-1}
\left(
\cosh(\alpha(R-r))-1
+
\frac{\alpha}{\pi}
\int_{R-r}^{R}
\theta_R(r,s)\sinh(\alpha s)\,ds
\right).
\end{align*}

Recall that $\beta:=\alpha-\frac12>0$. We first show that the first term is negligible. Using $\cosh x=(e^x+e^{-x})/2$, we have $\cosh(\alpha R)-1=\Theta(e^{\alpha R})$ and $\cosh(\alpha(R-r))-1=O(e^{\alpha(R-r)})$. Hence
\begin{align*}
\frac{n(\cosh(\alpha(R-r))-1)}{\cosh(\alpha R)-1}
=
O(ne^{-\alpha r}) 
=
o(\rho_n ne^{-r/2}),
\end{align*}
where we used that $\alpha>1/2$, and $\rho_n^{-1}$ is polylogarithmic.

It remains to estimate the integral term. Let $m_n:=\lceil 3\log(1/\rho_n)\rceil$. We first remove the endpoint interval $[R-r,R-r+m_n]$. Since $\theta_R(r,s)\le \pi$ and $\sinh(\alpha s)=O(e^{\alpha s})$, its contribution to $\bar d(r)$ is at most
\begin{align*}
O\left(
n e^{-\alpha R}
\int_{R-r}^{R-r+m_n} e^{\alpha s}\,ds
\right)
&=
O\left(ne^{-\alpha r+\alpha m_n}\right)
=
o(\rho_n ne^{-r/2}),
\end{align*}
where we used that $\alpha>1/2$ and that $e^{\alpha m_n}$ is only polylogarithmic in $n$. Thus the endpoint interval contributes $o(\rho_n ne^{-r/2})$.

On the remaining interval $s\in[R-r+m_n,R]$, Lemma~\ref{lem:angular-threshold-asymptotics} gives
\begin{align*}
\theta_R(r,s)
=
2e^{(R-r-s)/2}\left(1+O(e^{R-r-s})\right)
=
2e^{(R-r-s)/2}(1+o(\rho_n)),
\end{align*} 
since $e^{R-r-s}\le e^{-m_n}=o(\rho_n)$. Further, since $\sinh(\alpha s)=\frac12e^{\alpha s}(1-e^{-2\alpha s})$ and $e^{-2\alpha s}=o(\rho_n)$ uniformly on this interval, we have $\sinh(\alpha s)=\frac12e^{\alpha s}(1-o(\rho_n))$.

Therefore,
\begin{align*}
\int_{R-r+m_n}^{R}
\theta_R(r,s)\sinh(\alpha s)\,ds
&=
(1+o(\rho_n))(1-o(\rho_n))
e^{(R-r)/2}
\int_{R-r+m_n}^{R} e^{\beta s}\,ds \\
&=
(1\pm o(\rho_n))
\frac{1}{\beta}
e^{(R-r)/2}e^{\beta R}.
\end{align*}

Here the last step used that $e^{\beta(R-r+m_n)}=o(\rho_n e^{\beta R})$. Finally, since 
\begin{align*}
\cosh(\alpha R)-1=\frac12 e^{\alpha R}
\left(1-2e^{-\alpha R}+e^{-2\alpha R}\right)
=\frac12 e^{\alpha R}(1-o(\rho_n)),
\end{align*}
the main integral contribution is, 
\begin{align*}
\frac{n}{\cosh(\alpha R)-1}
\frac{\alpha}{\pi}
\int_{R-r+m_n}^{R}
\theta_R(r,s)\sinh(\alpha s)\,ds
&=
\left(\frac{2\alpha}{\pi\beta}\pm o(\rho_n)\right)
ne^{-r/2}.
\end{align*}

The omitted endpoint interval and the first term are both $o(\rho_n ne^{-r/2})$. Since $\beta=\alpha-\frac12$, this proves the desired claim uniformly for all $r\in[R/2-a,\,R-4\log\log n+b]$.
\end{proofE}

The next lemma shows that $\bar d$ can be inverted uniquely on $(0,R]$.

\begin{lemmaE}[Expected degree monotonicity][category=degree]
\label{lem:expected-degree-monotone}
Let $G\sim \mathcal G(n,\alpha,C)$ be an HRG with $\alpha\in(1/2,1)$. The function $\bar d$ is strictly decreasing on $(0,R]$.
\end{lemmaE}
\begin{proofE}
Let $u$ be a vertex with radial coordinate $r \in (0,R]$. By rotational symmetry, we may assume that its angular coordinate is $0$. We denote the radial density by
\(\displaystyle
f(s) \coloneqq
\frac{\alpha \sinh(\alpha s)}{\cosh(\alpha R)-1}
\).
For a point of radius $s \in [0,R]$, let $p_R(r,s)$ be the probability that its angular coordinate lies in the connection interval of $u$. Then
\begin{align*}
p_R(r,s)
=
\begin{cases}
1, & s\le R-r,\\[1mm]
\displaystyle
\frac{1}{\pi}
\arccos\left(
\frac{\cosh r\cosh s-\cosh R}{\sinh r\sinh s}
\right),
& s>R-r.
\end{cases}
\end{align*}

We first show that, for every fixed $s \in [0,R]$, the function $r\mapsto p_R(r,s)$ is non-increasing. The claim is immediate on the part of the domain where $r+s\le R$. On the part where $r+s>R$, necessarily $s>0$, and we define 
\begin{align*}
q(r,s)
:=
\frac{\cosh r\cosh s-\cosh R}{\sinh r\sinh s}.
\end{align*}

A direct calculation gives
\begin{align*}
\frac{\partial q}{\partial r}
=
\frac{\cosh R\cosh r-\cosh s}{\sinh^2 r\,\sinh s}.
\end{align*}

Here the denominator is positive. Moreover, since $s\le R$ and $r>0$, $\cosh R\cosh r > \cosh R \ge \cosh s,$ which yields $\partial q/\partial r>0$. In the regime $r+s>R$, we also have $|r-s|<R<r+s$, and hence $-1<q(r,s)<1$. Since $\arccos$ is strictly decreasing on $(-1,1)$, it follows that $p_R(r,s)$ is strictly decreasing in $r$ whenever $r+s>R$.

Now let $0<r_1<r_2\le R$. By the preceding paragraph, $p_R(r_2,s)\le p_R(r_1,s)$ for all $s\in[0,R]$. Moreover, for every $s\in(R-r_2,R)$, the strict inequality holds. This interval has positive measure, and $f(s)>0$ on $(0,R)$. Therefore
\begin{align*}
\bar d(r_2)
=
n\int_0^R p_R(r_2,s)f(s)\,ds 
<
n\int_0^R p_R(r_1,s)f(s)\,ds
=
\bar d(r_1).
\end{align*}

Thus $\bar d$ is strictly decreasing on $(0,R]$.
\end{proofE}

If a vertex is placed at the origin, then its ball of radius $R$ coincides with the whole hyperbolic disk $\mathcal D_R$. Hence its expected degree is the expected number of other vertices in the disc, which is $n$ in the Poissonized model. We therefore set $\bar d(0)=n$. Together with Lemma~\ref{lem:expected-degree-monotone}, this gives a continuous strictly decreasing extension of $\bar d$ to $[0,R]$, whose image is $\bar d([0,R])=[\bar d(R),n]$. It remains to determine the lower endpoint of this image.

\begin{lemmaE}[Boundary expected degree][category=degree]
\label{lem:boundary-expected-degree}
Let $G\sim\mathcal G(n,\alpha,C)$ be an HRG with $\alpha\in(1/2,1)$. We have $\bar d(R)=\Theta(1)$.
\end{lemmaE}

\begin{proofE}
Let $u$ be a vertex at radius $R$. Since there is no positive radial range on which all angular coordinates are adjacent to $u$, the decomposition of the expected degree gives
\begin{align*}
\bar d(R)
=
\frac{n}{\cosh(\alpha R)-1}
\frac{\alpha}{\pi}
\int_0^R \theta_R(R,s)\sinh(\alpha s)\,ds.
\end{align*}

The part of the integral over $s\in[0,1]$ contributes at most
$O(ne^{-\alpha R})=o(1)$. On $s\in[1,R]$, the angular-threshold bounds from
Lemma~\ref{lem:angular-threshold-asymptotics} give
$\theta_R(R,s)=\Theta(e^{-s/2})$. Hence
\begin{align*}
\bar d(R)
=
\Theta\left(
ne^{-\alpha R}
\int_1^R e^{-s/2}e^{\alpha s}\,ds
\right)
=
\Theta\left(
ne^{-\alpha R}
\int_1^R e^{(\alpha-1/2)s}\,ds
\right)
=
\Theta(ne^{-R/2}).
\end{align*}

Since $R=2\log n+C$, we have $ne^{-R/2}=e^{-C/2}$, and therefore
$\bar d(R)=\Theta(1)$.
\end{proofE}

We write $\deg^+(u):=|\{v\in N(u):\deg(v)\ge \deg(u)\}|$ for the number of
neighbours of $u$ whose degree is at least $\deg(u)$.

\begin{lemma}[Larger degree neighbourhood {\cite[Lemma~13]{mausruff2026distributed}}]
\label{lem:larger-degree-neighbourhood}
Let $G\sim \mathcal G(n,\alpha,C)$ be an HRG with $\alpha\in(1/2,1)$, and let $u$ be a vertex with $r(u)\in (R-\ell-1, R-\ell]$. Then, w.e.h.p.,
\begin{align*}
\deg^+(u)
\in
\begin{cases}
O(e^{\ell/2}+\log n),
& 0\le \ell\le \left\lfloor \frac{2}{1-\alpha}\log\log n\right\rfloor,\\
O(e^{\ell(1-\alpha)}),
& \left\lceil \frac{2}{1-\alpha}\log\log n\right\rceil
\le \ell\le \left\lceil R/2\right\rceil,\\
O(ne^{-\alpha\ell}),
& \left\lceil R/2\right\rceil
\le \ell\le
\left\lceil \alpha^{-1}(\log n-2\log\log n)\right\rceil.
\end{cases}
\end{align*}
\end{lemma}

Recall that the previous degree of a vertex is the number of neighbours that lie in the already processed inner regions. Since the decomposition is radial, this quantity is controlled by the number of neighbours in centred balls. We first record the monotonicity of the corresponding expectation.

\begin{lemmaE}[Internal degree in centred balls][category=degree]
\label{lem:internal-ball-degree-monotone}
Let $G\sim \mathcal G(n,\alpha,C)$ be an HRG with $\alpha\in(1/2,1)$ and let $q\in(0,R]$. For $r\in[0,q]$ define $D_q(r):=\mathbb E[|N(u)\cap B_0(q)|\mid r(u)=r]$. Then $D_q(r)$ is non-increasing in $r$. 
\end{lemmaE}
\begin{proofE}
By rotational symmetry, assume that $u$ has angular coordinate $0$. Let $f(s)$ be the radial density, and let $p_R(r,s)$ be the probability that a vertex of radius $s$ is adjacent to $u$. Then
\begin{align*}
D_q(r)=n\int_0^q p_R(r,s)f(s)\,ds.
\end{align*}

As shown in the proof of Lemma~\ref{lem:expected-degree-monotone}, for every fixed $s\in[0,q]$ the function $r\mapsto p_R(r,s)$ is non-increasing. Hence, if $0\le r_1<r_2\le q$, then $p_R(r_2,s)\le p_R(r_1,s)$ for every $s\in[0,q]$. Integrating against the non-negative density $f(s)$ gives $D_q(r_2)\le D_q(r_1)$, as claimed.
\end{proofE}

This monotonicity is used in Lemma~\ref{lem:centred-ball-degeneracy-lower-bound} and, through that lemma, in Corollary~\ref{cor:previous-degree-degeneracy-refined}, where we bound the previous degree in the refined bands by $(1+o(1))\kappa$.

\begin{lemmaE}[Centred-ball lower bound on degeneracy][category=degree]
\label{lem:centred-ball-degeneracy-lower-bound}
Let $G\sim \mathcal G(n,\alpha,C)$ be an HRG with $\alpha\in(1/2,1)$ and let $q\geq R/2$. Suppose that $D_q(q)=\omega(\log n)$  (where $D_q(q)$ is as in~\Cref{lem:internal-ball-degree-monotone}) and consider the minimum degree $\delta(G[V(G) \cap B_0(q)])$ of the subgraph induced by $V(G) \cap B_0(q)$. Then, with high
probability,
\begin{align*}
\delta\left(G[V(G)\cap B_0(q)]\right)
\ge
(1-o(1))D_q(q).
\end{align*}

Consequently, $\kappa\ge (1-o(1))D_q(q)$ w.h.p.
\end{lemmaE}

\begin{proofE}
Let $G_q:=G[V(G)\cap B_0(q)]$. Choose $\xi_n=o(1)$ such that
$\xi_n^2D_q(q)\ge 4\log n$, which is possible since
$D_q(q)=\omega(\log n)$.

We use the following Poisson Chernoff bound \cite[Theorem~5.4]{mitzenmacher2005probabilit}: if $X$ is Poisson with mean $\mu$ and $x<\mu$, then $\Pr(X\le x)\le e^{-\mu}(e\mu)^x/x^x$. Conditional on the position of a vertex $u$ with $r(u)=r\le q$, the random variable $|N(u)\cap B_0(q)|$ is Poisson with mean $D_q(r)$. By Lemma~\ref{lem:internal-ball-degree-monotone}, $D_q(r)\ge D_q(q)$ for every $r\le q$. Therefore,
\begin{align*}
\Pr\left(
|N(u)\cap B_0(q)|<(1-\xi_n)D_q(q)\mid r(u)=r
\right)
&\le
\Pr\left(
|N(u)\cap B_0(q)|\leq(1-\xi_n)D_q(r)\mid r(u)=r
\right) \\
&\le
\exp\left(
-D_q(r)\left(\xi_n+(1-\xi_n)\log(1-\xi_n)\right)
\right) \\
&\le
\exp\left(-\frac{\xi_n^2D_q(r)}{2}\right)
\le
\exp\left(-\frac{\xi_n^2D_q(q)}{2}\right)
\le
n^{-2},
\end{align*}
where we used $\xi+(1-\xi)\log(1-\xi)\ge \frac{\xi^2}{2}$ for $\xi\in(0,1)$. On the standard event $|V(G)|=O(n)$, a union bound over all vertices gives that the probability that some vertex in $V(G)\cap B_0(q)$ has fewer than $(1-\xi_n)D_q(q)$ neighbours inside $B_0(q)$ is at most $O(n)\cdot n^{-2}=o(1)$. Hence w.h.p. no such vertex exists, and therefore $\delta(G_q)\ge(1-\xi_n)D_q(q)=(1-o(1))D_q(q)$.

Finally, degeneracy is the maximum minimum degree over all induced subgraphs. Since $G_q$ is an induced subgraph of $G$, we have
$\kappa
\ge
\delta(G_q)
\ge
(1-o(1))D_q(q)$, which proves the claim.
\end{proofE}

We use the core size $\sigma$ as a common reference scale for comparing degeneracy and chromatic number. The next two estimates imply in particular that w.e.h.p. $\kappa$ is only a constant factor larger than $\chi$.

\begin{lemma}[Bounds on the degeneracy {\cite[Theorem~9]{baguley2025hrg}}]
\label{lem:degeneracy-bounds}
Let $G\sim \mathcal G(n,\alpha,C)$ be an HRG with $\alpha\in(1/2,1)$. Then w.e.h.p. its degeneracy $\kappa$ satisfies
\begin{align*}
\left(
\frac{4-o(1)}{\pi}
\left(
\frac{2(1-\alpha)}
{\pi/2-\alpha(\pi-2)}
\right)^{\frac{1-\alpha}{2\alpha-1}}
\right)
\sigma
\le
\kappa
\le
\left(
\left(\frac{4}{3}\right)^\alpha+o(1)
\right)
\sigma.
\end{align*}
\end{lemma}

The corresponding chromatic-number estimate uses the same reference scale.

\begin{lemma}[Bounds on the chromatic number {\cite[Corollary~10]{baguley2025hrg}}]
\label{lem:chromatic-number-bound}
Let $G\sim \mathcal G(n,\alpha,C)$ be an HRG with $\alpha\in(1/2,1)$. Then w.e.h.p. the chromatic number $\chi$ satisfies
\begin{align*}
\sigma\le \chi\le \bigl((4/3)^\alpha+o(1)\bigr)\sigma.
\end{align*}

In particular, $\chi=\Theta(n^{1-\alpha})$ w.e.h.p.
\end{lemma}

Indeed, since $\sigma\le\chi$, Lemma~\ref{lem:degeneracy-bounds} gives $\kappa \le \bigl((4/3)^\alpha+o(1)\bigr)\sigma \le \bigl((4/3)^\alpha+o(1)\bigr)\chi$ w.e.h.p. Thus a palette of size proportional to $\kappa$ is also within a constant factor of the chromatic number, with factor at most $(4/3)^\alpha+o(1) \in (\sqrt{(4/3)}, 4/3)$.

\subsection{One-round domination tool}
\label{subsec:one-round-domination}

Consider one round of an active region, and let $\mathcal U$ be the set of uncoloured vertices at the beginning of this round, with $U:=|\mathcal U|$. Order the vertices of $\mathcal U$ by decreasing priority as $v_1,\dots,v_U$, and let $I_j$ be the indicator that $v_j$ remains uncoloured in this round. Suppose that $v_j$ has at least $a_j$ available colours and at most $b_j\le a_j$ higher-priority neighbours in $\mathcal U$. We have the following lemma, which will allow us to argue that the number of uncoloured vertices drops predictably over rounds of the algorithm.

\begin{lemmaE}[One-step failure bound][category=oneRoundDom]
\label{lem:one-step-failure-bound}
In the setting above, for every $j\in\{1,\dots,U\}$ and every $i_1,\dots,i_{j-1}\in\{0,1\}$, we have
\begin{align*}
\Pr\left(I_j=1\mid \mathcal U,I_1=i_1,\dots,I_{j-1}=i_{j-1}\right)
\le
\frac{b_j}{a_j}.
\end{align*}
\end{lemmaE}

\begin{proofE}
Fix $j\in\{1,\dots,U\}$ and a realisation $i_1,\dots,i_{j-1}$ of the indicators for the vertices $v_1,\dots,v_{j-1}$. Let $\Gamma_{<j}$ denote the vector of colours chosen\footnote{Here and below, $\Gamma_{<j}$ is understood only over realisations that are compatible with the prescribed values $I_1=i_1,\dots,I_{j-1}=i_{j-1}$.} in this round by $v_1,\dots,v_{j-1}$. We condition on $\mathcal U$, on $I_1=i_1,\dots,I_{j-1}=i_{j-1}$, and on the colour choices $\Gamma_{<j}$. Under this conditioning, only higher-priority neighbours of $v_j$ in $\mathcal U$ can prevent $v_j$ from keeping its colour. There are at most $b_j$ such higher-priority neighbours, and each forbids at most one colour. Since $v_j$ chooses uniformly from at least $a_j$ available colours, for every realization $\Gamma_{<j}$ we have
\begin{align*}
\Pr\left(I_j=1\mid
\mathcal U,I_1=i_1,\dots,I_{j-1}=i_{j-1},\Gamma_{<j}
\right)
\le
\frac{b_j}{a_j}.
\end{align*}

Taking total probability over all such realisations $\gamma$ gives
\begin{align*}
\Pr\left(I_j=1\mid \mathcal U,I_1=i_1,\dots,I_{j-1}=i_{j-1}\right)
&=
\sum_{\gamma}
\Pr\left(I_j=1\mid \mathcal U,I_1=i_1,\dots,I_{j-1}=i_{j-1},\Gamma_{<j}=\gamma\right) \\
&\qquad\cdot
\Pr\left(\Gamma_{<j}=\gamma\mid \mathcal U,I_1=i_1,\dots,I_{j-1}=i_{j-1}\right) \\
&\le
\frac{b_j}{a_j}.
\end{align*}
\end{proofE}

This allows us to couple the indicators of uncoloured vertices to independent Bernoulli random variables.




\begin{corollaryE}[Sequential stochastic domination for subsets][category=oneRoundDom]
\label{cor:sequential-stochastic-domination-sub}
In the setting above, let $S\subseteq\{1,\dots,U\}$. Let $(Z_j)_{j\in S}$ be conditionally independent Bernoulli random variables given $\mathcal U$, with $\Pr(Z_j=1\mid \mathcal U)=b_j/a_j$ for all $j \in S$. Then, conditional on $\mathcal U$,
\begin{align*}
\sum_{j\in S} I_j
\preceq
\sum_{j\in S} Z_j .
\end{align*}
\end{corollaryE}

\begin{proofE}
Write $S=\{s_1<\cdots<s_m\}$. We first show that, for every $\ell\in\{1,\dots,m\}$ and every realisation $i_1,\dots,i_{\ell-1}\in\{0,1\}$ satisfying $\Pr(I_{s_1}=i_1,\dots,I_{s_{\ell-1}}=i_{\ell-1}\mid\mathcal U)>0$, we have
\begin{align*}
\Pr\left(I_{s_\ell}=1\mid \mathcal U,I_{s_1}=i_1,\dots,I_{s_{\ell-1}}=i_{\ell-1}\right)
\le
\frac{b_{s_\ell}}{a_{s_\ell}} .
\end{align*}

Indeed, condition further on a full realisation of $I_1,\dots,I_{s_\ell-1}$ that is compatible with $I_{s_1}=i_1,\dots,I_{s_{\ell-1}}=i_{\ell-1}$. Lemma~\ref{lem:one-step-failure-bound} gives the same upper bound under each such full conditioning. Averaging over these full realisations preserves the bound.

Thus, conditional on $\mathcal U$, the sequence $I_{s_1},\dots,I_{s_m}$ satisfies the required sequential domination hypothesis. Hence, by \cite[Lemma~1.8.7]{doerr2020probabilistic}, we get the desired claim.
\end{proofE}

\section{Sequential Radial Colouring with \texorpdfstring{$\frac{4}{3}\chi$}{4/3χ(G)} colours}\label{sec:sequential-radial-colouring}

In this section we analyse Sequential Radial Colouring with $(1+\varepsilon)\kappa$ many colours. The algorithm is based on recovering the radial structure of the hyperbolic random graph from vertex degrees: the degrees induce approximate radial bands, which are activated sequentially from the core outwards. We define this degree-based decomposition and the resulting colouring process in \Cref{subsec:sequential-radial-colouring}. In \Cref{subsec:pseudo-band-size-estimates}, we show that the approximate regions are close to the corresponding true radial regions. 

The colouring analysis then proceeds through the three types of regions: \Cref{subsec:colouring-pseudo-clique} analyses the pseudo-clique, \Cref{subsec:colouring-pseudo-bands} analyses the intermediate pseudo-bands, and \Cref{subsec:colouring-pseudo-outer-region} analyses the pseudo-outer band. Combining these three colouring guarantees proves that SRC colours $G$ a.a.s. in $O((\log\log n)^2)$ rounds with a palette giving a constant-factor approximation\footnote{Where in the worst-case this constant factor is $\frac{4}{3}$.} of $\chi$ and a $(1+\varepsilon)$-approximation of $\kappa$ (\Cref{subsec:proof-near-chromatic-theorem}).

\subsection{Sequential Radial Colouring algorithm}
\label{subsec:sequential-radial-colouring}

The input is a graph $G\sim\mathcal G(n,\alpha,C)$ with $\alpha\in(1/2,1)$, together with a colour palette $\Psi$. The algorithm first uses the graph structure to recover an approximate radial decomposition. Each vertex estimates its radial coordinate from its degree, and these estimates determine whether the vertex belongs to the pseudo-clique, one of the intermediate pseudo-bands, or the pseudo-outer region.\footnote{The exact decomposition depends on the colour palette size.} The colouring is then carried out on these pseudo-regions from the centre outwards, using random colour trials with degree-priority conflict resolution (i.e.\ exactly RCTDEG as in~\cite{mausruff2026distributed}) inside each currently active region. That is, in the beginning only the vertices of the pseudo-clique attempt to colour themselves for a fixed time period (of order $O(\log \log n)$), then the next time period only vertices in the first pseudo-band are active and so on.

\subsubsection*{Radius estimation}
\label{subsubsec:radius-estimation}

To recover the radial structure of the hyperbolic random graph from vertex
degrees, Sequential Radial Colouring uses the expected-degree function. It turns
each observed degree into an estimated radial position. For a vertex $u$ with radius $r(u)=r$, let
$\bar d(r):=\mathbb E[\deg(u)\mid r(u)=r]$ denote the expected degree of a
vertex at radius $r$. Recall that at the origin, we have $\bar d(0)=n$, because a vertex placed at the origin is adjacent to every vertex in the disc, whose expected number is $n$. Since $\bar d(r)$ is strictly decreasing in $r$ on $(0,R]$ by
\Cref{lem:expected-degree-monotone}, this defines an inverse on the image
$[\bar d(R),n]$. 
To assign an estimated radius to every possible
observed degree\footnote{We assume that the vertices know $n$ and $\alpha$, which is possible to estimated.}, we extend the inverse to all of $[0,\infty)$ by setting
\begin{align*}
\bar d_{\mathrm{ext}}^{-1}(x)
:=
\begin{cases}
R, & x\le \bar d(R),\\[1mm]
\bar d^{-1}(x), & \bar d(R)<x<n,\\[1mm]
0, & n \leq x.
\end{cases}
\end{align*}

For every vertex $v$, we define the estimated radius
$\widehat r(v):=\bar d_{\mathrm{ext}}^{-1}(\deg(v))$ and the estimated offset
from the core boundary $\widehat x(v):=\widehat r(v)-R/2$. The estimated offsets now induce the radial decomposition used by the algorithm.

\subsubsection*{Pseudo-regions}
\label{subsubsec:pseudo-regions}

The algorithm assigns vertices to radial regions using the estimated offsets. We first specify a radial interval $[r_{\mathrm L},r_{\mathrm U}]$, with lower cutoff radius $r_{\mathrm L}$ and upper cutoff radius $r_{\mathrm U}$, chosen so that all relevant pseudo-region boundaries lie inside this interval. In this range, vertex degrees are highly concentrated around their expectations (Lemma~\ref{lem:degree-concentration-above-cutoff}), which implies that the estimated radius differs from the true radius by only $o(1)$ (Lemma~\ref{lem:radius-estimate-accuracy}). Thus, vertices can cross a pseudo-region boundary only if their true radius lies within this small error window around the boundary. 

The lower cutoff is $r_{\mathrm L}:=R/2-1$, giving a constant buffer before the pseudo-clique boundary $R/2$. For the upper cutoff, fix a sufficiently large constant $K>0$; the next lemma defines a radius $r_{\mathrm U}$ beyond which all vertices have degree at most $K\log^2 n$ w.h.p.

\begin{lemma}[Upper degree cutoff]
\label{lem:right-degree-cutoff}
For every sufficiently large constant $K>0$, there exists a radius
$r_{\mathrm U}=R-4\log\log n+O(1)$ such that, w.h.p., every vertex $v$ with
$r(v)\ge r_{\mathrm U}$ satisfies $\deg(v)\le K\log^2 n$.
\end{lemma}

\begin{proof}
Choose $r_{\mathrm U}$ as the unique radius satisfying $\bar d(r_{\mathrm U})=(K/4)\log^2 n$. The radius $r_{\mathrm U}$ is well-defined for all sufficiently large $n$. Indeed, $\bar d(R)=\Theta(1)$ by Lemma~\ref{lem:boundary-expected-degree}, whereas Lemma~\ref{lem:uniform-expected-degree-asymptotic} gives $\bar d(R/2)=\Theta(ne^{-R/4})=\Theta(n^{1/2})$. Since $\bar d$ is continuous and strictly decreasing on $[R/2,R]$ by Lemma~\ref{lem:expected-degree-monotone}, there is a unique $r_{\mathrm U}\in(R/2,R)$ with $\bar d(r_{\mathrm U})=(K/4)\log^2 n$.

By Lemma~\ref{lem:uniform-expected-degree-asymptotic}, for every fixed constant $c_0$,
\begin{align*}
\bar d(R-4\log\log n+c_0)
=
\left(\frac{2\alpha}{\pi(\alpha-1/2)}e^{-C/2}e^{-c_0/2}+o(1)\right)\log^2 n.
\end{align*}

Choosing fixed constants $c_1<c_2$ so that the leading constant is respectively larger and smaller than $K/4$, monotonicity implies $r_{\mathrm U}\in[R-4\log\log n+c_1,\,R-4\log\log n+c_2]$. Hence $r_{\mathrm U}=R-4\log\log n+O(1)$. It remains to prove the degree bound. Conditional on the position of a vertex $v$ with $r(v)\ge r_{\mathrm U}$, the degree of $v$ is Poisson with mean $\bar d(r(v))\le\bar d(r_{\mathrm U})=(K/4)\log^2 n$. Therefore, by a Chernoff bound, there is a constant $c_3>0$ such that
\begin{align*}
\Pr\left(\deg(v)>K\log^2 n\mid r(v)\ge r_{\mathrm U}\right)
\le
\exp(-c_3K\log^2 n).
\end{align*}

A union bound over all vertices gives the claim.
\end{proof}

The innermost pseudo-region is the pseudo-clique $\widehat{\mathcal C}:=\{v\in V(G):\widehat r(v)\le R/2\}$, the set of vertices whose estimated radius is at most $R/2$. This boundary lies a constant distance above the lower cutoff $r_{\mathrm L}=R/2-1$. Hence the $o(1)$ error in the radius estimate cannot move vertices from below the lower cutoff across the pseudo-clique boundary.

We next choose a provisional boundary for the pseudo-outer region. Fix a constant $K'>K$, and let $r_{\mathrm{out}}'$ be the unique radius satisfying $\bar d(r_{\mathrm{out}}')=K'\log^2 n$. Since $K'>K/4$ and $\bar d$ is strictly decreasing, we have $r_{\mathrm{out}}'<r_{\mathrm U}$. The same calculation as in the proof of Lemma~\ref{lem:right-degree-cutoff} gives $r_{\mathrm{out}}'=R-4\log\log n+O(1)$. Thus, this provisional outer radius lies below the upper cutoff radius $r_{\mathrm U}$. Hence the $o(1)$ error in the radius estimate cannot move vertices from beyond the upper cutoff across the provisional outer boundary. Later, the final outer boundary may be shifted by at most a constant amount to avoid a last intermediate band that is too thin.

We next define the intermediate pseudo-bands between the pseudo-clique and the provisional pseudo-outer boundary. The decomposition has two parts. Close to the core, we use fine bands of width $1/\log\log n$ in order to keep the internal degree inside each active band small. After the constant offset $x_{\mathrm{crit}}$, we use coarse bands whose offsets, and hence also widths, grow essentially geometrically. We choose $x_{\mathrm{crit}}>0$ large enough so that the previous-degree bound in Lemma~\ref{lem:degree-previous-pseudo-bands} is at most $\chi$ for all offsets at least $x_{\mathrm{crit}}$. The pseudo-band boundaries are defined by a sequence of offsets from the core boundary $R/2$. We set $\beta:=\alpha-1/2$ and $w_n:=1/\log\log n$, and we fix a sufficiently small constant $k>0$.

Starting with $x_0:=0$, we first use steps of size $w_n$ until the already defined offset $x_{\mathrm{crit}}$ is reached. More precisely, while $x_h<x_{\mathrm{crit}}$, we set $x_{h+1}:=x_h+w_n$, and we define $h_{\mathrm{crit}}:=\min\{h:x_h\ge x_{\mathrm{crit}}\}$. From this point on, we use the coarser recursive decomposition. For $h\ge h_{\mathrm{crit}}$, set
\begin{align*}
x_{h+1}
:=
x_h+\frac1\beta\log\left(1+k e^{(1-\alpha)x_h}\right)
\end{align*}
as long as $x_{h+1}< r_{\mathrm{out}}'-R/2$. Let $J$ be the last index obtained.

It remains to decide whether the remaining interval up to the provisional pseudo-outer boundary should form one more intermediate pseudo-band. This technical case distinction avoids creating a final intermediate pseudo-band of too small width: if the remaining gap is smaller than the minimum width of a coarser band, it is absorbed into the pseudo-outer region; otherwise, one last intermediate pseudo-band is added up to the provisional pseudo-outer boundary. Thus we set: 
\begin{align*}
L
:=
\begin{cases}
J, & \text{if } r_{\mathrm{out}}'-R/2-x_J<\beta^{-1}\log(1+k),\\
J+1, & \text{otherwise,}
\end{cases}
\end{align*}
and in the second case we additionally set $x_L:=r_{\mathrm{out}}'-R/2$. The offsets $x_0,\dots,x_L$ from the core boundary $R/2$ define the intermediate pseudo-bands. For each $h\in\{0,\dots,L-1\}$, let
\begin{align*}
\widehat A_h
:=
\{v\in V(G):x_h<\widehat x(v)\le x_{h+1}\}.
\end{align*}

The remaining vertices form the pseudo-outer region. We now denote its final radial boundary by $r_{\mathrm{out}}:=R/2+x_L$. By the preceding case distinction, this boundary either coincides with the provisional boundary $r_{\mathrm{out}}'$ or differs from it by less than the minimum coarser band width, which is $O(1)$. Hence, we have $r_{\mathrm{out}}=R-4\log\log n+O(1)$. We define the pseudo-outer region as 
\begin{align*}
\widehat{\mathcal O}
:=
\{v\in V(G):\widehat x(v) >x_L\}.
\end{align*}

We have now defined a subdivision of the disk into the pseudo-clique, the intermediate pseudo-bands, and the pseudo-outer region. The next corollary records that the total number of pseudo-regions in this subdivision is $O(\log\log n)$.

\begin{corollary}[Number of pseudo-regions]
\label{cor:number-of-radial-bands}
The number of pseudo-regions is $O(\log\log n)$.
\end{corollary}

\begin{proof}
The fine part covers the fixed offset interval $[0,x_{\mathrm{crit}}]$.
Since $x_{\mathrm{crit}}$ is a constant and the step size there is $w_n=1/\log\log n$, this creates $O(\log\log n)$ thin bands.

It remains to count the bands in the coarse part. Since $k>0$ is fixed, there is a constant $x_\ast\ge x_{\mathrm{crit}}$ such that, for every $x\ge x_\ast$,
$\log(1+k e^{(1-\alpha)x})\ge (1-\alpha)x/2$. Before reaching $x_\ast$, each recursive step has width at least $\beta^{-1}\log(1+k)$, so only constantly many additional bands are created. Once $x_h\ge x_\ast$, the recursion gives
\begin{align*}
x_{h+1}
=
x_h+\frac1\beta\log(1+k e^{(1-\alpha)x_h})
\ge
\left(1+\frac{1-\alpha}{2\beta}\right)x_h .
\end{align*}

Thus, the offsets grow geometrically until they reach $r_{\mathrm{out}}-R/2=\Theta(\log n)$. Starting from a constant offset, this takes only $O(\log\log n)$ further steps.
Therefore $L=O(\log\log n)$.
\end{proof}

The corresponding true regions are defined in the same way, but using the true offset $x(v):=r(v)-R/2$ instead of the estimated offset $\widehat x(v)$. They are denoted by $\mathcal C$, $A_h$ for $h\in\{0,\dots,L-1\}$, and $\mathcal O$.

\subsubsection*{Colouring process}
\label{subsubsec:sequential-colouring-process}

The pseudo-regions are processed sequentially in the order $\widehat{\mathcal C},\widehat A_0,\dots,\widehat A_{L-1},\widehat{\mathcal O}$. Each pseudo-region is active for $T=O(\log\log n)$ rounds, with a sufficiently large hidden constant. The colouring lemmas in the subsequent sections show that this choice ensures that, a.a.s., all vertices in the active pseudo-region are coloured within these $T$ rounds\footnote{Here, one round refers to one iteration of RCTDEG in an active band. Such an iteration is implemented by two CONGEST communication rounds: first, each active vertex sends its sampled colour to its neighbours; second, vertices communicate whether the sampled colour was kept permanently. This constant factor is absorbed in all round-complexity bounds.}. The algorithm then proceeds to the next pseudo-region. 
 
For a vertex $v$, at the end of round $t$, we write $\psi_t(v)\in\Psi$ for its permanent colour, and $\psi_t(v)=\bot$ if $v$ is still uncoloured. Conflicts inside an active pseudo-region are resolved by degree priority. For vertices $u$ and $v$, we set
\begin{align*}
u\prec v
\quad\Longleftrightarrow\quad
(\deg(u),\operatorname{id}(u))
<_{\mathrm{lex}}
(\deg(v),\operatorname{id}(v)).
\end{align*}

Thus, vertices of larger degree have higher priority, and identifiers break ties. When a pseudo-region $\mathcal R$ becomes active, each vertex $v\in\mathcal R$ first removes all colours already used by coloured neighbours in previously active pseudo-regions, yielding its initial palette $\Psi^{1}(v)$. At the beginning of round $t\ge 2$, every uncoloured vertex $v\in\mathcal R$ further removes colours already kept permanently by neighbours in the same pseudo-region:
\begin{align*}
\Psi^{t}(v)
:=
\Psi^{1}(v)
\setminus
\{\psi_{t-1}(u):u\in N(v)\cap\mathcal R,\ \psi_{t-1}(u)\neq\bot\}.
\end{align*}

During round $t$ of an active phase, the uncoloured vertices in $\mathcal R$ run RCTDEG as in~\cite{mausruff2026distributed}, using the current palettes $\Psi^t(v)$ and the degree-priority order restricted to $\mathcal R$. Each active vertex $v$ samples a colour uniformly from $\Psi^t(v)$ and sends it to its neighbours. The sampled colour becomes permanent unless some higher-priority neighbour in $\mathcal R$ samples the same colour in that round; in that case, $v$ stays uncoloured and participates again in the next round.

All messages have size $O(\log|\Psi|)=O(\log n)$ bits, so the algorithm is CONGEST-compliant. By Corollary~\ref{cor:number-of-radial-bands}, there are $O(\log\log n)$ pseudo-regions, and each is processed for $O(\log\log n)$ rounds. Hence the total number of rounds is $O((\log\log n)^2)$.

\subsection{Accuracy of pseudo-regions}
\label{subsec:pseudo-band-size-estimates}

In this subsection we prove the basic estimates used in the pseudo-band analysis. We first show that observed degrees concentrate around their expectations for all vertices with radius at most $r_{\mathrm U}$. This gives an additive $o(1)$ error bound for the estimated radius on $[r_{\mathrm L},r_{\mathrm U}]$. In the subsequent comparison lemmas, this radius-estimation bound is used to show that each pseudo-band has size within a factor $1\pm o(1)$ of the expected size of its corresponding true radial region.


\subsubsection*{Degree concentration and radius estimates}

The next estimate establishes a uniform concentration bound for degrees. More precisely, it shows that every vertex with radius up to the upper cutoff $r_{\mathrm U}$ has observed degree within a small multiplicative error of its expected degree. Recall that $r_{\mathrm U}$ was chosen so that $\bar d(r_{\mathrm U})=(K/4)\log^2 n$, for a sufficiently large constant $K>0$. We set $\delta_n:=\sqrt{48/(K\log n)}=o(1)$.

\begin{lemma}[Degree concentration]
\label{lem:degree-concentration-above-cutoff}
With high probability, every vertex $v$ with $r(v)\le r_{\mathrm U}$ satisfies $\deg(v)=(1\pm\delta_n)\bar d(r(v))$.
\end{lemma}

\begin{proof}
Let $v$ be a vertex with $r(v)\le r_{\mathrm U}$, and set $\mu_v:=\bar d(r(v))$. By the definition of $r_{\mathrm U}$ and the monotonicity of $\bar d$ from Lemma~\ref{lem:expected-degree-monotone}, we have $\mu_v\ge \bar d(r_{\mathrm U})=(K/4)\log^2 n$.

Conditioned on its radius, the degree of $v$ is Poisson distributed with mean $\mu_v$. We apply the standard Chernoff bounds for Poisson random variables \cite[Theorem~5.4]{mitzenmacher2005probabilit}: if $X$ is Poisson with mean $\mu$, then $\Pr[X\ge x]\le e^{-\mu}(e\mu/x)^x$ for $x>\mu$, and $\Pr[X\le x]\le e^{-\mu}(e\mu/x)^x$ for $x<\mu$. For the upper tail, apply this with $X=\deg(v)$, $\mu=\mu_v$, and
$x=(1+\delta_n)\mu_v$. Then
\begin{align*}
\Pr[\deg(v)\ge (1+\delta_n)\mu_v\mid r(v)]
&\le
e^{-\mu_v}
\left(\frac{e}{1+\delta_n}\right)^{(1+\delta_n)\mu_v} \\
&=
\exp\left(
-\mu_v\left((1+\delta_n)\log(1+\delta_n)-\delta_n\right)
\right) \\
&\le
\exp\left(-\frac{\delta_n^2\mu_v}{3}\right),
\end{align*}
where the last inequality uses
$(1+\delta)\log(1+\delta)-\delta\ge \delta^2/3$ for $\delta\in(0,1)$. 
The lower tail is bounded analogously with $x=(1-\delta_n)\mu_v$. Combining the
two tails gives
\begin{align*}
\Pr\left[
|\deg(v)-\mu_v|>\delta_n\mu_v
\mid r(v)
\right]
\le
2\exp\left(-\frac{\delta_n^2\mu_v}{3}\right)
\le
2\exp(-4\log n)
=
2n^{-4},
\end{align*}
because $\delta_n^2=48/(K\log n)$ and
$\mu_v\ge (K/4)\log^2 n$. A union bound over all vertices gives the claim.
\end{proof}

We next establish the deterministic separation estimate that converts degree concentration into radius concentration. 

\begin{lemma}[Radius sensitivity of expected degree]
\label{lem:local-change-expected-degree}
For all $r\in[r_{\mathrm L},r_{\mathrm U}]$, we have
\begin{align*}
\bar d(r-\sqrt{\delta_n})>(1+\delta_n)\bar d(r)
\qquad\text{and}\qquad
\bar d(r+\sqrt{\delta_n})<(1-\delta_n)\bar d(r).
\end{align*}
\end{lemma}

\begin{proof}
By Lemma~\ref{lem:uniform-expected-degree-asymptotic}, applied with a constant buffer around $[r_{\mathrm L},r_{\mathrm U}]$, we have uniformly for $r\in[r_{\mathrm L},r_{\mathrm U}]$,
\begin{align*}
\frac{\bar d(r+\sqrt{\delta_n})}{\bar d(r)}
=
e^{-\sqrt{\delta_n}/2}(1\pm o(\sqrt{\delta_n}))
=
1-\frac{\sqrt{\delta_n}}{2}\pm o(\sqrt{\delta_n}).
\end{align*}

Since $\delta_n=o(\sqrt{\delta_n})$, this is smaller than $1-\delta_n$ for all sufficiently large $n$. Hence $\bar d(r+\sqrt{\delta_n})<(1-\delta_n)\bar d(r)$. The estimate for $r-\sqrt{\delta_n}$ follows analogously from $\bar d(r-\sqrt{\delta_n})/\bar d(r)=e^{\sqrt{\delta_n}/2}(1\pm o(\sqrt{\delta_n}))$, which is larger than $1+\delta_n$ for all sufficiently large $n$.
\end{proof}

The preceding results give the desired accuracy guarantee for the inverse-degree radius estimator: on the interval $[r_{\mathrm L},r_{\mathrm U}]$, the estimated radius $\widehat r(v)$ of a vertex $v$ differs from the true radius $r(v)$ by at most $o(1)$.


\begin{lemma}[Accuracy of the radius estimate]
\label{lem:radius-estimate-accuracy}
With high probability, every vertex $v$ with $r(v)\in[r_{\mathrm L},r_{\mathrm U}]$ satisfies $|\widehat r(v)-r(v)|<\sqrt{\delta_n}$.
\end{lemma}

\begin{proof}
Let $v$ be a vertex with $r(v)\in[r_{\mathrm L},r_{\mathrm U}]$. Since $\delta_n=o(1)$, we have $r(v)-\sqrt{\delta_n},r(v)+\sqrt{\delta_n}\in(0,R)$ for all sufficiently large $n$. We work on the high-probability event from Lemma~\ref{lem:degree-concentration-above-cutoff}, which gives $\deg(v)=(1\pm\delta_n)\bar d(r(v))$. Moreover, Lemma~\ref{lem:local-change-expected-degree} shows that shifting the radius by $\sqrt{\delta_n}$ changes the expected degree by more than this multiplicative error. Hence,
\begin{align*}
\bar d(r(v)+\sqrt{\delta_n})
<
(1-\delta_n)\bar d(r(v))
\le
\deg(v)
\le
(1+\delta_n)\bar d(r(v))
<
\bar d(r(v)-\sqrt{\delta_n}).
\end{align*}

By Lemma~\ref{lem:expected-degree-monotone}, the function $\bar d$ is strictly decreasing on $(0,R]$. Therefore $\deg(v)\in(\bar d(R),\bar d(0))$, and by the definition of $\bar d_{\mathrm{ext}}^{-1}$ we have
$\widehat r(v)=\bar d^{-1}(\deg(v))$. Applying the inverse to the inequalities above gives
$\widehat r(v)\in(r(v)-\sqrt{\delta_n},r(v)+\sqrt{\delta_n})$, and thus
$|\widehat r(v)-r(v)|<\sqrt{\delta_n}$.
\end{proof}

\subsubsection*{Cutoff vertices}

We first separate the vertices beyond the upper cutoff $r_{\mathrm U}$. By definition, they have degree at most $K\log^2 n$ w.h.p. Since the pseudo-outer boundary has radius smaller than $r_{\mathrm U}$ and uses the larger threshold $K'\log^2 n$, with $K'>K$, they are assigned to the pseudo-outer region w.h.p.

\begin{lemma}[Assignment of upper cutoff vertices]
\label{lem:assignment-outer-cutoff-vertices}
With high probability, every vertex \(v\) with \(r(v)\ge r_{\mathrm U}\) satisfies \(\widehat r(v)>r_{\mathrm{out}}\).
\end{lemma}

\begin{proof}
By the definition of the upper cutoff $r_{\mathrm U}$, w.h.p. every vertex $v$ with $r(v)\ge r_{\mathrm U}$ satisfies $\deg(v)\le K\log^2 n$. Since $K'>K$, we have $\deg(v)<K'\log^2 n=\bar d(r_{\mathrm{out}}')$. As $\bar d_{\mathrm{ext}}^{-1}$ is non-increasing, this implies $\widehat r(v)>r_{\mathrm{out}}'$. The final pseudo-outer boundary $r_{\mathrm{out}}$ is at most the provisional boundary $r_{\mathrm{out}}'$, and hence $\widehat r(v)>r_{\mathrm{out}}$, as claimed.
\end{proof}

\subsubsection*{Size of pseudo-clique}

We next analyse the innermost pseudo-region, the pseudo-clique. With high probability, it is squeezed between two centred balls around the core boundary: $B_0(R/2-o(1))\subseteq\widehat{\mathcal C}$ and $\widehat{\mathcal C}\subseteq B_0(R/2+o(1))$. Hence, the pseudo-clique differs from the true core only within an $o(1)$ neighbourhood of $R/2$, and has size $(1\pm o(1))$ times the expected core size. The proof uses a direct degree comparison for $r(v)\le r_{\mathrm L}$ and Lemma~\ref{lem:radius-estimate-accuracy} on $[r_{\mathrm L},r_{\mathrm U}]$.



\begin{lemma}[Assignment of inner cutoff vertices]
\label{lem:assignment-inner-cutoff-vertices}
With high probability, every vertex \(v\) with \(r(v)\le r_{\mathrm L}\) satisfies \(\widehat r(v)\le R/2\).
\end{lemma}

\begin{proof}
A vertex belongs to $\widehat{\mathcal C}$ whenever $\widehat r(v)\le R/2$, or equivalently whenever $\deg(v)\ge \bar d(R/2)$. By Lemma~\ref{lem:uniform-expected-degree-asymptotic}, uniformly for $r\in[R/2-1,R/2]$, $\bar d(r)=c_\alpha ne^{-r/2}(1\pm o(1))$, where $c_\alpha:=2\alpha/(\pi(\alpha-1/2))$. Hence $\bar d(R/2-1)=e^{1/2}(1\pm o(1))\bar d(R/2)$, and so for sufficiently large $n$ there exists a constant $\eta>0$ such that $\bar d(R/2-1)\ge(1+\eta)\bar d(R/2)$.

Let $v$ satisfy $r(v)\le r_{\mathrm L}=R/2-1$. By the monotonicity of $\bar d$ from Lemma~\ref{lem:expected-degree-monotone},
$\bar d(r(v))\ge\bar d(R/2-1)$. On the high-probability event of Lemma~\ref{lem:degree-concentration-above-cutoff},
\begin{align*}
\deg(v)
\ge
(1-\delta_n)\bar d(r(v))
\ge
(1-\delta_n)(1+\eta)\bar d(R/2)
>
\bar d(R/2),
\end{align*}
since $\delta_n=o(1)$. Therefore $\widehat r(v)\le R/2$, as claimed.
\end{proof}

For later use, let \(\mathcal E\) be the event that the conclusions of Lemmas~\ref{lem:radius-estimate-accuracy}, \ref{lem:assignment-outer-cutoff-vertices}, and~\ref{lem:assignment-inner-cutoff-vertices} hold and, consequently, for every \(h\in\{0,\dots,L-1\}\), every vertex in a pseudo-region preceding \(\widehat A_h\) has true radius at most \(R/2+x_h+\sqrt{\delta_n}\), while \(\widehat A_h\subseteq A_h^+\), where \(A_h^+:=\{w\in V(G):x_h-\sqrt{\delta_n}<x(w)\le x_{h+1}+\sqrt{\delta_n}\}\). By the above lemmas, \(\mathcal E\) holds w.h.p.

\begin{lemma}[Size of the pseudo-clique]
\label{lem:pseudoclique-size}
W.h.p. we have that $|\widehat{\mathcal C}|=(1\pm o(1))\mathbb E[|\mathcal C|]$.
\end{lemma}

\begin{proof}
Define $\mathcal B_-:=V(G)\cap B_0(R/2-\sqrt{\delta_n})$ and $\mathcal B_+:=V(G)\cap B_0(R/2+\sqrt{\delta_n})$. We first show that, w.h.p.,
\begin{align}
\mathcal B_-
\subseteq
\widehat{\mathcal C}
\subseteq
\mathcal B_+ .
\label{eq:pseudoclique-sandwich}
\end{align}

For the lower inclusion, let $v\in\mathcal B_-$. If $r(v)\le r_{\mathrm L}$, then $v\in\widehat{\mathcal C}$ by Lemma~\ref{lem:assignment-inner-cutoff-vertices}. Otherwise, $r(v)\in[r_{\mathrm L},R/2-\sqrt{\delta_n}]$, and Lemma~\ref{lem:radius-estimate-accuracy} gives $\widehat r(v)\le r(v)+\sqrt{\delta_n}\le R/2$. Hence $v\in\widehat{\mathcal C}$. For the upper inclusion, Lemma~\ref{lem:assignment-outer-cutoff-vertices} implies that w.h.p. every vertex with $r(v)\ge r_{\mathrm U}$ belongs to $\widehat{\mathcal O}$, and hence not to $\widehat{\mathcal C}$. Thus, if $v\in\widehat{\mathcal C}$ and $r(v)>R/2+\sqrt{\delta_n}$, then necessarily $r(v)\in[r_{\mathrm L},r_{\mathrm U}]$. By Lemma~\ref{lem:radius-estimate-accuracy}, $\widehat r(v)\ge r(v)-\sqrt{\delta_n}>R/2$, contradicting $v\in\widehat{\mathcal C}$. This proves \eqref{eq:pseudoclique-sandwich}.

It remains to compare the sizes. By Lemma~\ref{lem:centered-ball-measure}, applied with $r=R/2+\sqrt{\delta_n}$, and since $\delta_n=o(1)$, we have
\begin{align*}
\mathbb E[|\mathcal B_+|]
\le
e^{\alpha\sqrt{\delta_n}}(1+o(1))\mathbb E[|\mathcal C|]
=
(1+o(1))\mathbb E[|\mathcal C|].
\end{align*}


Moreover, since $\mathbb E[|\mathcal C|]=\Theta(n^{1-\alpha})$ and $|\mathcal B_+|$ is a Poisson random variable, a Chernoff bound gives, w.h.p., $|\mathcal B_+|\le (1+o(1))\mathbb E[|\mathcal B_+|]$. Hence, w.h.p., $|\mathcal B_+|\le (1+o(1))\mathbb E[|\mathcal C|]$. The same argument with $\mathcal B_-$ gives, w.h.p., $|\mathcal B_-|\ge (1-o(1))\mathbb E[|\mathcal C|]$. Together with \eqref{eq:pseudoclique-sandwich}, this implies the claim.
\end{proof}

\subsubsection*{Size of pseudo-bands}

We next consider the intermediate pseudo-bands. By construction, their boundary radii lie in the range where Lemma~\ref{lem:radius-estimate-accuracy} applies. We first prove a single-interval estimate: a pseudo-band has the correct asymptotic size whenever the corresponding true interval is wider than the estimation error and has polynomially large expected size.

\begin{lemma}[Size of pseudo-bands]
\label{lem:pseudoband-size}
Let \(I=(a,b]\subseteq (R/2,r_{\mathrm {out}})\) be a radial interval, and define
\begin{align*}
\mathcal B(I)
&:=\{v\in V(G): r(v)\in I\},
&
\widehat{\mathcal B}(I)
&:=\{v\in V(G): \widehat r(v)\in I\}.
\end{align*}

Suppose that \(b-a=\omega(\sqrt{\delta_n})\) and
\(\mathbb E[|\mathcal B(I)|]=n^{\Omega(1)}\). Then, with high probability,
\(|\widehat{\mathcal B}(I)|=(1\pm o(1))\mathbb E[|\mathcal B(I)|]\).
\end{lemma}

\begin{proof}
We define $I_-:=(a+\sqrt{\delta_n},b-\sqrt{\delta_n}]$ and $I_+:=(a-\sqrt{\delta_n},b+\sqrt{\delta_n}]$. We first show that $\mathcal B(I_-) \subseteq \widehat{\mathcal B}(I) \subseteq \mathcal B(I_+)$. For the lower inclusion, let $v\in\mathcal B(I_-)$. Since $I\subseteq(R/2,r_{\mathrm {out}})$, we have $r(v)\in[r_{\mathrm L},r_{\mathrm U}]$ for all sufficiently large $n$. Hence, Lemma~\ref{lem:radius-estimate-accuracy} gives that w.h.p. $|\widehat r(v)-r(v)|<\sqrt{\delta_n}$, and therefore $\widehat r(v)\in I$. Thus $v\in\widehat{\mathcal B}(I)$.

For the upper inclusion, let $v\in\widehat{\mathcal B}(I)$. Lemma~\ref{lem:assignment-inner-cutoff-vertices} implies that w.h.p. $r(v)>r_{\mathrm L}$, since otherwise $v\in\widehat{\mathcal C}$, which is disjoint from $\widehat{\mathcal B}(I)$. Similarly, Lemma~\ref{lem:assignment-outer-cutoff-vertices} implies that w.h.p. $r(v)<r_{\mathrm U}$, since otherwise $v\in\widehat{\mathcal O}$, which is disjoint from $\widehat{\mathcal B}(I)$. Hence $r(v)\in[r_{\mathrm L},r_{\mathrm U}]$. By Lemma~\ref{lem:radius-estimate-accuracy}, $|r(v)-\widehat r(v)|<\sqrt{\delta_n}$. As $\widehat r(v)\in I$, this implies $r(v)\in I_+$, and hence $v\in\mathcal B(I_+)$.

Set $\mu_I:=\mathbb E[|\mathcal B(I)|]$. By Lemma~\ref{lem:centered-ball-measure}, applied to the endpoints $a$ and $b$, which satisfy $a,b\ge r_{\mathrm L}=\omega(1)$, we have
\begin{align*}
\mu_I
=
(1+o(1))n e^{-\alpha(R-b)}
\bigl(1-e^{-\alpha(b-a)}\bigr).
\end{align*}

Applying the same estimate to $I_-$ and $I_+$, and using $b-a=\omega(\sqrt{\delta_n})$ and $\delta_n=o(1)$, gives $\mathbb E[|\mathcal B(I_-)|]\geq(1-o(1))\mu_I$ and $\mathbb E[|\mathcal B(I_+)|]\leq(1+o(1))\mu_I$.

Moreover, $|\mathcal B(I_-)|$ and $|\mathcal B(I_+)|$ are Poisson random variables with expectations $n^{\Omega(1)}$. Chernoff concentration therefore gives, w.h.p., $|\mathcal B(I_-)|\ge(1-o(1))\mu_I$ and $|\mathcal B(I_+)|\le(1+o(1))\mu_I$. Combining these estimates with the inclusion above yields $|\widehat{\mathcal B}(I)|=(1\pm o(1))\mu_I$, as claimed.
\end{proof}

Applying Lemma~\ref{lem:pseudoband-size} to the intervals defining the intermediate pseudo-bands gives the following simultaneous estimate.

\begin{corollary}[Sizes of the pseudo-bands]
\label{cor:pseudoband-sizes}
With high probability, $|\widehat A_h|=(1\pm o(1))\mathbb E[|A_h|]$ for all $h\in\{0,\dots,L-1\}$.
\end{corollary}

\begin{proof}
For $h\in\{0,\dots,L-1\}$, let $I_h:=(R/2+x_h,R/2+x_{h+1}]$, so that $A_h=\mathcal B(I_h)$, and set $\Delta_h:=x_{h+1}-x_h$. In the fine range, $\Delta_h=w_n=1/\log\log n$, while in the coarse range, including a possible final truncated band, $\Delta_h\ge\beta^{-1}\log(1+k)=\Omega(1)$. Since $\sqrt{\delta_n}=O((\log n)^{-1/4})=o(1/\log\log n)$, every band satisfies $\Delta_h=\omega(\sqrt{\delta_n})$.

It remains to verify that the expected band sizes are polynomially large. Set $a_h:=R/2+x_h$ and $b_h:=R/2+x_{h+1}$. By Lemma~\ref{lem:centered-ball-measure},
\begin{align*}
\mathbb E[|A_h|]
=
n\bigl(\mu(B_0(b_h))-\mu(B_0(a_h))\bigr)
=
(1+o(1))n e^{-\alpha(R-b_h)}
\bigl(1-e^{-\alpha\Delta_h}\bigr).
\end{align*}

Since \(b_h\ge R/2\), the first factor is \(\Omega(n^{1-\alpha})\). Moreover, \(1-e^{-\alpha\Delta_h}=\Theta(1/\log\log n)\) in the fine range and is bounded away from zero in the coarse range. Hence, uniformly over all bands,
\begin{align*}
\mathbb E[|A_h|]
=
\Omega\left(\frac{n^{1-\alpha}}{\log\log n}\right)
=
n^{\Omega(1)}.
\end{align*}

Thus Lemma~\ref{lem:pseudoband-size} applies to every $I_h$. Since $L=O(\log\log n)$, a union bound proves the claim simultaneously for all $h\in\{0,\dots,L-1\}$.
\end{proof}

\subsubsection*{Size of pseudo-outer band}

It remains to estimate the size of the pseudo-outer region, whose radial boundary is $r_{\mathrm{out}}$. By Lemma~\ref{lem:assignment-outer-cutoff-vertices}, all vertices with radius at least $r_{\mathrm U}$ belong to $\widehat{\mathcal O}$ with high probability, while vertices with radius in $[r_{\mathrm{out}},r_{\mathrm U})$ can be controlled using Lemma~\ref{lem:radius-estimate-accuracy}.

\begin{lemma}[Size of the pseudo-outer region]
\label{lem:size-pseudo-outer-region}
With high probability, $|\widehat{\mathcal O}|= (1\pm o(1))\mathbb E[|\mathcal O|]$.
\end{lemma}

\begin{proof}
We first prove that
\begin{align}
|\mathcal O_-|:=|V(G)\setminus B_0(r_{\mathrm{out}}+\sqrt{\delta_n})|\le|\widehat{\mathcal O}|\le|V(G)\setminus B_0(r_{\mathrm{out}}-\sqrt{\delta_n})|=:|\mathcal O_+|.
\label{eq:pseudo-outer-sandwich}
\end{align}

For the lower bound, let $v$ satisfy $r(v)\ge r_{\mathrm{out}}+\sqrt{\delta_n}$. If $r(v)\ge r_{\mathrm U}$, then w.h.p. $v\in\widehat{\mathcal O}$ by Lemma~\ref{lem:assignment-outer-cutoff-vertices}. Otherwise, $r(v)\in[r_{\mathrm{out}}+\sqrt{\delta_n},r_{\mathrm U})$, and w.h.p. Lemma~\ref{lem:radius-estimate-accuracy} gives $\widehat r(v)>r(v)-\sqrt{\delta_n}\ge r_{\mathrm{out}}$. Hence $v\in\widehat{\mathcal O}$.

For the upper bound, let $v$ satisfy $r(v)\le r_{\mathrm{out}}-\sqrt{\delta_n}$. If $r(v)\le r_{\mathrm L}$, then w.h.p. $v\in\widehat{\mathcal C}$ by Lemma~\ref{lem:assignment-inner-cutoff-vertices}, and hence $v\notin\widehat{\mathcal O}$. Otherwise, $r(v)\in(r_{\mathrm L},r_{\mathrm{out}}-\sqrt{\delta_n}]$, and w.h.p. Lemma~\ref{lem:radius-estimate-accuracy} gives $\widehat r(v)<r(v)+\sqrt{\delta_n}\le r_{\mathrm{out}}$. Thus $v\notin\widehat{\mathcal O}$.

Recall that $\mathcal O:=V(G)\setminus B_0(r_{\mathrm{out}})$. Since $r_{\mathrm{out}}=R-4\log\log n+O(1)$, Lemma~\ref{lem:centered-ball-measure} gives
\begin{align*}
\mathbb E[|\mathcal O|]=n\bigl(1-\mu(B_0(r_{\mathrm{out}}))\bigr)=O(n)\left(1-O\left((\log n)^{-4\alpha}\right)\right)=(1-o(1))n.
\end{align*}







Since $\sqrt{\delta_n}=o(1)$, Lemma~\ref{lem:centered-ball-measure} applied at $r_{\mathrm{out}}\pm\sqrt{\delta_n}$ yields $\mathbb E[|\mathcal O_-|]\ge(1-o(1))\mathbb E[|\mathcal O|]$ and $\mathbb E[|\mathcal O_+|]\le(1+o(1))\mathbb E[|\mathcal O|]$. Finally, $|\mathcal O_-|$ and $|\mathcal O_+|$ are Poisson random variables with expectations $\Theta(n)$. Chernoff concentration therefore gives, w.h.p., $|\mathcal O_-|\ge(1-o(1))\mathbb E[|\mathcal O|]$ and $|\mathcal O_+|\le(1+o(1))\mathbb E[|\mathcal O|]$. Combining this with \eqref{eq:pseudo-outer-sandwich} proves $|\widehat{\mathcal O}|=(1\pm o(1))\mathbb E[|\mathcal O|]$.
\end{proof}

\subsubsection*{Comparison with true regions}

We conclude this subsection by combining the pseudo-region estimates with concentration of the corresponding true regions.

\begin{lemma}[Concentration of true regions]
\label{lem:true-region-concentration}
With high probability, we have
\begin{align*}
|\mathcal C|=(1\pm o(1))\mathbb E[|\mathcal C|],\quad
|A_h|=(1\pm o(1))\mathbb E[|A_h|]\ \forall h\in\{0,\dots,L-1\},
\quad
|\mathcal O|=(1\pm o(1))\mathbb E[|\mathcal O|].
\end{align*}
\end{lemma}

\begin{proof}
The sizes $|\mathcal C|$, $|A_h|$, and $|\mathcal O|$ are Poisson random variables with means equal to their expected sizes. By Lemma~\ref{lem:centered-ball-measure}, $\mathbb E[|\mathcal C|]=\Theta(n^{1-\alpha})$, while the proof of Corollary~\ref{cor:pseudoband-sizes} shows that $\mathbb E[|A_h|]=n^{\Omega(1)}$ uniformly for all $h\in\{0,\dots,L-1\}$. Moreover, since $r_{\mathrm{out}}=R-4\log\log n+O(1)$,
\begin{align*}
\mathbb E[|\mathcal O|]=n\bigl(1-\mu(B_0(r_{\mathrm{out}}))\bigr)=(1-o(1))n.
\end{align*}




Chernoff concentration therefore gives the claimed estimate for each region. Since $L=O(\log\log n)$, a union bound over all intermediate bands proves that the estimates hold simultaneously for all bands w.h.p.
\end{proof}

Combining the preceding pseudo-region estimates with Lemma~\ref{lem:true-region-concentration} yields the following corollary.

\begin{corollary}[Pseudo-region size comparison]
\label{cor:pseudo-region-size-comparison}
With high probability, the pseudo- and true regions have asymptotically the same sizes: $|\widehat{\mathcal C}|=(1\pm o(1))|\mathcal C|$, $|\widehat A_h|=(1\pm o(1))|A_h|$ for all $h\in\{0,\dots,L-1\}$, and $|\widehat{\mathcal O}|=(1\pm o(1))|\mathcal O|$.
\end{corollary}





Recall that an HRG is called \emph{typical} if it satisfies all structural properties that we have shown to hold a.a.s. Thus, from this point on, when analysing the colouring algorithm, we fix a typical HRG and only the randomness of the algorithm remains.

\subsection{Colouring the pseudo-clique}
\label{subsec:colouring-pseudo-clique}

We first analyse the pseudo-clique \(\widehat{\mathcal C}\). Although \(\widehat{\mathcal C}\) need not itself be a clique, its size is asymptotically bounded by the size of the core clique, and hence by \(\chi\). We therefore apply the general-graph near-optimal \textsc{RCTDEG} guarantee from Corollary~\ref{cor:near-tight-general-graph-colouring} to the induced graph \(G[\widehat{\mathcal C}]\): a palette of size $(1+o(1))\chi$ suffices for SRC to colour it within $O(\log\log n)$ rounds a.a.s.

\begin{corollary}[Colouring the pseudo-clique]
\label{cor:near-tight-pseudo-clique-colouring}
There exists a deterministic sequence \(\varepsilon_n=o(1)\) such that, if \(|\Psi|\ge(1+\varepsilon_n)\chi\), then SRC colours \(\widehat{\mathcal C}\) within \(O(\log\log n)\) rounds a.a.s.
\end{corollary}

\begin{proof}
During the first phase of SRC no colours have been fixed yet, and the active region is precisely $\widehat{\mathcal C}$. Thus SRC restricted to $\widehat{\mathcal C}$ is exactly RCTDEG on the induced graph $G[\widehat{\mathcal C}]$.

By Corollary~\ref{cor:pseudo-region-size-comparison}, there exists a deterministic sequence \(\xi_n=o(1)\) such that, w.h.p., $|\widehat{\mathcal C}|\le (1+\xi_n)|\mathcal C|\le (1+\xi_n)\chi $. Moreover, \(\chi=\Theta(n^{1-\alpha})\) by Lemma~\ref{lem:chromatic-number-bound}. Hence the relative slack required by Corollary~\ref{cor:near-tight-general-graph-colouring} is \(O(1/\log|\widehat{\mathcal C}|)=O(1/\log\chi)=O(1/\log n)\). Choose a deterministic sequence \(\varepsilon_n=o(1)\) such that \(\varepsilon_n\ge C_0(\xi_n+1/\log n)\) for a sufficiently large constant \(C_0\). Then, on the high-probability event above,
\begin{align*}
|\Psi|
\ge (1+\varepsilon_n)\chi
\ge \frac{1+\varepsilon_n}{1+\xi_n}\,|\widehat{\mathcal C}|
\ge \left(1+\Omega(\varepsilon_n-\xi_n)\right)|\widehat{\mathcal C}|
\ge \left(1+\Omega(1/\log n)\right)|\widehat{\mathcal C}|.
\end{align*}

Thus the palette has sufficient slack for Corollary~\ref{cor:near-tight-general-graph-colouring}. Applying that corollary to \(G[\widehat{\mathcal C}]\) shows that SRC colours \(\widehat{\mathcal C}\) within \(O(\log\log n)\) rounds a.a.s.
\end{proof}

\subsection{Colouring the pseudo-bands}
\label{subsec:colouring-pseudo-bands}

We next analyse the colouring phases for the intermediate pseudo-bands $\widehat A_0,\dots,\widehat A_{L-1}$. For a band $\widehat A_h$, let $\widehat P_h$ be the union of the previously processed pseudo-regions. For $v\in\widehat A_h$, define
\begin{align*}
d_{\mathrm{pre}}(v)
&:=
|N(v)\cap\widehat P_h|,
&
d_{\mathrm{int}}(v)
&:=
|N(v)\cap\widehat A_h|.
\end{align*}
These quantities separate two sources of colour restrictions from neighbouring vertices in a pseudo-band. The previous-region degree $d_{\mathrm{pre}}$ accounts for colours already used by neighbours before the band becomes active, while the internal degree $d_{\mathrm{int}}$ controls conflicts with neighbours during the active phase. The single-band lemma isolates the active phase: if each vertex starts with $\chi$ available colours and $d_{\mathrm{int}}$ is sufficiently small, then Sequential Radial Colouring colours the band in $O(\log\log n)$ rounds a.a.s. The previous-region degree $d_{\mathrm{pre}}$ is only needed later to bound inherited colour-palette restrictions when pseudo-bands are activated sequentially.

\subsubsection*{Degree to previous pseudo-regions}

We next bound the number of neighbours that a vertex in a pseudo-band may have in previously processed pseudo-regions. Recall that $x_{\mathrm{crit}}$ was chosen so that, for every band starting at an offset at least $x_{\mathrm{crit}}$, this previous-region degree is at most $\chi$ w.h.p.

\begin{lemma}[Degree to previous pseudo-regions]
\label{lem:degree-previous-pseudo-bands}
There exists a constant $c=c(\alpha,C)>0$ such that, for a sufficiently large choice of $x_{\mathrm{crit}}$, w.h.p.,
\begin{align*}
d_{\mathrm{pre}}(v)
\le
\begin{cases}
c\chi, & \text{if } x_h<x_{\mathrm{crit}},\\
\chi, & \text{if } x_h\ge x_{\mathrm{crit}},
\end{cases}
\end{align*}
for every $h\in\{0,\dots,L-1\}$ and every $v\in\widehat A_h$.
\end{lemma}

\begin{proof}
Fix $h\in\{0,\dots,L-1\}$ and $v\in\widehat A_h$, and write $x:=r(v)-R/2$. On the radius accuracy event \(\mathcal E\), which holds w.h.p., we have \(x\ge x_h-\sqrt{\delta_n}\), while every vertex in a pseudo-region preceding \(\widehat A_h\) has true radius at most \(R/2+x_h+\sqrt{\delta_n}\). Hence \(d_{\mathrm{pre}}(v)\) is bounded by the degree of \(v\) into the deterministic ball \(B_0(R/2+x_h+\sqrt{\delta_n})\).\footnote{The pseudo-region itself is degree-defined, so membership in it may be correlated with the edge to the fixed vertex. We therefore first dominate the pseudo-region degree by the degree into a deterministic enlarged radial region, obtained from the radius-accuracy event by shifting the relevant boundaries by \(\sqrt{\delta_n}\). Conditional on the fixed vertex, this dominating count is Poisson, and all expectation and Chernoff estimates are applied to it.} Conditioning on the position of \(v\), Lemma~\ref{lem:centered-ball-measure} and Lemma~\ref{lem:angular-threshold-asymptotics} yield
\begin{align*}
\mathbb E\left[|N(v)\cap B_0(R/2+x_h+\sqrt{\delta_n})|\mid r(v),\theta(v)\right]
&\le O\left(n^{1-\alpha}e^{-\alpha x}+n^{1-\alpha}e^{-x/2} \int_{-x}^{x_h+\sqrt{\delta_n}}e^{(\alpha-1/2)y}\,dy\right)\\
&=O\left(n^{1-\alpha}e^{-(1-\alpha)x_h}\right),
\end{align*}
where we used $x\ge x_h-\sqrt{\delta_n}$, $\sqrt{\delta_n}=o(1)$, and $x_h\ge0$.

By Lemma~\ref{lem:chromatic-number-bound} w.e.h.p. there exists a constant \(c_0>0\) such that \(\chi\ge c_0 n^{1-\alpha}\). Thus, after choosing the constant \(c=c(\alpha,C)>0\) in the statement sufficiently large, the deterministic threshold \(c\,c_0 n^{1-\alpha}\) is a fixed factor larger than the expectation $O\left(n^{1-\alpha}e^{-(1-\alpha)x_h}\right)$, uniformly over all relevant \(h\) and \(v\). A Poisson Chernoff bound gives
\begin{align*}
\Pr\left(
|N(v)\cap B_0(R/2+x_h+\sqrt{\delta_n})|>c\,c_0 n^{1-\alpha}
\,\middle|\,
r(v),\theta(v)
\right)
\le
\exp\left(-\Omega(n^{1-\alpha})\right).
\end{align*}


Together with \(\mathcal E\) and \(\chi\ge c_0n^{1-\alpha}\), this implies \(d_{\mathrm{pre}}(v)\le c\chi\) w.h.p. As the preceding expectation bound contains the factor \(e^{-(1-\alpha)x_h}\), we may choose a sufficiently large constant \(x_{\mathrm{crit}}=x_{\mathrm{crit}}(\alpha,C)\) such that, whenever \(x_h\ge x_{\mathrm{crit}}\), the conditional expectation is at most \(c_0 n^{1-\alpha}/4\). A Poisson Chernoff bound gives
\begin{align*}
\Pr\left(
|N(v)\cap B_0(R/2+x_h+\sqrt{\delta_n})|>c_0 n^{1-\alpha}
\,\middle|\,
r(v),\theta(v)
\right)
\le
\exp\left(-\Omega(n^{1-\alpha})\right).
\end{align*}

Together with \(\mathcal E\) and \(\chi\ge c_0n^{1-\alpha}\), this gives \(d_{\mathrm{pre}}(v)\le\chi\) w.h.p. for all \(h\) with \(x_h\ge x_{\mathrm{crit}}\) and all \(v\in\widehat A_h\). A union bound over all vertices and the \(L=O(\log\log n)\) pseudo-bands completes the proof.
\end{proof}

In the refined bands, the previous-region degree is essentially bounded by the degeneracy.

\begin{corollary}[Previous degree in refined bands]
\label{cor:previous-degree-degeneracy-refined}
With high probability, every $v\in\widehat A_h$ with $x_h<x_{\mathrm{crit}}$ satisfies
$d_{\mathrm{pre}}(v)\le(1+o(1))\kappa$.
\end{corollary}

\begin{proof}
Let $M_n$ denote the maximum of $\mathbb E[|N(u)\cap B_0(r)|\mid r(u)=r]$ over
$r\in[R/2,R/2+x_{\mathrm{crit}}+1]$. By continuity, this maximum is attained
at some $r_n$. Since $M_n=\Theta(n^{1-\alpha})=\omega(\log n)$,
Lemma~\ref{lem:centred-ball-degeneracy-lower-bound}, applied with $q=r_n$,
gives $M_n\le(1+o(1))\kappa$ w.h.p.

Fix a refined band $\widehat A_h$ and $v\in\widehat A_h$. On the radius accuracy event \(\mathcal E\), which holds w.h.p., the pseudo-previous region is contained in \(B_0(R/2+x_h+\sqrt{\delta_n})\), whereas \(r(v)\ge R/2+x_h-\sqrt{\delta_n}\). Hence \(d_{\mathrm{pre}}(v)\le |N(v)\cap B_0(r(v)+2\sqrt{\delta_n})|\). Since the additional annulus has width \(2\sqrt{\delta_n}=o(1)\), Lemmas~\ref{lem:centered-ball-measure} and~\ref{lem:angular-threshold-asymptotics} imply that
\begin{align*}
\mathbb E\left[|N(v)\cap B_0(r(v)+2\sqrt{\delta_n})|\mid r(v),\theta(v)\right]
\le
(1+o(1))M_n.
\end{align*}


A Poisson Chernoff bound applied to \(|N(v)\cap B_0(r(v)+2\sqrt{\delta_n})|\), followed by a union bound gives \(|N(v)\cap B_0(r(v)+2\sqrt{\delta_n})|\le(1+o(1))M_n\) w.h.p. simultaneously for all refined bands and all \(v\in\widehat A_h\). Together with \(\mathcal E\) and \(M_n\le(1+o(1))\kappa\), this proves \(d_{\mathrm{pre}}(v)\le(1+o(1))\kappa\) w.h.p.
\end{proof}

\subsubsection*{Internal degree within pseudo-bands}

The number of conflicts arising during an active phase is determined by the internal degree of the corresponding pseudo-band. Near the core, the refined decomposition yields a logarithmic reduction in this degree. Beyond $x_{\mathrm{crit}}$, the recursive choice of band widths ensures that, for sufficiently small $k$, the internal degree is an arbitrarily small constant fraction of $\chi$.

\begin{lemma}[Internal degree of pseudo-bands]
\label{lem:internal-degree-bands}
For every fixed $\eta>0$, there exists $k_0=k_0(\eta)>0$ such that, for every $k\in(0,k_0]$\footnote{Recall that $k$ is part of the band definition. The larger $k$, the wider the bands.}, w.h.p.,
\begin{align*}
d_{\mathrm{int}}(v)
\le
\begin{cases}
\chi/\sqrt{\log\log n}, & \text{if } x_h<x_{\mathrm{crit}},\\
\eta\chi, & \text{if } x_h\ge x_{\mathrm{crit}},
\end{cases}
\end{align*}
simultaneously for all $h\in\{0,\dots,L-1\}$ and all $v\in\widehat A_h$.
\end{lemma}

\begin{proof}
Fix $h\in\{0,\dots,L-1\}$ and $v\in\widehat A_h$. Recall that $\beta=\alpha-\frac12$, and write $x:=r(v)-R/2$ and $\Delta_h:=x_{h+1}-x_h$. On the radius accuracy event \(\mathcal E\), which holds w.h.p., we have \(\widehat A_h\subseteq A_h^+\), where \(A_h^+=\{w\in V(G):x_h-\sqrt{\delta_n}<x(w)\le x_{h+1}+\sqrt{\delta_n}\}\), and \(x\ge x_h-\sqrt{\delta_n}\). Hence \(d_{\mathrm{int}}(v)\le |N(v)\cap A_h^+|\). Conditioning on the position of \(v\), Lemmas~\ref{lem:centered-ball-measure} and~\ref{lem:angular-threshold-asymptotics} give
\begin{align*}
\mathbb E\left[|N(v)\cap A_h^+|\mid r(v),\theta(v)\right]
&\le
O\left(
n^{1-\alpha}e^{-x/2}
\int_{x_h-\sqrt{\delta_n}}^{x_{h+1}+\sqrt{\delta_n}}
e^{\beta y}\,dy
\right)\\
&\le
O\left(
n^{1-\alpha}e^{-(1-\alpha)x_h}
\left(e^{\beta(\Delta_h+2\sqrt{\delta_n})}-1\right)
\right).
\end{align*}

If \(x_h<x_{\mathrm{crit}}\), then \(\Delta_h=1/\log\log n\). Since \(\sqrt{\delta_n}=o(1/\log\log n)\) and \(x_h\ge0\), the above expectation satisfies
\begin{align*}
\mathbb E\left[|N(v)\cap A_h^+|\mid r(v),\theta(v)\right]
=
O\left(\frac{n^{1-\alpha}}{\log\log n}\right)
=
o\left(\frac{n^{1-\alpha}}{\sqrt{\log\log n}}\right).
\end{align*}

By Lemma~\ref{lem:chromatic-number-bound}, there is a deterministic constant \(c_0>0\) such that \(\chi\ge c_0n^{1-\alpha}\) w.e.h.p. Since the expectation above is \(o(n^{1-\alpha}/\sqrt{\log\log n})\), a Poisson Chernoff bound applied to \(|N(v)\cap A_h^+|\) with threshold \(c_0n^{1-\alpha}/\sqrt{\log\log n}\), followed by a union bound, gives \(|N(v)\cap A_h^+|\le c_0n^{1-\alpha}/\sqrt{\log\log n}\) w.h.p. simultaneously for all such \(h\) and \(v\). Together with the radius-accuracy event \(\mathcal E\) and the chromatic-number bound, which both hold w.h.p., this implies \(d_{\mathrm{int}}(v)\le \chi/\sqrt{\log\log n}\) w.h.p. simultaneously for all \(h\) with \(x_h<x_{\mathrm{crit}}\) and all \(v\in\widehat A_h\).

In the recursive range \(x_h\ge x_{\mathrm{crit}}\), the recursive construction gives \(\Delta_h\le\beta^{-1}\log(1+k e^{(1-\alpha)x_h})\), since the final band may be truncated and hence \(e^{\beta\Delta_h}-1\le k e^{(1-\alpha)x_h}\). Together with \(e^{2\beta\sqrt{\delta_n}}=1+o(1)\), this yields
\begin{align*}
e^{\beta(\Delta_h+2\sqrt{\delta_n})}-1
=
e^{2\beta\sqrt{\delta_n}}
\left(e^{\beta\Delta_h}-1\right)
+
e^{2\beta\sqrt{\delta_n}}-1
\le
(1+o(1))k e^{(1-\alpha)x_h}+o(1).
\end{align*}

It follows that $|N(v)\cap A_h^+|$ has conditional expectation \(O((k+o(1))n^{1-\alpha})\). Let \(c_0>0\) be such that \(\chi\ge c_0n^{1-\alpha}\) w.e.h.p. by Lemma~\ref{lem:chromatic-number-bound}. Choosing \(k_0(\eta)>0\) sufficiently small, this expectation is at most \(c_0\eta n^{1-\alpha}/4\) for every \(k\in(0,k_0(\eta)]\) and all sufficiently large \(n\). A Poisson Chernoff bound, followed by a union bound gives \(|N(v)\cap A_h^+|\le c_0\eta n^{1-\alpha}\) simultaneously for all such \(h\) and \(v\), w.h.p. Together with the radius-accuracy event \(\mathcal E\) and the chromatic-number bound, this implies \(d_{\mathrm{int}}(v)\le\eta\chi\) w.h.p. simultaneously for all \(h\) with \(x_h\ge x_{\mathrm{crit}}\) and all \(v\in\widehat A_h\).
\end{proof}

\subsubsection*{Band colouring process}

We next establish a colouring bound for each pseudo-band during its active phase. Once inherited restrictions from earlier pseudo-regions have been removed, the bands can be analysed independently: if every vertex starts with at least $\chi$ available colours upon activation, then Sequential Radial Colouring colours the band within $O(\log\log n)$ rounds a.a.s. For a pseudo-band $\widehat A_h$ with $h\in\{0,\dots,L-1\}$, define
\begin{align*}
\widetilde M_{\mathrm{int},h}
&:=
\max_{v\in\widehat A_h}d_{\mathrm{int}}(v),
&
M_{\mathrm{int},h}
&:=
\max\left\{
\widetilde M_{\mathrm{int},h},
\frac{\chi}{\sqrt{\log\log n}}
\right\}.
\end{align*}

\begin{lemma}[Colouring the pseudo-bands]
\label{lem:colouring-pseudo-bands}
Fix $C_{\mathrm{int}}>2$. Suppose that, upon activation of a pseudo-band $\widehat A_h$, every vertex in $\widehat A_h$ has at least $C_{\mathrm{int}}\cdot M_{\mathrm{int},h}$ available colours. Then, a.a.s., Sequential Radial Colouring colours every pseudo-band within $O(\log\log n)$ rounds of its activation.
\end{lemma}

\begin{proof}
Fix $h\in\{0,\dots,L-1\}$ and for simplicity we write $M_{\mathrm{int}}:=M_{\mathrm{int},h}$. Since $M_{\mathrm{int}}\ge\chi/\sqrt{\log\log n}$ and $\chi=\Theta(n^{1-\alpha})$, we have $M_{\mathrm{int}}=n^{\Omega(1)}$. Set $\gamma:=C_{\mathrm{int}}-1>1$. For $i\ge0$, let $\mathcal U_i\subseteq\widehat A_h$ be the set of vertices remaining uncoloured after round $i$ of the active phase. Define the maximum uncoloured degree after round $i$ by $D_i:=\max_{x\in\widehat A_h}|N(x)\cap\mathcal U_i|$. 

Initially, $D_0\le M_{\mathrm{int}}$. Fix $i\ge1$ and condition on $\mathcal U_{i-1}$. By assumption, every vertex has at least $C_{\mathrm{int}}\cdot M_{\mathrm{int}}$ available colours when the band becomes active. During the active phase, only neighbours within the same pseudo-band can block additional colours. Since every vertex has at most $M_{\mathrm{int}}$ such neighbours, every $v\in\mathcal U_{i-1}$ has at least $C_{\mathrm{int}}\cdot M_{\mathrm{int}}-M_{\mathrm{int}}=\gamma M_{\mathrm{int}}$ available colours in round $i$.

Order the vertices of $\mathcal U_{i-1}$ by decreasing priority as $v_1^i,\dots,v_{U_{i-1}}^i$, where $U_{i-1}:=|\mathcal U_{i-1}|$, expose their colour choices in this order, and let $I_j^i$ be the indicator variable that $v_j^i$ remains uncoloured after round $i$. Then, for every $x\in\widehat A_h$,
\begin{align*}
|N(x)\cap\mathcal U_i|
=
\sum_{v_j^i\in N(x)\cap\,\mathcal U_{i-1}} I_j^i.
\end{align*}

Every vertex $v_j^i$ in the sum has at most $D_{i-1}$ higher-priority uncoloured neighbours and at least $\gamma M_{\mathrm{int}}$ available colours. Since $D_{i-1}\le D_0\le M_{\mathrm{int}}<\gamma M_{\mathrm{int}}$, Corollary~\ref{cor:sequential-stochastic-domination-sub}, applied with $a_j=\gamma M_{\mathrm{int}}$ and $b_j=D_{i-1}$, implies that, conditional on $\mathcal U_{i-1}$, $|N(x)\cap\mathcal U_i|$ is stochastically dominated by a sum of independent Bernoulli random variables with success probability $D_{i-1}/(\gamma M_{\mathrm{int}})$. Since $|N(x)\cap\mathcal U_{i-1}|\le D_{i-1}$, the random variable $|N(x)\cap\mathcal U_i|$ is stochastically dominated by $Z_i\sim\operatorname{Bin}(D_{i-1},D_{i-1}/(\gamma M_{\mathrm{int}}))$, whose conditional mean is $\nu_i:=\mathbb E[Z_i\mid\mathcal U_{i-1}]=D_{i-1}^2/(\gamma M_{\mathrm{int}})$. Hence, for every $t\ge0$ and every $x\in\widehat A_h$,
\begin{align*}
\Pr\left(|N(x)\cap\mathcal U_i|\ge t\mid\mathcal U_{i-1}\right)
\le
\Pr\left(Z_i\ge t\mid\mathcal U_{i-1}\right).
\end{align*}

A union bound over all $x\in\widehat A_h$ therefore gives $\Pr\left(D_i\ge t\mid\mathcal U_{i-1}\right) \le |\widehat A_h|\Pr\left(Z_i\ge t\mid\mathcal U_{i-1}\right)$. Fix $\lambda\in(1/\gamma,1)$, and let $\delta\in(0,1)$ be chosen later. Since $D_{i-1}\le M_{\mathrm{int}}$, we have $\nu_i\le D_{i-1}/\gamma$, so $\lambda D_{i-1}$ is a constant factor larger than $\nu_i$. Hence, for some constant $c_\lambda>0$ and every $x\in\widehat A_h$, a Chernoff bound gives
\begin{align*}
\Pr\left(|N(x)\cap\mathcal U_i|>\lambda D_{i-1}\mid\mathcal U_{i-1}\right)
\le
\Pr\left(Z_i>\lambda D_{i-1}\mid\mathcal U_{i-1}\right)
\le
\exp(-c_\lambda D_{i-1}).
\end{align*} 

We use this estimate to show that $D_i\le\delta M_{\mathrm{int}}$ within $O(1)$ rounds a.a.s. As long as $D_{i-1}>\delta M_{\mathrm{int}}$, we have $D_{i-1}=n^{\Omega(1)}$, and hence, for every fixed $c>0$ and all sufficiently large $n$, the preceding probability is at most $n^{-c-1}$. A union bound over all $x\in\widehat A_h$ therefore gives $\Pr\left(D_i>\lambda D_{i-1}\mid\mathcal U_{i-1}\right)\le n^{-c}$. Since $D_0\le M_{\mathrm{int}}$, after $t$ successful rounds we have $D_t\le\lambda^tM_{\mathrm{int}}$. 

Choose $t_0:=\left\lceil \log(1/\delta)/\log(1/\lambda)\right\rceil=O(1)$. If $D_i\le\delta M_{\mathrm{int}}$ for some $i<t_0$, then the same holds for $D_{t_0}$ by monotonicity. Otherwise, a union bound over the first $t_0$ rounds and $\lambda^{t_0}D_0\le\delta M_{\mathrm{int}}$ give $D_{t_0}\le\delta M_{\mathrm{int}}$ with failure probability at most $t_0n^{-c}=O(n^{-c})$.

We next show that, a.a.s., within $O(\log\log n)$ additional rounds, the conditional mean $\nu_i$ falls below $B\log n$ for a sufficiently large constant $B>0$. Whenever $\nu_i\ge B\log n$, a Chernoff bound implies that, for a sufficiently large constant $C_0=C_0(B,c)>1$ and every fixed $x\in\widehat A_h$,
\begin{align*}
\Pr\left(|N(x)\cap\mathcal U_i|>C_0\nu_i\mid\mathcal U_{i-1}\right)
\le
\Pr\left(Z_i>C_0\nu_i\mid\mathcal U_{i-1}\right)
\le
n^{-c-1}.
\end{align*}

A union bound over all $x\in\widehat A_h$ gives, with conditional failure probability at most $n^{-c}$,
\begin{align}
D_i
\le
C_0\nu_i
=
\frac{C_0}{\gamma}\frac{D_{i-1}^2}{M_{\mathrm{int}}}.
\label{eq:pseudoband-quadratic-decay}
\end{align}

Choose $\delta>0$ sufficiently small that $C_0\delta/\gamma\le1/2$. Starting from $D_{t_0}\le\delta M_{\mathrm{int}}$, we claim that, for every $r\ge0$, as long as $\nu_{t_0+j}\ge B\log n$ for all $j\in\{1,\dots,r\}$, $D_{t_0+r}\le \delta M_{\mathrm{int}}\,2^{-(2^r-1)}$ with failure probability at most $(t_0+r)n^{-c}$. The case $r=0$ is immediate. If the bound holds for $r$ and $\nu_{t_0+r+1}\ge B\log n$, then \eqref{eq:pseudoband-quadratic-decay} gives
\begin{align*}
D_{t_0+r+1}
\le
\frac{C_0}{\gamma}\frac{D_{t_0+r}^2}{M_{\mathrm{int}}}
\le
\delta M_{\mathrm{int}}\,2^{-(2^{r+1}-1)}
\end{align*}
with conditional failure probability at most $n^{-c}$. A union bound shows that the bound at round $t_0+r+1$ holds with failure probability at most $(t_0+r+1)n^{-c}$, completing the induction.

Choose $r=O(\log\log n)$ sufficiently large that $\left(\delta M_{\mathrm{int}}\,2^{-(2^r-1)}\right)^2/(\gamma M_{\mathrm{int}})<B\log n$.
If $\nu_i\ge B\log n$ for every $i\in\{t_0+1,\dots,t_0+r\}$, then the induction bound holds up to round $t_0+r$ with failure probability at most $(t_0+r)n^{-c}$ and yields
\begin{align*}
\nu_{t_0+r+1}
\le
\frac{\left(\delta M_{\mathrm{int}}\,2^{-(2^r-1)}\right)^2}{\gamma M_{\mathrm{int}}}
<
B\log n.
\end{align*}

Otherwise, $\nu_i<B\log n$ already holds for some $i\in\{t_0+1,\dots,t_0+r\}$. Thus, with failure probability at most $(t_0+r)n^{-c}=o(1)$, there exists a round $i\le t_0+r+1$ such that $\nu_i<B\log n$.

Let $i$ be the first round with $\nu_i<B\log n$. Conditional on $\mathcal U_{i-1}$, the random variable $|N(x)\cap\mathcal U_i|$ is stochastically dominated by $Z_i$ for every $x\in\widehat A_h$. Choose $C_1>eB$ sufficiently large that $C_1\log(C_1/(eB))\ge c+1$. Then, by Chernoff's bound~\cite[Theorem~1.10.1]{doerr2020probabilistic},
\begin{align*}
\Pr\left(|N(x)\cap\mathcal U_i|\ge C_1\log n\mid\mathcal U_{i-1}\right)
\le
\Pr\left(Z_i\ge C_1\log n\mid\mathcal U_{i-1}\right)
\le
\left(\frac{e\nu_i}{C_1\log n}\right)^{C_1\log n}
\le
n^{-c-1},
\end{align*}
where the last inequality follows from $\nu_i<B\log n$ and the choice of $C_1$. A union bound over all $x\in\widehat A_h$ gives $D_i\le C_1\log n=O(\log n)$ with conditional failure probability at most $n^{-c}$.


Finally, suppose that $D_i\le C_1\log n$. Since the uncoloured sets only shrink, every uncoloured vertex remains uncoloured in each subsequent round with conditional probability at most $C_1\log n/(\gamma M_{\mathrm{int}})$. Hence, for any fixed $q\ge1$, a union bound over all vertices gives
\begin{align*}
\Pr\left(\mathcal U_{i+q}\neq\emptyset\mid\mathcal U_i\right)
\le
O(n)\left(\frac{C_1\log n}{\gamma M_{\mathrm{int}}}\right)^q.
\end{align*}

Since $M_{\mathrm{int}}=n^{\Omega(1)}$, choosing a sufficiently large constant $q$ makes this probability $o(1/\log\log n)$. Thus, for a fixed pseudo-band $\widehat A_h$, the total failure probability is at most \begin{align*}
    O(\log\log n)n^{-c}+n\left(C_1\log n/(\gamma M_{\mathrm{int}})\right)^q =o(1/\log\log n) 
\end{align*}
for sufficiently large constants $c$ and $q$. Since $L=O(\log\log n)$, a union bound over all $h\in\{0,\dots,L-1\}$ shows that every pseudo-band is coloured a.a.s. within $O(\log\log n)$ rounds of its activation.
\end{proof}

\subsection{Colouring the pseudo-outer region}
\label{subsec:colouring-pseudo-outer-region}

It remains to analyse the pseudo-outer region $\widehat{\mathcal O}$. Unlike the intermediate pseudo-bands, no separate internal-degree estimate is required: by construction, every vertex in $\widehat{\mathcal O}$ has at most polylogarithmic degree. Consequently, a palette of size $\chi$ leaves all but a negligible fraction of the colours available at each vertex, which allows the region to be coloured in constantly many rounds a.a.s.

\begin{lemma}[Pseudo-outer region degree bound]
\label{lem:pseudo-outer-region-degree-bound}
Every vertex in $\widehat{\mathcal O}$ has degree $O(\log^2 n)$.
\end{lemma}

\begin{proof}
By the construction, the boundary radius $r_{\mathrm{out}}=R-4\log\log n+O(1)$. Hence, Lemma~\ref{lem:uniform-expected-degree-asymptotic} gives $\bar d(r_{\mathrm{out}})=O(\log^2 n)$. Fix $v\in\widehat{\mathcal O}$. If $\deg(v)\le\bar d(R)$, then $\deg(v)=O(1)$ by Lemma~\ref{lem:boundary-expected-degree}. Otherwise, since $\widehat r(v)>r_{\mathrm{out}}$, we have $\deg(v)=\bar d(\widehat r(v))<\bar d(r_{\mathrm{out}})=O(\log^2 n)$ by Lemma~\ref{lem:expected-degree-monotone}.
\end{proof}

\begin{lemma}[Colouring the pseudo-outer region]
\label{lem:colouring-pseudo-outer-region}
If $|\Psi|\ge\chi$, then, a.a.s., Sequential Radial Colouring colours $\widehat{\mathcal O}$ within $O(1)$ rounds.
\end{lemma}

\begin{proof}
By Lemma~\ref{lem:pseudo-outer-region-degree-bound}, every vertex in $\widehat{\mathcal O}$ has degree $O(\log^2 n)$. Since $|\Psi|\ge\chi=\Theta(n^{1-\alpha})$, every uncoloured vertex has $\Theta(n^{1-\alpha})$ available colours throughout the active phase and at most $O(\log^2 n)$ higher-priority uncoloured neighbours. Its conditional probability of remaining uncoloured in any round is therefore
$O\left(\log^2 n/n^{1-\alpha}\right)$.

Choose a constant integer $q$ such that $q(1-\alpha)>1$. The probability that a fixed vertex remains uncoloured for $q$ rounds is at most
$O\left((\log^2 n/n^{1-\alpha})^q\right)$. Hence, a union bound over all vertices gives
\begin{align*}
\Pr\left(\mathcal U_q\neq\emptyset\right)
\le
O\left(n \left(\frac{\log^2 n}{n^{1-\alpha}}\right)^q
\right)
=
o(1),
\end{align*}
where $\mathcal U_q$ denotes the set of vertices in $\widehat{\mathcal O}$ remaining uncoloured after $q$ rounds. Hence $\widehat{\mathcal O}$ is coloured in $O(1)$ rounds a.a.s.
\end{proof}

\subsection{Proof of the near-chromatic colour theorem}
\label{subsec:proof-near-chromatic-theorem}

We now prove the main theorem for Sequential Radial Colouring with a near-degeneracy palette. For every fixed $\varepsilon>0$, the algorithm uses $(1+\varepsilon)\kappa$ colours and terminates within $O((\log\log n)^2)$ rounds a.a.s. This gives a $(1+\varepsilon)$-approximation of the degeneracy. Since $\kappa\le((4/3)^\alpha+o(1))\chi$ by Lemma~\ref{lem:degeneracy-bounds}, the same palette is also a constant-factor approximation of the chromatic number, with a worst-case factor ranging from $\sqrt{4/3}$ to $4/3$ depending on the parameter $\alpha$.

\begin{theorem}[SRC with a near-degeneracy palette]
\label{thm:src-near-degeneracy-colours}
Let $G\sim \mathcal G(n,\alpha,C)$ be a typical hyperbolic random graph with
$\alpha\in(1/2,1)$. For every fixed $\varepsilon>0$, Sequential Radial Colouring with palette size $|\Psi|=\left\lceil(1+\varepsilon)\kappa\right\rceil$ colours $G$ in $O((\log\log n)^2)$ rounds a.a.s.
\end{theorem}

\begin{proof}
Since $\chi\le\kappa+1$ and $\chi=\Theta(n^{1-\alpha})$ a.a.s., we have a.a.s., for all sufficiently large $n$, $|\Psi|\ge(1+\varepsilon/2)\chi$.

First, the algorithm processes the pseudo-clique $\widehat{\mathcal C}$. Since $|\Psi|\ge(1+\varepsilon/2)\chi$, Corollary~\ref{cor:near-tight-pseudo-clique-colouring} implies that $\widehat{\mathcal C}$ is coloured in $O(\log\log n)$ rounds a.a.s.

Next, consider the intermediate pseudo-bands. For each pseudo-band
$\widehat A_h$, recall
\begin{align*}
\widetilde M_{\mathrm{int},h}
&:=
\max_{v\in\widehat A_h}d_{\mathrm{int}}(v),
&
M_{\mathrm{int},h}
&:=
\max\left\{
\widetilde M_{\mathrm{int},h},
\frac{\chi}{\sqrt{\log\log n}}
\right\}.
\end{align*}

To apply Lemma~\ref{lem:colouring-pseudo-bands}, it remains to show that, upon activation of $\widehat A_h$, every vertex in $\widehat A_h$ has at least $C_{\mathrm{int}}\cdot M_{\mathrm{int},h}$ available colours.

We first consider the refined bands, where $x_h<x_{\mathrm{crit}}$. By Corollary~\ref{cor:previous-degree-degeneracy-refined}, every vertex in such a band has at most $(1+o(1))\kappa$ neighbours in previously processed regions. Hence, after removing colours used by previous regions, every vertex still has at least
\begin{align*}
|\Psi^1(v)|
\ge
(1+\varepsilon)\kappa-(1+o(1))\kappa
\ge
\frac{\varepsilon}{2}\kappa
\geq 
C_{\mathrm{int}}\chi/\sqrt{\log\log n}
=
C_{\mathrm{int}}M_{\mathrm{int},h}
\end{align*}
available colours, for all sufficiently large $n$, where we used that $\chi\leq \kappa+1$ and Lemma~\ref{lem:internal-degree-bands}, which gives $\widetilde M_{\mathrm{int},h}\leq \chi/\sqrt{\log\log n}$ in the refined range.

Now consider a band with $x_h\ge x_{\mathrm{crit}}$. By Lemma~\ref{lem:degree-previous-pseudo-bands}, every vertex in $\widehat A_h$ has at most $\chi$ neighbours in previously processed regions. Thus
\begin{align*}
|\Psi^1(v)|
\ge
|\Psi|-\chi
\ge
(1+\varepsilon/2)\chi-\chi
=
\frac{\varepsilon}{2}\chi.
\end{align*}

Set $\eta:=\varepsilon/(2C_{\mathrm{int}})$ and choose $k\le k_0(\eta)$ as in Lemma~\ref{lem:internal-degree-bands}. Then, w.h.p., $\widetilde M_{\mathrm{int},h}\le\eta\chi$ simultaneously for all such bands, and hence also $M_{\mathrm{int},h}\le\eta\chi$ implying that $|\Psi^1(v)|\ge C_{\mathrm{int}}M_{\mathrm{int},h}$.

Thus, the assumptions of Lemma~\ref{lem:colouring-pseudo-bands} hold for every intermediate pseudo-band. Since the $L=O(\log\log n)$ bands are processed sequentially and each requires $O(\log\log n)$ rounds, all intermediate pseudo-bands are coloured in $O((\log\log n)^2)$ rounds a.a.s.

Finally, the algorithm processes the pseudo-outer region $\widehat{\mathcal O}$. By Lemma~\ref{lem:pseudo-outer-region-degree-bound}, every vertex in $\widehat{\mathcal O}$ has degree at most $O(\log^2 n)$. Since $|\Psi|\ge(1+\varepsilon/2)\chi$ and $\chi=\Theta(n^{1-\alpha})$ a.a.s., we have $|\Psi|-O(\log^2 n)\ge\chi$ for all sufficiently large $n$. Thus Lemma~\ref{lem:colouring-pseudo-outer-region} applies, and $\widehat{\mathcal O}$ is coloured in $O(1)$ rounds a.a.s.

Combining the pseudo-clique, the intermediate pseudo-bands, and the pseudo-outer region, all vertices are coloured in $O((\log\log n)^2)$ rounds a.a.s.
\end{proof}

There are two possible ways to interpret Theorem~\ref{thm:src-near-degeneracy-colours}, depending on the relation between \(\kappa\) and \(\chi\), as summarized in \Cref{tab:kappa-chi-translation}. If \(\kappa=(1+o(1))\chi\), then the \((1+\eps)\kappa\)-palette can be read as a \((1+\eps+o(1))\chi\)-palette, so SRC gives an arbitrarily good constant approximation to the chromatic number. If instead, as conjectured in the introduction, \(\kappa\) and \(\chi\) differ by a genuine constant factor, then \(\kappa\) is itself the obstruction: the theorem gives an arbitrarily close approximation to \(\kappa\), but not necessarily to \(\chi\). Even in this case, the known comparison between \(\kappa\) and \(\chi\) gives a constant-factor approximation to the chromatic number, with factor between \(\sqrt{4/3}\) and \(4/3\).

\begin{table}[t]
\centering
\renewcommand{\arraystretch}{1.25}
\begin{tabular}{@{}ll@{}}
\toprule
\textbf{Relation} & \textbf{Palette in terms of \(\chi\)} \\
\midrule
\(\kappa=(\delta+o(1))\chi\), for some fixed \(\delta>1\)
& \((\delta+\eps)\chi\), for any fixed \(\eps>0\) \\
\addlinespace
\(\kappa=(1+o(1))\chi\)
& \((1+\eps)\chi\), for any fixed \(\eps>0\) \\
\bottomrule
\end{tabular}
\caption{How the near-degeneracy guarantee translates into a statement in terms of the chromatic number.}
\label{tab:kappa-chi-translation}
\end{table}

Theorem~\ref{thm:src-near-degeneracy-colours} implies the first item of Theorem~\ref{thm:summary}. Indeed, the theorem is stated for the palette size $\lceil(1+\varepsilon)\kappa\rceil$, and larger palettes can only increase the set of available colours at every vertex. Since the structural assumptions defining a typical HRG hold a.a.s., SRC therefore colours $G$ in $O((\log\log n)^2)$ rounds a.a.s. whenever $|\Psi|\ge(1+\varepsilon)\kappa$.

\section{Parallel Radial Colouring with \texorpdfstring{$O(\chi\log\log n)$}{O(χ(G) log log n)} colours}

The goal of this section is to reduce the round complexity by activating all pseudo-regions simultaneously. To avoid conflicts between different regions, Parallel Radial Colouring assigns each pseudo-region a distinct colour set. Consequently, the algorithm uses $O(\chi\log\log n)$ colours and terminates within $O(\log\log n)$ rounds a.a.s.

The input is a graph $G\sim\mathcal G(n,\alpha,C)$ with $\alpha\in(1/2,1)$, together with a colour palette $\Psi$.\footnote{The palette $\Psi$ is provided as input. In the parallel variant, its partition into subpalettes is also fixed globally.} Let $\varepsilon_{n}=o(1)$ be a deterministic sequence for which Corollary~\ref{cor:near-tight-pseudo-clique-colouring} applies. Assign pairwise disjoint palettes $\Psi_{\mathcal C},\Psi_0,\dots,\Psi_{L-1},\Psi_{\mathcal O}$ to $\widehat{\mathcal C},\widehat A_0,\dots,\widehat A_{L-1},\widehat{\mathcal O}$, respectively, with $|\Psi_{\mathcal C}|=\lceil(1+\varepsilon_{n})\chi\rceil$ and $|\Psi_h|=|\Psi_{\mathcal O}|=\chi$ for every $h\in\{0,\dots,L-1\}$. Since the number of pseudo-regions is $O(\log\log n)$ by Corollary~\ref{cor:number-of-radial-bands}, the combined palette size is $O(\chi\log\log n)$.

\begin{theorem}[Parallel Radial Colouring]
\label{th:parallel-radial-colouring}
Let $G\sim \mathcal G(n,\alpha,C)$ be a typical hyperbolic random graph with $\alpha\in(1/2,1)$. Then Parallel Radial Colouring colours $G$ in $O(\log\log n)$ rounds a.a.s. using $O(\chi\log\log n)$ colours.
\end{theorem}

\begin{proof}
Since the palettes are pairwise disjoint, vertices in distinct pseudo-regions cannot create colour conflicts. Moreover, colours used outside a pseudo-region impose no restrictions within it. Hence all pseudo-regions may be processed simultaneously.

By Corollary~\ref{cor:near-tight-pseudo-clique-colouring}, there exists a deterministic sequence $\varepsilon_{n}=o(1)$ such that a palette of size $\lceil(1+\varepsilon_{n})\chi\rceil$ colours $\widehat{\mathcal C}$ a.a.s. within $O(\log\log n)$ rounds.

For every \(h\in\{0,\dots,L-1\}\), all colours in \(\Psi_h\) are available upon activation of \(\widehat A_h\). Choose \(\eta>0\) sufficiently small and then choose the width parameter \(k\le k_0(\eta)\) as in Lemma~\ref{lem:internal-degree-bands}. Then \(M_{\mathrm{int},h}\le \chi/C_{\mathrm{int}}\) for every intermediate band. Hence \(|\Psi_h|=\chi\ge C_{\mathrm{int}}M_{\mathrm{int},h}\), so Lemma~\ref{lem:colouring-pseudo-bands} applies and each intermediate pseudo-band is coloured within \(O(\log\log n)\) rounds a.a.s. Finally, Lemma~\ref{lem:colouring-pseudo-outer-region} implies that $\widehat{\mathcal O}$ is coloured within $O(1)$ rounds a.a.s. A union bound over all the pseudo-regions shows that Parallel Radial Colouring terminates within $O(\log\log n)$ rounds a.a.s.
Corollary~\ref{cor:number-of-radial-bands} gives $O(\log\log n)$ pseudo-regions, and hence the total number of colours is
\begin{align*}
|\Psi_{\mathcal C}|+\sum_{h=0}^{L-1}|\Psi_h|+|\Psi_{\mathcal O}|
&=
(L+2+o(1))\chi
=
O(\chi\log\log n),
\end{align*}
which proves the desired claim.
\end{proof}

Thus, compared with Sequential Radial Colouring, Parallel Radial Colouring trades an additional factor of $O(\log\log n)$ in the palette size for reducing the round complexity from $O((\log\log n)^2)$ to $O(\log\log n)$ a.a.s.

Theorem~\ref{th:parallel-radial-colouring} gives the second item of Theorem~\ref{thm:summary}. Since $\chi\le\kappa+1$ and $\kappa=\Theta(\chi)$ a.a.s for HRGs, by Lemmas~\ref{lem:degeneracy-bounds} and~\ref{lem:chromatic-number-bound}, its palette bound $O(\chi\log\log n)$ is also $O(\kappa\log\log n)$. Conversely, if $|\Psi|\ge c\kappa\log\log n$ for any fixed constant $c>0$, we split the palette into $\Theta(\log\log n)$ subpalettes of size $\Theta(\chi)$; if this gives fewer subpalettes than pseudo-regions, we process a constant number of batches sequentially. Each batch takes $O(\log\log n)$ rounds by the same PRC argument, so the total time remains $O(\log\log n)$ a.a.s.

\section{Sequential Radial Colouring with \texorpdfstring{$O(\chi^{1+\varepsilon})$}{O(χ(G) raised to (1 + ε))} colours}
\label{sec:seq-col-constant-rounds}

In this section we analyse the polynomial-slack version of Sequential Radial Colouring, where the palette size is $|\Psi|=\lceil\chi^{1+\varepsilon}\rceil$ for a fixed constant $\varepsilon>0$. The additional colour space allows us to replace the fine/coarse decomposition from \Cref{subsec:sequential-radial-colouring} by a constant-band decomposition with only constantly many pseudo-regions. We define this decomposition and the corresponding colouring process in the next subsection. In \Cref{subsec:constant-pseudo-region-accuracy}, we show that the resulting pseudo-regions have the same asymptotic sizes as their true radial counterparts.

The colouring analysis again proceeds region by region. The inner pseudo-region is handled first in \Cref{subsec:constant-colouring-inner-region}, using the general constant-round RCTDEG guarantee from the clique-colouring section\footnote{Here, the inner region consists of the pseudo-clique together with a small radial buffer; the buffer ensures that vertices in the following bands already have sufficiently small higher-priority degree.}. The intermediate pseudo-bands are analysed in \Cref{subsec:constant-colouring-pseudo-bands}. There we control the three relevant bottlenecks in an active band: inherited restrictions from already coloured neighbours, internal restrictions from neighbours inside the active band, and priority conflicts caused by higher-priority neighbours choosing the same colour in a given round. The first two remove only $o(|\Psi|)$ colours, while the third gives each vertex only a polynomially small failure probability per round. Finally, \Cref{subsec:constant-colouring-pseudo-outer-region} treats the pseudo-outer region.

Combining these region-wise bounds in \Cref{subsec:proof-constant-round-colouring} shows that the constantly many pseudo-regions can be processed sequentially for a constant number of rounds each, yielding an $O(1)$-round colouring algorithm a.a.s.

\subsection{Sequential Radial Colouring algorithm}
\label{subsec:constant-sequential-radial-colouring-2}

The input is again a graph $G\sim\mathcal G(n,\alpha,C)$ with $\alpha\in(1/2,1)$, together with a colour palette $\Psi$, now of size $|\Psi|=\lceil\chi^{1+\varepsilon}\rceil$ for a fixed constant $\varepsilon>0$. The algorithm follows the Sequential Radial Colouring framework from Section~\ref{subsec:sequential-radial-colouring}: vertices estimate their radii, are assigned to pseudo-regions, and these regions are processed sequentially from the centre outwards using RCTDEG~\cite{mausruff2026distributed}. With palette size $|\Psi|=\lceil\chi^{1+\varepsilon}\rceil$, the decomposition of the disk is replaced by a constant number of larger regions, each processed for a constant number of rounds. This yields an $O(1)$-round colouring algorithm a.a.s.

\subsubsection*{Pseudo-regions}

We use the radius estimates $\widehat r(v)$ and the concentration interval $[r_{\mathrm L},r_{\mathrm U}]$ introduced in Section~\ref{subsec:sequential-radial-colouring}.  Let \(c_{\mathrm{in}}>0\) be a sufficiently large constant, and define \(r_{\mathrm{in}}:=R/2+c_{\mathrm{in}}\log\log n\) and \( r_{\mathrm{tail}}:=2\alpha\log n\). The innermost pseudo-region is \(\widehat{\mathcal I}:=\{v\in V(G):\widehat r(v)\le r_{\mathrm{in}}\}\).

The interval $(r_{\mathrm{in}},r_{\mathrm{tail}}]$ is partitioned into a constant number of intermediate pseudo-bands. Define $m:=\lceil\alpha(2\alpha-1)/(\varepsilon(1-\alpha))\rceil$\footnote{Unlike in the previous decomposition, the final band has width $\Theta(\log n)$, so no truncation case is needed.}, set $r_0:=r_{\mathrm{in}}$, and, for $i\in\{1,\dots,m\}$, let
\begin{align*}
r_i
:=
\min\left\{
r_{\mathrm{in}}+i\frac{\varepsilon(1-\alpha)}{\alpha}\log n,
r_{\mathrm{tail}}
\right\}.
\end{align*}

For $i\in\{1,\dots,m\}$, define $\widehat A_i:=\{v\in V(G):r_{i-1}<\widehat r(v)\le r_i\}$. Finally, the pseudo-outer region is $\widehat{\mathcal O}:=\{v\in V(G):\widehat r(v)>r_{\mathrm{tail}}\}$. The corresponding true regions $\mathcal I,A_1,\dots,A_m,\mathcal O$ are defined analogously using the true radius $r(v)$.

\subsubsection*{Colouring process}

The pseudo-regions are processed in increasing radial order:
$\widehat{\mathcal I},\widehat A_1,\dots,\widehat A_m,\widehat{\mathcal O}$.
Each region is active for
\begin{align*}
T
:=
1+\max\left\{
\left\lceil\frac{1}{\varepsilon(1-\alpha)}\right\rceil,
\left\lceil\frac{1}{\varepsilon}+m\right\rceil
\right\}
\end{align*}
rounds. The $j$th pseudo-region is activated after $jT$ rounds and then runs RCTDEG~\cite{mausruff2026distributed} for $T$ rounds. Since $m=O(1)$ and $T=O(1)$, all $m+2$ pseudo-regions are processed within $(m+2)T=O(1)$ rounds. Thus, a.a.s., the enlarged palette reduces the round complexity of Sequential Radial Colouring to $O(1)$.

\subsection{Accuracy of pseudo-regions}
\label{subsec:constant-pseudo-region-accuracy}

We next show that each pseudo-region in this constant-band decomposition has
asymptotically the same size as its corresponding true radial region. The comparison follows as in Section~\ref{subsec:pseudo-band-size-estimates}: the radius estimate is accurate up to $o(1)$ at all region boundaries, which lie in $[r_{\mathrm L},r_{\mathrm U}]$.

\subsubsection*{Size of pseudo-inner region}

The inner pseudo-region has only one relevant boundary, namely $r_{\mathrm{in}}$, so the radius-estimate accuracy sandwiches it between two centred balls with radii $r_{\mathrm{in}}\pm\sqrt{\delta_n}$.

\begin{lemma}[Size of the pseudo-inner region]
\label{lem:size-pseudo-inner-region}
With high probability, $|\widehat{\mathcal I}|=(1\pm o(1))\mathbb E[|\mathcal I|]$.
\end{lemma}

\begin{proof}
Recall that $\delta_n=\sqrt{48/(K\log n)}=o(1)$, and set $\mathcal I_-:=V(G)\cap B_0(r_{\mathrm{in}}-\sqrt{\delta_n})$ and $\mathcal I_+:=V(G)\cap B_0(r_{\mathrm{in}}+\sqrt{\delta_n})$. We first show that, w.h.p., $\mathcal I_-\subseteq\widehat{\mathcal I}\subseteq\mathcal I_+$.

For the lower inclusion, let $v\in\mathcal I_-$. If $r(v)\le r_{\mathrm L}$, then $v\in\widehat{\mathcal I}$ by Lemma~\ref{lem:assignment-inner-cutoff-vertices}. Otherwise, $r(v)\in[r_{\mathrm L},r_{\mathrm{in}}-\sqrt{\delta_n}]$, and Lemma~\ref{lem:radius-estimate-accuracy} gives $\widehat r(v)\le r(v)+\sqrt{\delta_n}\le r_{\mathrm{in}}$, so again $v\in\widehat{\mathcal I}$. For the upper inclusion, let $v\in\widehat{\mathcal I}$. By Lemma~\ref{lem:assignment-outer-cutoff-vertices}, w.h.p. we have $r(v)<r_{\mathrm U}$. If $r(v)>r_{\mathrm{in}}+\sqrt{\delta_n}$, then $r(v)\in[r_{\mathrm L},r_{\mathrm U}]$, and Lemma~\ref{lem:radius-estimate-accuracy} gives $\widehat r(v)>r_{\mathrm{in}}$, a contradiction. Hence $v\in\mathcal I_+$.

Since $\sqrt{\delta_n}=o(1)$, Lemma~\ref{lem:centered-ball-measure} implies $\mathbb E[|\mathcal I_-|]\ge(1-o(1))\mathbb E[|\mathcal I|]$ and $\mathbb E[|\mathcal I_+|]\le(1+o(1))\mathbb E[|\mathcal I|]$. Moreover, $\mathbb E[|\mathcal I|]=\Theta(n^{1-\alpha}(\log n)^{\alpha \cdot c_{\mathrm{in}}})$, so Poisson concentration for $|\mathcal I_-|$ and $|\mathcal I_+|$ yields $|\widehat{\mathcal I}|=(1\pm o(1))\mathbb E[|\mathcal I|]$ w.h.p.
\end{proof}

\subsubsection*{Size of pseudo-bands}

\begin{corollary}[Sizes of the pseudo-bands]
\label{cor:constant-pseudoband-sizes}
With high probability, $|\widehat A_i|=(1\pm o(1))\mathbb E[|A_i|]$ for all $i\in\{1,\dots,m\}$.
\end{corollary}

\begin{proof}
For every $i\in\{1,\dots,m\}$, the pseudo-band $\widehat A_i$ corresponds to the radial interval $(r_{i-1},r_i]$, whose width is $\Theta(\log n)=\omega(\sqrt{\delta_n})$. Moreover, $\mathbb E[|A_i|]=n^{\Omega(1)}$. Hence, Lemma~\ref{lem:pseudoband-size} applies to each band, and a union bound over the constant number of bands proves the claim.
\end{proof}

\subsubsection*{Size of pseudo-outer region}

\begin{lemma}[Size of the pseudo-outer region]
\label{lem:size-pseudo-outer-region-constant}
With high probability, $|\widehat{\mathcal O}|=(1\pm o(1))\mathbb E[|\mathcal O|]$.
\end{lemma}

\begin{proof}

Define $\mathcal O_-:=V(G)\setminus B_0(r_{\mathrm{tail}}+\sqrt{\delta_n})$ and $\mathcal O_+:=V(G)\setminus B_0(r_{\mathrm{tail}}-\sqrt{\delta_n})$. We first show that, w.h.p., $\mathcal O_-\subseteq\widehat{\mathcal O}\subseteq\mathcal O_+$. Let $v\in\mathcal O_-$. If $r(v)\le r_{\mathrm U}$, then Lemma~\ref{lem:radius-estimate-accuracy} gives $\widehat r(v)>r(v)-\sqrt{\delta_n}\ge r_{\mathrm{tail}}$. If $r(v)>r_{\mathrm U}$, then Lemma~\ref{lem:assignment-outer-cutoff-vertices} gives $\widehat r(v)>R-4\log\log n+O(1)>r_{\mathrm{tail}}$ for all sufficiently large $n$. Hence, $v\in\widehat{\mathcal O}$.

Conversely, let $v\in\widehat{\mathcal O}$. Lemma~\ref{lem:assignment-inner-cutoff-vertices} excludes $r(v)\le r_{\mathrm L}$, since those vertices have $\widehat r(v) \leq R/2$ w.h.p. Thus, if $r(v)\le r_{\mathrm{tail}}-\sqrt{\delta_n}$, then $r(v)\in[r_{\mathrm L},r_{\mathrm U}]$, and Lemma~\ref{lem:radius-estimate-accuracy} gives $\widehat r(v)<r(v)+\sqrt{\delta_n}\le r_{\mathrm{tail}}$, a contradiction. Therefore $v\in\mathcal O_+$.

By Lemma~\ref{lem:centered-ball-measure}, $\mathbb E[|\mathcal O|]=n(1-\mu(B_0(r_{\mathrm{tail}})))=(1-o(1))n$. Since $\sqrt{\delta_n}=o(1)$, the same lemma gives $\mathbb E[|\mathcal O_-|]\ge(1-o(1))\mathbb E[|\mathcal O|]$ and $\mathbb E[|\mathcal O_+|]\leq (1+o(1))\mathbb E[|\mathcal O|]$. Moreover, $|\mathcal O_-|$ and $|\mathcal O_+|$ are Poisson random variables with expectations $\Theta(n)$, so Chernoff concentration and the inclusions above yield $|\widehat{\mathcal O}|=(1\pm o(1))\mathbb E[|\mathcal O|]$ w.h.p.
\end{proof}

\subsubsection*{Comparison with true regions}

\begin{lemma}[Concentration of true regions]
\label{lem:constant-true-region-concentration}
With high probability, $|\mathcal I|=(1\pm o(1))\mathbb E[|\mathcal I|]$, $|A_i|=(1\pm o(1))\mathbb E[|A_i|]$ for all $i\in\{1,\dots,m\}$, and $|\mathcal O|=(1\pm o(1))\mathbb E[|\mathcal O|]$.
\end{lemma}

\begin{proof}
The regions $\mathcal I,A_1,\dots,A_m,\mathcal O$ are deterministic radial regions, so their sizes are Poisson random variables with means equal to their expectations. By Lemma~\ref{lem:centered-ball-measure}, $\mathbb E[|\mathcal I|]=\Theta(n^{1-\alpha}(\log n)^{\alpha \cdot c_{\mathrm{in}}})$, $\mathbb E[|A_i|]=n^{\Omega(1)}$ for all $i\in\{1,\dots,m\}$, and $\mathbb E[|\mathcal O|]=(1-o(1))n$. Chernoff concentration gives the desired estimate for each region, and a union bound over the $m+2=O(1)$ regions proves the claim.
\end{proof}

\begin{corollary}[Pseudo-region size comparison]
\label{cor:constant-pseudo-region-size-comparison}
With high probability,
$|\widehat{\mathcal I}|=(1\pm o(1))|\mathcal I|$,
$|\widehat A_i|=(1\pm o(1))|A_i|$ for all $i\in\{1,\dots,m\}$, and
$|\widehat{\mathcal O}|=(1\pm o(1))|\mathcal O|$.
\end{corollary}


Recall the event \(\mathcal E\). For the constant pseudo-decomposition, on \(\mathcal E\), for every \(i\in\{1,\dots,m\}\), every vertex in a pseudo-region preceding \(\widehat A_i\) has true radius at most \(r_{i-1}+\sqrt{\delta_n}\), while \(\widehat A_i\subseteq A_i^+\), where
\(A_i^+:=\{v\in V(G):r_{i-1}-\sqrt{\delta_n}<r(v)\le r_i+\sqrt{\delta_n}\}\).

Thus, also in this decomposition, all true-region sizes concentrate around their expectations, while the corresponding pseudo-regions have asymptotically equal sizes. Recall that we call an HRG \emph{typical} if it satisfies all structural properties that we have shown to hold a.a.s.; from this point on, we fix such a typical HRG and analyse only the randomness of the colouring algorithm.

\subsection{Colouring the inner pseudo-region}
\label{subsec:constant-colouring-inner-region}

Recall that $|\Psi|=\lceil\chi^{1+\varepsilon}\rceil$ for a fixed constant $\varepsilon>0$. The first active region is the inner pseudo-region $\widehat{\mathcal I}$, which contains the pseudo-clique together with the radial buffer up to $r_{\mathrm{in}}$. Since $|\widehat{\mathcal I}|=O(n^{1-\alpha}(\log n)^{\alpha \cdot c_{\mathrm{in}}})$ a.a.s., this palette is polynomially larger than the size of the whole inner pseudo-region. We can therefore colour $\widehat{\mathcal I}$ in constantly many rounds by applying the general constant-round RCTDEG guarantee proved in \Cref{subsec:clique-upper-bounds}.

\begin{corollary}[Constant-round colouring of the inner pseudo-region]
\label{cor:constant-round-colouring-inner-pseudo-region}
For every fixed $t>1/\varepsilon$, the first phase of SRC a.a.s. colours $\widehat{\mathcal I}$ within $t$ rounds.
\end{corollary}

\begin{proof}
During the first phase of SRC, no colours have been fixed yet and the active region is precisely $\widehat{\mathcal I}$. Thus SRC restricted to $\widehat{\mathcal I}$ is exactly RCTDEG on the induced graph $G[\widehat{\mathcal I}]$. By Corollary~\ref{cor:constant-pseudo-region-size-comparison} and Lemma~\ref{lem:constant-true-region-concentration}, w.h.p., $|\widehat{\mathcal I}| = O\bigl(n^{1-\alpha}(\log n)^{\alpha \cdot c_{\mathrm{in}}}\bigr)$. Moreover, $\chi=\Theta(n^{1-\alpha})$ a.a.s. by Lemma~\ref{lem:chromatic-number-bound}. Fix $t>1/\varepsilon$, and choose a constant $\varepsilon'\in(1/t,\varepsilon)$. Then
\begin{align*}
|\widehat{\mathcal I}|^{1+\varepsilon'}
=
O\left(
n^{(1-\alpha)(1+\varepsilon')}
(\log n)^{\alpha \cdot c_{\mathrm{in}}(1+\varepsilon')}
\right)
=
o\left(\chi^{1+\varepsilon}\right).
\end{align*}

Hence, for all sufficiently large $n$, $|\Psi| = \left\lceil \chi^{1+\varepsilon}\right\rceil \ge |\widehat{\mathcal I}|^{1+\varepsilon'}$. Corollary~\ref{cor:constant-round-graph-colouring}, applied to $G[\widehat{\mathcal I}]$ with parameter $\varepsilon'$, implies that RCTDEG colours $G[\widehat{\mathcal I}]$ within $t$ rounds a.a.s., because $t>1/\varepsilon'$. This proves the claim.
\end{proof}

\subsection{Colouring the pseudo-bands}
\label{subsec:constant-colouring-pseudo-bands}

Before proving the colouring lemma for the pseudo-bands, we isolate the three degree bounds governing their active phases. For a vertex in $\widehat A_i$, the degree to previous pseudo-regions bounds the number of colours already removed when the band becomes active, the internal degree bounds the number of additional colours that may be removed during the active phase, and the higher-priority degree controls the probability that the vertex remains uncoloured in a given round.

\subsubsection*{Degree to previous pseudo-regions}

Recall that, for $i\in\{1,\dots,m\}$, $\widehat P_i$ denotes the union of the pseudo-regions processed before $\widehat A_i$, and that $d_{\mathrm{pre}}(u):=|N(u)\cap\widehat P_i|$ for every $u\in\widehat A_i$.

\begin{lemma}[Degree to previous pseudo-regions]
\label{lem:constant-prev-degree-pseudo-bands}
W.h.p., for every $i\in\{1,\dots,m\}$ and every $u\in\widehat A_i$, we have $d_{\mathrm{pre}}(u)=o(n^{1-\alpha})$.
\end{lemma}

\begin{proof}
Fix \(i\in\{1,\dots,m\}\) and \(u\in\widehat A_i\). Recall that
\(\widehat A_i=\{v:r_{i-1}<\widehat r(v)\le r_i\}\), and write
\(x_{i-1}:=r_{i-1}-R/2\) and \(x:=r(u)-R/2\).

We work on the high-probability radius-accuracy event \(\mathcal E\). Every vertex in the previously processed pseudo-regions \(\widehat P_i\) has true radius at most \(r_{i-1}+\sqrt{\delta_n}\): indeed, vertices with radius at least \(r_{\mathrm U}\) are assigned to the pseudo-outer region, while vertices in \([r_{\mathrm L},r_{\mathrm U}]\) satisfy the radius-accuracy estimate. Hence $d_{\mathrm{pre}}(u)\le |N(u)\cap B_0(r_{i-1}+\sqrt{\delta_n})|$. Moreover, \(u\in\widehat A_i\) implies \(r(u)\in(r_{\mathrm L},r_{\mathrm U})\) and, by radius accuracy, \(x\ge x_{i-1}-\sqrt{\delta_n}\).

Conditioning on the position of \(u\), Lemmas~\ref{lem:centered-ball-measure} and~\ref{lem:angular-threshold-asymptotics} give 
\begin{align*}
\mathbb E\left[|N(u)\cap B_0(r_{i-1}+\sqrt{\delta_n})|\mid r(u),\theta(u)\right]
&\le
O\left(n^{1-\alpha}e^{-\alpha x}\right)
+
O\left(
n^{1-\alpha}e^{-x/2}
\int_{-x}^{x_{i-1}+\sqrt{\delta_n}}
e^{(\alpha-1/2)y}\,dy
\right)\\
&\le
O\left(n^{1-\alpha}(\log n)^{-c_{\mathrm{in}}(1-\alpha)}\right),
\end{align*}
where we used \(x\ge x_{i-1}-\sqrt{\delta_n}\), \(x_{i-1}\ge c_{\mathrm{in}}\log\log n\), and \(\sqrt{\delta_n}=o(1)\).

A Poisson Chernoff bound applied to
\(|N(u)\cap B_0(r_{i-1}+\sqrt{\delta_n})|\) gives, for every fixed \(A>0\) and all sufficiently large \(n\),
\begin{align*}
\Pr\left(
|N(u)\cap B_0(r_{i-1}+\sqrt{\delta_n})|
>
n^{1-\alpha}(\log n)^{-c_{\mathrm{in}}(1-\alpha)/2}
\,\middle|\,
r(u),\theta(u)
\right)
\le n^{-A}.
\end{align*}

Since \(n^{1-\alpha}(\log n)^{-c_{\mathrm{in}}(1-\alpha)/2}=o(n^{1-\alpha})\), a union bound over all vertices and the \(m=O(1)\) pseudo-bands proves that \(d_{\mathrm{pre}}(u)=o(n^{1-\alpha})\) w.h.p. simultaneously for all \(i\in\{1,\dots,m\}\) and all \(u\in\widehat A_i\).
\end{proof}

\subsubsection*{Internal degree within pseudo-bands}

The next lemma bounds the internal degree of each pseudo-band, and hence the number of colours that may be removed by neighbours during its active phase.

\begin{lemma}[Internal degree of pseudo-bands]
\label{lem:constant-internal-degree-pseudo-bands}
With high probability, $d_{\mathrm{int}}(u)=o(|\Psi|)$ simultaneously for all $i\in\{1,\dots,m\}$ and all $u\in\widehat A_i$.
\end{lemma}

\begin{proof}
Fix \(i\in\{1,\dots,m\}\) and \(u\in\widehat A_i\), and set
\(A_i^+:=\{v\in V(G):r_{i-1}-\sqrt{\delta_n}<r(v)\le r_i+\sqrt{\delta_n}\}\).
On the radius-accuracy event \(\mathcal E\), \(\widehat A_i\subseteq A_i^+\), and hence \(d_{\mathrm{int}}(u)\le |N(u)\cap A_i^+|\). By Lemma~\ref{lem:centered-ball-measure}, the expected size of this enlarged annulus satisfies
\begin{align}
\mathbb E[|A_i^+|]
&=
n\left(\mu(B_0(r_i+\sqrt{\delta_n}))-\mu(B_0(r_{i-1}-\sqrt{\delta_n}))\right) \nonumber\\
&=
O\left(ne^{-\alpha(R-r_i)}\right)
=
O\left(
n^{1-\alpha+i\varepsilon(1-\alpha)}
(\log n)^{\alpha c_{\mathrm{in}}}
\right).
\label{eq:expected-size-constant-pseudoband}
\end{align}

Moreover, by Lemmas~\ref{lem:assignment-inner-cutoff-vertices} and~\ref{lem:assignment-outer-cutoff-vertices}, w.h.p.\ every vertex in \(\widehat A_i\) has radius in \((r_{\mathrm L},r_{\mathrm U})\). Since \(\widehat r(u)>r_{i-1}\), Lemma~\ref{lem:radius-estimate-accuracy} gives \(r(u)\ge r_{i-1}-\sqrt{\delta_n}\). Thus, by Lemma~\ref{lem:angular-threshold-asymptotics}, the angular connection probability between \(u\) and any vertex in \(A_i^+\) is at most
\begin{align}
O\left(e^{(R-2(r_{i-1}-\sqrt{\delta_n}))/2}\right)
=
O\left(
(\log n)^{-c_{\mathrm{in}}}
n^{-(i-1)\varepsilon(1-\alpha)/\alpha}
\right).
\label{eq:angular-threshold-constant-pseudoband}
\end{align}

Conditional on the position of \(u\), the count \(|N(u)\cap A_i^+|\) is Poisson. Its expectation is bounded by the expected size in~\eqref{eq:expected-size-constant-pseudoband} multiplied by the angular bound in~\eqref{eq:angular-threshold-constant-pseudoband}. Hence
\begin{align*}
\mathbb E\left[|N(u)\cap A_i^+|\mid r(u),\theta(u)\right]
&\le
O\left(
n^{(1-\alpha)(1+\varepsilon)}
(\log n)^{-c_{\mathrm{in}}(1-\alpha)}
\right)
=
o\left(n^{(1-\alpha)(1+\varepsilon)}\right)
=
o\left(|\Psi|\right).
\end{align*}

A Poisson Chernoff bound with deterministic threshold \(n^{(1-\alpha)(1+\varepsilon)}(\log n)^{-c_{\mathrm{in}}(1-\alpha)/2}\) gives failure probability at most \(n^{-A}\) for every fixed \(A>0\). A union bound over all vertices and the \(m=O(1)\) pseudo-bands, together with \(\mathcal E\) and the palette-size bound, gives \(d_{\mathrm{int}}(u)=o(|\Psi|)\) simultaneously for all \(i\in\{1,\dots,m\}\) and all \(u\in\widehat A_i\), w.h.p.
\end{proof}

\subsubsection*{Higher-priority degree}

The internal degree controls the number of colours removed during the active phase, whereas the probability that a vertex remains uncoloured in a given round depends only on its higher-priority neighbours. Lemma~\ref{lem:larger-degree-neighbourhood} yields the required bound.

\begin{lemma}[Small higher-priority degree]
\label{lem:small-higher-priority-degree}
Let $c'>0$ be constant. Then, w.e.h.p., every vertex
$u\in V(G)\setminus B_0(R/2+c'\log\log n)$ satisfies
$\deg^+(u)=o(n^{1-\alpha})$.
\end{lemma}

\begin{proof}
Let $\ell(u):=\lfloor R-r(u)\rfloor$ be the layer as in Lemma~\ref{lem:larger-degree-neighbourhood}. If $u\notin B_0(R/2+c'\log\log n)$, then $\ell(u)< \lfloor R/2-c'\log\log n\rfloor$, so only the first two cases of Lemma~\ref{lem:larger-degree-neighbourhood} can apply.

If $0\le \ell(u)\le \frac{2}{1-\alpha}\log\log n$, then w.e.h.p.
\begin{align*}
\deg^+(u)
=
O\left(e^{\ell(u)/2}+\log n\right)
=
O\left((\log n)^{1/(1-\alpha)}+\log n\right)
=
o(n^{1-\alpha}).
\end{align*}

Otherwise,
$\frac{2}{1-\alpha}\log\log n\le \ell(u)\le R/2-c'\log\log n$, and w.e.h.p.
\begin{align*}
\deg^+(u)
=
O\left(e^{(1-\alpha)\ell(u)}\right)
=
O\left(n^{1-\alpha}(\log n)^{-c'(1-\alpha)}\right)
=
o(n^{1-\alpha}),
\end{align*}
where we used $R=2\log n+O(1)$. This proves the claim.
\end{proof}

\subsubsection*{Band colouring process}

Next, we establish the constant-round colouring guarantee for an intermediate pseudo-band. By Lemma~\ref{lem:constant-prev-degree-pseudo-bands}, previously coloured neighbours have removed at most $d_{\mathrm{pre}}(u)=o(|\Psi|)$ colours when $\widehat A_i$ becomes active, w.e.h.p. During the active phase, Lemma~\ref{lem:constant-internal-degree-pseudo-bands} ensures that neighbours within $\widehat A_i$ remove at most $d_{\mathrm{int}}(u)=o(|\Psi|)$ further colours w.h.p. Thus, each vertex retains almost the entire palette throughout the phase. Finally, $\deg^+(u)=o(n^{1-\alpha})$ w.e.h.p. by Lemma~\ref{lem:small-higher-priority-degree}, so only few competing neighbours can prevent $u$ from fixing its chosen colour in a given round.

\begin{lemma}[Constant-round colouring of a pseudo-band]
\label{lem:band-constant-rounds}
For every $i\in\{1,\dots,m\}$ and every fixed integer $t>1/\varepsilon+i$, a.a.s. Sequential Radial Colouring colours all vertices of $\widehat A_i$ within $t$ rounds after its activation.
\end{lemma}

\begin{proof}
Fix $i\in\{1,\dots,m\}$. By Corollary~\ref{cor:constant-pseudo-region-size-comparison}, we have $|\widehat A_i|=(1+o(1))|A_i|$. Together with Lemma~\ref{lem:constant-true-region-concentration}, this gives
\begin{align*}
|\widehat A_i|
=
O\left(n^{1-\alpha+i\varepsilon(1-\alpha)}(\log n)^{\alpha c_{\mathrm{in}}}\right).
\end{align*}

By Lemma~\ref{lem:constant-prev-degree-pseudo-bands}, w.h.p., previously coloured neighbours forbid at most $d_{\mathrm{pre}}(u)=o(n^{1-\alpha})=o(|\Psi|)$ colours when $\widehat A_i$ becomes active. Moreover, Lemma~\ref{lem:constant-internal-degree-pseudo-bands} implies that, w.h.p., at most $d_{\mathrm{int}}(u)=o(|\Psi|)$ further colours are forbidden during the active phase. Consequently, throughout this phase, every $u\in\widehat A_i$ has at least $|\Psi|-d_{\mathrm{pre}}(u)-d_{\mathrm{int}}(u)=(1-o(1))|\Psi|$ available colours.

Fix a round and condition on the preceding history. Since $\widehat r(u)>r_{\mathrm{in}}$, Lemma~\ref{lem:radius-estimate-accuracy} gives $r(u)\ge r_{\mathrm{in}}-\sqrt{\delta_n}\ge R/2+(c_{\mathrm{in}}/2)\log\log n$ for all sufficiently large $n$, w.h.p. 
Hence, applying Lemma~\ref{lem:small-higher-priority-degree} with \(c'=c_{\mathrm{in}}/2\), we obtain \(\deg^+(u)=o(n^{1-\alpha})\) w.e.h.p. Vertex $u$ remains uncoloured only if a higher-priority uncoloured neighbour chooses the same colour. Since every active vertex has $(1-o(1))|\Psi|$ available colours, we get
\begin{align*}
\Pr\left(u\text{ remains uncoloured}\mid\text{history}\right)
\le
\frac{\deg^+(u)}{(1-o(1))|\Psi|}
=
o\left(n^{-\varepsilon(1-\alpha)}\right).
\end{align*}

As this estimate holds conditionally in every round, the probability that $u$ remains uncoloured after $t$ rounds is $o\left(n^{-\varepsilon(1-\alpha)t}\right)$. Therefore, a union bound over $\widehat A_i$ gives
\begin{align*}
\Pr\left(\exists u\in\widehat A_i\text{ uncoloured after $t$ rounds}\right)
\le
o\left(
n^{1-\alpha+i\varepsilon(1-\alpha)-\varepsilon(1-\alpha)t}
(\log n)^{\alpha c_{\mathrm{in}}}
\right)
=
o(1),
\end{align*}
since $t>1/\varepsilon+i$. Thus, $\widehat A_i$ is coloured within $t$ rounds a.a.s.
\end{proof}

\subsection{Colouring the pseudo-outer region}
\label{subsec:constant-colouring-pseudo-outer-region}

The pseudo-outer region $\widehat{\mathcal O}$ is processed last, so its vertices may inherit colour restrictions from all preceding regions. The following degree bound shows that each vertex has only $O(n^{1-\alpha})$ forbidden colours, which is polynomially smaller than $|\Psi|$ and hence permits colouring in a constant number of rounds.

\begin{lemma}[Degree bound in the pseudo-outer region]
\label{lem:degree-pseudo-outer-region}
With high probability, every $u\in\widehat{\mathcal O}$ satisfies $\deg(u)=O(n^{1-\alpha})$.
\end{lemma}

\begin{proof}
Fix $u\in\widehat{\mathcal O}$. If $r(u)\le r_{\mathrm U}$, then $\widehat r(u)>r_{\mathrm{tail}}$, and Lemma~\ref{lem:radius-estimate-accuracy} gives $r(u)\ge r_{\mathrm{tail}}-\sqrt{\delta_n}$ w.h.p. Hence, by Lemma~\ref{lem:degree-concentration-above-cutoff},
\begin{align*}
\deg(u)
\le
(1+o(1))\bar d(r(u))
\le
O\left(ne^{-(r_{\mathrm{tail}}-\sqrt{\delta_n})/2}\right)
=
O\left(
ne^{-\alpha\log n}e^{\sqrt{\delta_n}/2}
\right)
=
O\left(n^{1-\alpha}\right),
\end{align*}
where we used $r_{\mathrm{tail}}=2\alpha \log n$. If $r(u)>r_{\mathrm U}$, Lemma~\ref{lem:right-degree-cutoff} gives $\deg(u)\le K\log^2 n=O(n^{1-\alpha})$ w.h.p. Thus every vertex in $\widehat{\mathcal O}$ has degree $O(n^{1-\alpha})$ w.h.p.
\end{proof}

\begin{lemma}[Colouring the pseudo-outer region]
\label{lem:outer-colouring}
For every fixed number $t>1/(\varepsilon(1-\alpha))$ of rounds SRC colours $\widehat{\mathcal O}$ a.a.s. after $t$ rounds of its activation.
\end{lemma}

\begin{proof}
Fix $u\in\widehat{\mathcal O}$. By Lemma~\ref{lem:degree-pseudo-outer-region}, $\deg(u)=O(n^{1-\alpha})$ w.h.p. Throughout the active phase, $u$ has at least $|\Psi|-\deg(u)=(1-o(1))|\Psi|$ available colours. 

Fix a round and condition on the preceding history. Vertex $u$ remains uncoloured only if a higher-priority uncoloured neighbour chooses the same colour. Consequently,
\begin{align*}
\Pr\left(u\text{ remains uncoloured}\mid\text{history}\right)
&\le
\frac{\deg(u)}{(1-o(1))|\Psi|}
=
O\left(n^{-\varepsilon(1-\alpha)}\right).
\end{align*}

Thus, $u$ remains uncoloured after $t$ rounds with probability $O(n^{-\varepsilon(1-\alpha)t})$. A union bound over at most $O(n)$ vertices gives $\Pr(\exists u\in\widehat{\mathcal O}\text{ uncoloured after $t$ rounds})\leq O(n^{1-\varepsilon(1-\alpha)t})=o(1)$, since $t>1/(\varepsilon(1-\alpha))$.
\end{proof}

\subsection{Proof of the \texorpdfstring{$O(\chi^{1+\varepsilon})$}{O(χ(G) raised to (1 + ε))}-colour theorem}
\label{subsec:proof-constant-round-colouring}

We now combine the preceding region-wise colouring bounds. Each pseudo-region is coloured within a constant number of rounds after its activation a.a.s. Since the pseudo-decomposition contains only constantly many regions, the total number of rounds is $O(1)$ a.a.s.

\begin{theorem}[SRC with $O(\chi^{1+\varepsilon})$ colours]
\label{thm:constant-round-colouring}
Let $G\sim \mathcal G(n,\alpha,C)$ be a typical hyperbolic random graph with $\alpha\in(1/2,1)$, and fix $\varepsilon>0$. Using a palette $\Psi$ of size $|\Psi|=\lceil\chi^{1+\varepsilon}\rceil$, Sequential Radial Colouring colours $G$ within $O(1)$ rounds a.a.s.
\end{theorem}

\begin{proof}
We define $T:=1+\max\left\{\left\lceil 1/(\varepsilon(1-\alpha))\right\rceil, m+\left\lceil 1/\varepsilon\right\rceil\right\}$. Since $T>1/\varepsilon$, Corollary~\ref{cor:constant-round-colouring-inner-pseudo-region} applies to $\widehat{\mathcal I}$. Moreover, $T>1/\varepsilon+i$ for every $i\in\{1,\dots,m\}$ and $T>1/(\varepsilon(1-\alpha))$, so Lemmas~\ref{lem:band-constant-rounds} and~\ref{lem:outer-colouring} apply to all pseudo-bands and to $\widehat{\mathcal O}$, respectively. Since $m=O(1)$, we have $T=O_\varepsilon(1)$, and all $m+2$ pseudo-regions are processed within $(m+2)T=O(1)$ rounds. A union bound over their constantly many failure probabilities completes the proof.
\end{proof}

Theorem~\ref{thm:constant-round-colouring} implies the third item of Theorem~\ref{thm:summary}. Since $\chi\le\kappa+1$, for every fixed $\varepsilon>0$ we have $\kappa^{1+\varepsilon}\ge\chi^{1+\varepsilon/2}$ for all sufficiently large $n$. Thus any palette with $|\Psi|\ge\kappa^{1+\varepsilon}$ satisfies the hypothesis of Theorem~\ref{thm:constant-round-colouring} with parameter $\varepsilon/2$. Hence SRC colours $G$ in $O(1)$ rounds a.a.s.

\section{Clique colouring}


We give a self-contained analysis of RCTDEG on a clique $\mathcal C$ with $n$ vertices. Since all vertices have the same degree, the degree-based priority reduces to an arbitrary tie-breaking order, for instance one induced by vertex identifiers. In each round, every uncoloured vertex chooses uniformly at random from the colours not already used by coloured neighbours and keeps its choice unless a higher-priority uncoloured neighbour chooses the same colour.

This section establishes two upper bounds for distinct optimisation regimes. The first minimises the palette size, showing that an essentially optimal number of colours still permits an $O(\log\log n)$-round colouring a.a.s. The second minimises the round complexity, showing that polynomial multiplicative slack yields constant-round colouring. We then complement these results with lower bounds.

\subsection{Upper bounds}
\label{subsec:clique-upper-bounds}

We establish upper bounds for two palette regimes:
\begin{itemize}
    \item If $|\Psi|=n+s_n$ with $1\le s_n=o(n)$, then $\mathcal C$ is coloured within $O(\log(n/s_n)+\log\log n)$ rounds a.a.s.; in particular, $(1+o(1))n$ colours suffice for $O(\log\log n)$ rounds a.a.s.
    \item If $|\Psi|=\lceil n^{1+\varepsilon_n}\rceil$ with $\varepsilon_n>0$ and $n^{\varepsilon_n}\to\infty$, then $\mathcal C$ is coloured within $O(1/\varepsilon_n)$ rounds a.a.s.; in particular, fixed $\varepsilon_n=\varepsilon>0$ gives $O(1)$ rounds.
\end{itemize}

\subsubsection*{Colouring with additive slack}

We first establish an upper bound for a general graph $H$ on $n$ vertices when additive palette slack is available. The proof proceeds in three phases: geometric decay to a sufficiently small constant fraction of the slack, quadratic decay to $O(\log n)$ uncoloured vertices, and completion in one additional round.

\begin{lemma}[Graph colouring with additive slack]
\label{lem:clique-colouring-additive-slack}
Let $H$ be a graph on $n$ vertices, and let $1\le s_n=o(n)$. If $|\Psi|\ge n+s_n$, then RCTDEG colours $H$ a.a.s. within $O\left(\log(n/s_n)+\log\log n\right)$ rounds.
\end{lemma}

\begin{proof}
For $i\ge1$, let $\mathcal U_{i-1}\subseteq V(H)$ be the set of vertices uncoloured after round $i-1$, and write $U_{i-1}:=|\mathcal U_{i-1}|$. At the beginning of round $i$, at most $n-U_{i-1}$ colours are forbidden at any uncoloured vertex by coloured neighbours. Hence every vertex in $\mathcal U_{i-1}$ has at least $|\Psi|-(n-U_{i-1})\ge s_n+U_{i-1}$ available colours.

Order the vertices of $\mathcal U_{i-1}$ by decreasing priority as $v_1^i,\dots,v_{U_{i-1}}^i$, and let $I_j^i$ be the indicator that $v_j^i$ remains uncoloured after round $i$. The vertex $v_j^i$ has at most $j-1$ higher-priority uncoloured neighbours and at least $s_n+U_{i-1}$ available colours. Hence Corollary~\ref{cor:sequential-stochastic-domination-sub} applies with $a_j=s_n+U_{i-1}$ and $b_j=j-1$. 

Hence, conditional on $\mathcal U_{i-1}$, the number $U_i$ of vertices remaining uncoloured after round $i$ is stochastically dominated by $Z_i:=\sum_{j\le U_{i-1}}Z_j^i$, where the $Z_j^i$ are conditionally independent Bernoulli random variables with $\Pr(Z_j^i=1\mid\mathcal U_{i-1})=(j-1)/(s_n+U_{i-1})$. Consequently, we have
\begin{align*}
\nu_i:=\mathbb E[Z_i\mid\mathcal U_{i-1}]
=\sum_{j=1}^{U_{i-1}}\frac{j-1}{s_n+U_{i-1}}
=\frac{U_{i-1}(U_{i-1}-1)}{2(s_n+U_{i-1})}
\le\frac{U_{i-1}^2}{2(s_n+U_{i-1})}.
\end{align*}

We now distinguish two cases according to the size of $s_n$, namely whether $s_n=\omega(\log^2 n)$ or $s_n=O(\log^2 n)$.

\textbf{Case 1 $[s_n = \omega(\log^2 n)]$:} Fix $\lambda\in(1/2,1)$. Since $\nu_i\le U_{i-1}/2$, the threshold $\lambda U_{i-1}$ exceeds the conditional expectation of $Z_i$ by a constant factor. Hence, stochastic domination and a Chernoff bound imply that there exists a constant $c_\lambda>0$ such that
\begin{align*}
\Pr\left(U_i>\lambda U_{i-1}\mid\mathcal U_{i-1}\right)\le\Pr\left(Z_i>\lambda U_{i-1}\mid\mathcal U_{i-1}\right)\le\exp(-c_\lambda U_{i-1}).
\end{align*}

We use this estimate to show that, a.a.s., the number of uncoloured vertices falls below $\delta s_n$ within $O\left(\log\left(n/s_n\right)\right)$ rounds, where $\delta\in(0,1)$ is a constant chosen later. As long as $U_{i-1}\ge\delta s_n$, the assumption $s_n=\omega(\log^2 n)$ implies $U_{i-1}=\omega(\log^2 n)$. Hence, for every fixed $c>0$ and all sufficiently large $n$, the preceding estimate yields $\Pr\left(U_i>\lambda U_{i-1}\mid\mathcal U_{i-1}\right)\le\exp(-c_\lambda U_{i-1})\le n^{-c}$. Thus, in every such round, the bound $U_i\le\lambda U_{i-1}$ holds with conditional probability at least $1-n^{-c}$. Since $U_0=n$, after $t$ successful rounds we have $U_t\le\lambda^t n$. Therefore, choose
\begin{align*}
t_0:=\left\lceil\frac{\log\left(n/(\delta s_n)\right)}{\log(1/\lambda)}\right\rceil=O\left(\log\left(n/s_n\right)\right).
\end{align*}

If $U_{i-1}<\delta s_n$ for some $i\in\{1,\dots,t_0\}$, then $U_{t_0}<\delta s_n$ by monotonicity. Otherwise, $U_{i-1}\ge\delta s_n$ for every $i\in\{1,\dots,t_0\}$, and the bound $U_i\le\lambda U_{i-1}$ fails in any such round with probability at most $n^{-c}$. Since $t_0=O(\log n)$, a union bound shows that $U_{t_0}\le\lambda^{t_0}n=\delta s_n$ w.h.p.

Next, we show that, a.a.s., $U_i=O(\log n)$ after $O(\log\log n)$ additional rounds. For this, we first show that $\nu_i$ falls below $C_0\log n$, after which another Chernoff bound gives $U_i=O(\log n)$. While $\nu_i\ge C_0\log n$, stochastic domination and a Chernoff bound yield, for sufficiently large $C_1=C_1(C_0,c)>1$, $\Pr\left(U_i>C_1\nu_i\mid\mathcal U_{i-1}\right)\le\Pr\left(Z_i>C_1\nu_i\mid\mathcal U_{i-1}\right)\le n^{-c}$. Consequently, with conditional failure probability at most $n^{-c}$,
\begin{align}
U_i\le C_1\nu_i\le\frac{C_1}{2}\frac{U_{i-1}^2}{s_n}.
\label{eq:clique-quadratic-decay}
\end{align}

Choose $\delta>0$ sufficiently small that $C_1\cdot \delta\le1$. By the preceding argument, $U_{t_0}\le\delta s_n$ with failure probability at most $t_0n^{-c}$. We claim that, for every $r\ge0$, as long as $\nu_{t_0+j}\ge C_0\log n$ for all $j\in\{1,\dots,r\}$, the bound $U_{t_0+r}\le\delta s_n\,2^{-(2^r-1)}$ holds with failure probability at most $(t_0+r)n^{-c}$. We prove this by induction on $r$.

For $r=0$, the claim follows from the bound on $U_{t_0}$. Suppose that it holds for some $r\ge0$ and that $\nu_{t_0+r+1}\ge C_0\log n$. By \eqref{eq:clique-quadratic-decay}, conditional on the induction bound,
\begin{align*}
U_{t_0+r+1}
\le
\frac{C_1}{2}\frac{U_{t_0+r}^2}{s_n}
\le
\delta s_n\,2^{-(2^{r+1}-1)}
\end{align*}
with conditional failure probability at most $n^{-c}$. Hence, by a union bound, the bound at round $t_0+r+1$ holds with failure probability at most $(t_0+r+1)n^{-c}$, completing the induction.

Choose an integer $r=O(\log\log n)$ sufficiently large that $\left(\delta s_n\,2^{-(2^r-1)}\right)^2/(2s_n)<C_0\log n$. If $\nu_i\ge C_0\log n$ for every $i\in\{t_0+1,\dots,t_0+r\}$, then the induction bound holds up to round $t_0+r$ with failure probability at most $(t_0+r)n^{-c}$ and yields
\begin{align*}
\nu_{t_0+r+1}\le\frac{U_{t_0+r}^2}{2s_n}\le\frac{\left(\delta s_n\,2^{-(2^r-1)}\right)^2}{2s_n}<C_0\log n.
\end{align*}

Otherwise, $\nu_i<C_0\log n$ already holds for some $i\in\{t_0+1,\dots,t_0+r\}$. Thus, with failure probability at most $(t_0+r)n^{-c}=o(1)$, there exists a round $i\le t_0+r+1$ such that $\nu_i<C_0\log n$.

Once $\nu_i<C_0\log n$, stochastic domination converts this bound on the conditional mean into a logarithmic bound on $U_i$. Choose a constant $C_2>eC_0$ sufficiently large that $C_2\log(C_2/(eC_0))\ge c$. Conditional on $\mathcal U_{i-1}$, the random variable $U_i$ is stochastically dominated by $Z_i$, and Chernoff's bound~\cite[Theorem~1.10.1]{doerr2020probabilistic} gives
\begin{align*}
\Pr\left(Z_i\ge C_2\log n\mid\mathcal U_{i-1}\right)
\le
\left(\frac{e\nu_i}{C_2\log n}\right)^{C_2\log n}
\le
\left(\frac{eC_0}{C_2}\right)^{C_2\log n}
=
n^{-C_2\log(C_2/(eC_0))}
\le
n^{-c}.
\end{align*}

Therefore, in the first round with $\nu_i<C_0\log n$, we have $U_i=O(\log n)$ with conditional failure probability at most $n^{-c}$.


It remains to colour the final $O(\log n)$ uncoloured vertices. Conditional on $\mathcal U_i$, if some vertex remains uncoloured after round $i+1$, then some pair of vertices in $\mathcal U_i$ must have chosen the same colour. Each such vertex has at least $s_n+U_i$ available colours, so any fixed pair chooses the same colour with probability at most $1/(s_n+U_i)$. A union bound over all pairs
therefore gives
\begin{align*}
\Pr(\mathcal U_{i+1}\neq\emptyset\mid\mathcal U_i)
&\le
\binom{U_i}{2}\frac{1}{s_n+U_i}
\le
O\left(\frac{\log^2 n}{s_n}\right)
=
o(1),
\end{align*}
where we used $s_n=\omega(\log^2 n)$.


Combining the three phases and taking a union bound over the corresponding failure events, whose total probability is $o(1)$, all vertices in $H$ are coloured within $O\left(\log(n/s_n)+\log\log n\right)$ rounds a.a.s.

\textbf{Case 2 $[s_n = O(\log^2(n))]$:}  Recall that, conditional on
$\mathcal U_{i-1}$,
\begin{align*}
\nu_i
:=
\mathbb E[Z_i\mid\mathcal U_{i-1}]
\le
\frac{U_{i-1}^2}{2(s_n+U_{i-1})}
\le
\frac{U_{i-1}}{2}.
\end{align*}

Since $U_i$ is stochastically dominated by $Z_i$, we have
$\mathbb E[U_i\mid\mathcal U_{i-1}] \le \mathbb E[Z_i\mid\mathcal U_{i-1}] =\nu_i \le U_{i-1}/2$. Taking expectations and iterating yields
\begin{align*}
\mathbb E[U_t]
\le
\frac{\mathbb E[U_{t-1}]}{2}
\le
\cdots
\le
\frac{U_0}{2^t}
\le
\frac{n}{2^t}.
\end{align*}

Choose $t:=\lceil 3\log_2 n\rceil$. Then $\mathbb E[U_t]\le n/2^t\le 1/n^2$. By Markov's inequality, $\Pr(U_t\ge1)\le\mathbb E[U_t]\le 1/n^2=o(1)$. Hence $U_t=0$ a.a.s. Since $s_n=O(\log^2 n)$ implies
$\log(n/s_n)=\Theta(\log n)$, this is within
$O(\log(n/s_n)+\log\log n)$ rounds.
\end{proof}

The preceding lemma can be viewed as a slack-sensitive refinement of Johansson's random-colour-trial analysis~\cite{johansson1999simple}. For a clique $\mathcal C$, it recovers the standard $O(\log|\mathcal C|)$ bound when the additive slack is $s=1$. More generally, it makes explicit how the round complexity depends on $s$. In particular, choosing $s=o(n)$ appropriately yields a $(1+o(1))n$-colouring guarantee for every $n$-vertex graph in $O(\log\log n)$ rounds a.a.s.

\begin{corollary}[$(1+o(1))n$ colouring for general graphs]
\label{cor:near-tight-general-graph-colouring}
For every graph $H$ on $n$ vertices, there exists $\varepsilon_n=o(1)$ such that for $|\Psi|\ge(1+\varepsilon_n)n$ RCTDEG colours $H$ within $O(\log\log n)$ rounds a.a.s.
\end{corollary}

\begin{proof}
Set $s_n:=\lceil n/\log n\rceil$. Then $s_n=o(n)$, so Lemma~\ref{lem:clique-colouring-additive-slack} applies to $H$. Since $\log(n/s_n)=O(\log\log n)$, it yields an $O(\log\log n)$-round colouring a.a.s. Finally, setting $\varepsilon_n:=s_n/n=o(1)$ gives $n+s_n=(1+\varepsilon_n)n$, proving the claim.
\end{proof}

Thus, RCTDEG colours every $n$-vertex graph $H$, and in particular any clique $\mathcal C$, within $O(\log\log n)$ rounds using $(1+o(1))n$ colours a.a.s., which is asymptotically optimal for cliques.

\subsubsection*{Colouring with multiplicative slack}

We next analyse multiplicative palette slack. For every graph $H$ on $n$ vertices, $|\Psi|=\lceil n^{1+\varepsilon_n}\rceil$ colours imply a $O(1/\varepsilon_n)$-round colouring when $n^{\varepsilon_n}\to\infty$. For cliques, polynomial slack therefore gives constant-round colouring.



\begin{lemma}[Colouring with variable multiplicative slack]
\label{lem:variable-slack-graph-colouring}
Let $H$ be a graph on $n$ vertices, and let $|\Psi|=\lceil n^{1+\varepsilon_n}\rceil$, where $\varepsilon_n>0$ and $n^{\varepsilon_n}\to\infty$. Then RCTDEG colours $H$ in $(1+o(1))/\varepsilon_n$ rounds a.a.s.
\end{lemma}

\begin{proof}
Let $t_n:=\left\lceil(\log n+\gamma_n)/\log(n^{\varepsilon_n}-1)\right\rceil$, where $\gamma_n=\omega(1)$ and $\gamma_n=o(\log n)$. Since every vertex has at most $n$ neighbours and at most $n$ colours can be forbidden by already coloured neighbours, every uncoloured vertex has at least $|\Psi|-n\ge n(n^{\varepsilon_n}-1)$ available colours in every round. Hence, conditional on the previous history, a fixed uncoloured vertex remains uncoloured in one round with probability at most $n/(|\Psi|-n)\le(n^{\varepsilon_n}-1)^{-1}$.

Therefore, the probability that a fixed vertex remains uncoloured after $t_n$ rounds is at most $(n^{\varepsilon_n}-1)^{-t_n}$. A union bound over all vertices
gives
\begin{align*}
\Pr\left(\exists v\in V(H)\text{ uncoloured after $t_n$ rounds}\right)
&\le
n(n^{\varepsilon_n}-1)^{-t_n}
\le
\exp(-\gamma_n)
=
o(1).
\end{align*}

Finally, since $n^{\varepsilon_n}\to\infty$, we have $\log(n^{\varepsilon_n}-1)=(1+o(1))\varepsilon_n\log n$. Since $\gamma_n=o(\log n)$, this gives $t_n=(1+o(1))/\varepsilon_n$.
\end{proof}

\begin{corollary}[Constant-round colouring]
\label{cor:constant-round-graph-colouring}
Let $H$ be a graph on $n$ vertices, and let $|\Psi|=\lceil n^{1+\varepsilon}\rceil$ for a fixed constant $\varepsilon>0$. Then RCTDEG colours $H$ a.a.s. within any fixed number $t>1/\varepsilon$ of rounds.
\end{corollary}

\begin{proof}
Since $n^\varepsilon\to\infty$, we can apply Lemma~\ref{lem:variable-slack-graph-colouring} with $\varepsilon_n=\varepsilon$, and the number of rounds is $(1+o(1))/\varepsilon$ a.a.s. Hence, for every fixed $t>1/\varepsilon$, this is at most $t$ for all sufficiently large $n$.
\end{proof}

For a clique $\mathcal C$ on $n$ vertices and every fixed $\varepsilon>0$, $n^{1+\varepsilon}$ colours suffice for RCTDEG to colour $\mathcal C$ in any fixed number of rounds larger than $1/\varepsilon$ a.a.s. The resulting trade-off between palette size and round complexity is summarised in Table~\ref{tab:clique-upper-bounds}.

\begin{table}[t]
\centering
\renewcommand{\arraystretch}{1.35}
\begin{tabular}{l@{\hspace{1.5cm}}l}
\toprule
\textbf{Palette size} & \textbf{Round complexity} \\
\midrule
$n+s_n$, $1\le s_n=o(n)$
&
$O(\log(n/s_n)+\log\log n)$
\\
\addlinespace
$\left(1+\frac{1}{\log n}\right)n$
&
$O(\log\log n)$
\\
\addlinespace
$n^{1+\varepsilon_n}$, $n^{\varepsilon_n}\to\infty$
&
$(1+o(1))/\varepsilon_n$
\\
\addlinespace
$n^{1+\varepsilon}$, $\varepsilon>0$ constant
&
$O(1)$
\\
\bottomrule
\end{tabular}
\caption{Upper bounds for colouring a clique on $n$ vertices, where $s_n$ denotes additive slack.}
\label{tab:clique-upper-bounds}
\end{table}

\subsection{Lower bounds}

We complement the preceding upper bounds with lower bounds for the same clique process (see Table~\ref{tab:clique-lower-bounds}):
\begin{itemize}
    \item If $|\Psi|=n(\log n)^{O(1)}$, then RCTDEG requires $\Omega(\log\log n)$ rounds a.a.s.; in particular, this matches the $O(\log\log n)$ upper bound for near-optimal palettes up to constant factors.
    \item If $|\Psi|=n^{1+\varepsilon_n}$ with $n^{\varepsilon_n}\to\infty$, then RCTDEG needs $\Omega\!\left(\log(1+1/\varepsilon_n)\right)$ rounds a.a.s. If $\varepsilon_n=o(1)$, every constant number of rounds leaves $\omega(1)$ vertices uncoloured a.a.s.
\end{itemize}

Recall that $\mathcal U_r\subseteq\mathcal C$ denotes the set of vertices remaining uncoloured after round $r$, and let $U_r:=|\mathcal U_r|$. To prove these statements in a unified form, suppose that $|\Psi|=\Theta(nf(n))$ for some $f(n)\ge1$. We first use concentration for self-bounding functions to obtain a one-round lower-tail bound for $U_r$, and then derive a general recursive lower bound, which we instantiate with $f(n)=\Theta(1)$ and $f(n)=n^{o(1)}$.

\subsubsection*{Lower-tail concentration for one clique round}

Let $\mathcal X:=\prod_{i=1}^m\mathcal X_i$ and let $g\colon\mathcal X\to\mathbb R_{\ge0}$. For $x\in\mathcal X$ and $i\in\{1,\dots,m\}$, define
\begin{align*}
g_i(x)
:=
\inf\left\{
g(x'):
x'_j=x_j\text{ for every }j\neq i
\right\}.
\end{align*}

Thus, $g_i(x)$ is the infimum of the values obtained by changing only the $i$th coordinate of $x$, and hence $0\le g(x)-g_i(x)$. For $a\ge0$ and $b\in\mathbb R$, the function $g$ is called \emph{$(a,b)$-self-bounding} if, for every $x\in\mathcal X$,
\begin{align*}
g(x)-g_i(x)
\le
1
\quad\text{for every }i\in\{1,\dots,m\},
\qquad
\sum_{i=1}^m\bigl(g(x)-g_i(x)\bigr)
\le
ag(x)+b.
\end{align*}

We use the following concentration bound.

\begin{lemma}[Lower-tail concentration for self-bounding functions
{\cite[Theorem~1]{mcdiarmid2006concentration}}]
\label{lem:self-bounding-concentration}
Let $X_1,\dots,X_m$ be independent random variables, where $X_i$ takes values in $\mathcal X_i$, and let $Z:=g(X_1,\dots,X_m)$ have mean $\mu$. If $g$ is measurable and $(a,b)$-self-bounding for some $a\ge0$ and $b\in\mathbb R$, then, for every $t>0$,
\begin{align*}
\Pr(Z\le\mu-t)
\le
\exp\left(
-\frac{t^2}{2(a\mu+b+t/3)}
\right).
\end{align*}
\end{lemma}

We now apply Lemma~\ref{lem:self-bounding-concentration} to $U_r$ by showing that, conditional on $\mathcal U_{r-1}$, it is a self-bounding function of the independent colour choices made in round $r$.

\begin{lemma}[Lower-tail concentration in one clique round]
\label{lem:clique-round-lower-tail}
Let $\mathcal C$ be a clique on $n$ vertices with $|\Psi|\ge n$. Fix a clique round $r$ and condition on $\mathcal U_{r-1}$. Let
$\mu_r:=\mathbb E\left[U_r\,\middle|\,\mathcal U_{r-1}\right]$. Then
\begin{align*}
\Pr\left(
U_r\le\frac{\mu_r}{2}
\,\middle|\,
\mathcal U_{r-1}
\right)
\le
\exp\left(-\frac{3\mu_r}{52}\right).
\end{align*}
\end{lemma}

\begin{proof}
Condition on $\mathcal U_{r-1}$ and write $m:=U_{r-1}$. Let $\Psi_r$ denote the set of colours not used before round $r$, which is the common set of colours available to the vertices in $\mathcal U_{r-1}$. Since the $n-m$ already coloured vertices use pairwise distinct colours, we have $|\Psi_r|=|\Psi|-(n-m)\ge m$.

Since only equality between colour choices matters, we may relabel the colours in $\Psi_r$ arbitrarily. For every $\mathbf c=(c_1,\dots,c_m)\in\Psi_r^m$ and $c\in\Psi_r$, let
$N_c(\mathbf c):=|\{j\in[m]:c_j=c\}|$. Define, for all $i\in[m]$,
\begin{align*}
f(\mathbf c)
&:=
\sum_{c\in\Psi_r}\left(N_c(\mathbf c)-1\right)_+,
&
f_i(\mathbf c)
&:=
\inf_{\widehat c\in\Psi_r}
f(c_1,\dots,c_{i-1},\widehat c,c_{i+1},\dots,c_m),
\end{align*}
where $x_+:=\max\{0,x\}$. To evaluate the infimum defining $f_i$, fix all coordinates other than $i$. Assigning a colour $\widehat c\in\Psi_r$ to the $i$th coordinate increases only the count $N_{\widehat c}$ by one; hence it can increase $f$ by at most one. Since the other $m-1$ coordinates use at most $m-1$ colours and $|\Psi_r|\ge m$, some colour is unused. Choosing such a colour does not increase any term and therefore attains the infimum. Consequently, $f_i(\mathbf c)=\sum_{c\in\Psi_r}\left(N_c(\mathbf c)-\mathbf 1_{\{c_i=c\}}-1\right)_+$, where $\mathbf 1_{\{c_i=c\}}$ denotes the indicator that $c_i=c$. It follows that removing the $i$th colour choice decreases $f$ by one precisely when at least one other vertex has colour $c_i$. Thus,
\begin{align*}
f(\mathbf c)-f_i(\mathbf c)
=
\begin{cases}
1, & \text{if }N_{c_i}(\mathbf c)\ge2,\\
0, & \text{if }N_{c_i}(\mathbf c)=1,
\end{cases}
\end{align*}
and hence $0\le f(\mathbf c)-f_i(\mathbf c)\le1$. Moreover,
\begin{align*}
\sum_{i=1}^m\bigl(f(\mathbf c)-f_i(\mathbf c)\bigr)
&=
\sum_{\substack{c\in\Psi_r\\N_c(\mathbf c)\ge2}}
N_c(\mathbf c)
\le
2\sum_{\substack{c\in\Psi_r\\N_c(\mathbf c)\ge2}}
\left(N_c(\mathbf c)-1\right)
=
2f(\mathbf c),
\end{align*}
where we used $N_c(\mathbf c)\le2(N_c(\mathbf c)-1)$ whenever
$N_c(\mathbf c)\ge2$. Therefore, $f$ is $(2,0)$-self-bounding.

Let $C_1,\dots,C_m$ denote the colours chosen by the uncoloured vertices in round $r$, and let $\mathbf C:=(C_1,\dots,C_m)$. Conditional on $\mathcal U_{r-1}$, these choices are independent and uniformly distributed over $\Psi_r$. Among the vertices choosing the same colour, precisely the vertex of highest priority becomes coloured. Hence, a colour chosen by $N_c(\mathbf C)$ vertices leaves $(N_c(\mathbf C)-1)_+$ vertices uncoloured, and therefore, $U_r=\sum_{c\in\Psi_r}(N_c(\mathbf C)-1)_+=f(\mathbf C)$.

If $\mu_r=0$, the claimed bound holds trivially. Hence, assume that $\mu_r>0$. Since $\mathbb E[f(\mathbf C)\mid\mathcal U_{r-1}]=\mu_r$, Lemma~\ref{lem:self-bounding-concentration}, applied conditionally on $\mathcal U_{r-1}$ with $a=2$, $b=0$, and $t=\mu_r/2$, gives
\begin{align*}
\Pr\left(
U_r\le\frac{\mu_r}{2}
\,\middle|\,
\mathcal U_{r-1}
\right)
=
\Pr\left(
f(\mathbf C)\le
\frac{\mathbb E\left[f(\mathbf C)\,\middle|\,\mathcal U_{r-1}\right]}{2}
\,\middle|\,
\mathcal U_{r-1}
\right) 
\leq 
\exp\left(-\frac{3\mu_r}{52}\right),
\end{align*}
which proves the claim.
\end{proof}

\subsubsection*{General clique lower bound}

We first state a general lower bound that tracks the number of uncoloured vertices after an arbitrary number of rounds, via a deterministic sequence whose failure probability is explicitly controlled. For constants $c_0,C_0>0$, to be fixed in the proof, define $(a_r)_{r=0}^t$ by
\begin{align*}
a_0:=c_0n,
\qquad
a_r:=\frac{C_0}{2}\cdot\frac{a_{r-1}^2}{|\Psi|}
\quad\text{for }r\in\{1,\dots,t\}.
\end{align*}

\begin{lemma}[General clique lower bound]
\label{lem:clique-general-lower-bound}
Let $\mathcal C$ be a clique on $n$ vertices and let $|\Psi|=\Theta(nf(n))$ for some $f(n)\ge1$. There exist constants $c_0,C_0>0$ and $D\ge1$ such that
\begin{align*}
a_r\ge c_0\frac{n}{(Df(n))^{2^r-1}}
\end{align*}
for every $r\in\{0,\dots,t\}$. Moreover, if $a_{t-1}\ge2$, then, for all sufficiently large $n$,
\begin{align*}
\Pr(U_t\ge a_t)
\ge
\prod_{r=1}^t
\left(1-\exp\left(-\frac{3a_r}{26}\right)\right).
\end{align*}
\end{lemma}

\begin{proof}
Since $U_0=|\mathcal C|=n$ and $|\Psi|=\Theta(nf(n))$ with $f(n) \geq 1$, there exist constants $c_\Psi,C_\Psi>0$ such that $c_\Psi nf(n)\le|\Psi|\le C_\Psi nf(n)$.

Fix $r\in\{1,\dots,t\}$ and condition on $\mathcal U_{r-1}$. Order the vertices of $\mathcal U_{r-1}$ by decreasing priority, and let $I_j^r$ be the indicator variable whether the $j$th vertex remains uncoloured after round $r$. Then $U_r=\sum_{j\leq U_{r-1}}I_j^r$. All active vertices use the same palette\footnote{Since in a clique the colour restrictions are the same for every vertex.} $\Psi_r\subseteq\Psi$, and the $j$th vertex remains uncoloured if one of the $j-1$ higher-priority vertices chooses the same colour. Hence,
\begin{align*}
\Pr\left(I_j^r=1\mid\mathcal U_{r-1}\right)
=
1-\left(1-\frac{1}{|\Psi_r|}\right)^{j-1}
\ge
1-\left(1-\frac{1}{|\Psi|}\right)^{j-1}.
\end{align*}

Writing $\mu_r:=\mathbb E[U_r\mid\mathcal U_{r-1}]$, we obtain
\begin{align*}
\mu_r
\ge
\sum_{s=0}^{U_{r-1}-1}
\left(1-\left(1-\frac{1}{|\Psi|}\right)^s\right)
\ge
\frac{U_{r-1}-1}{2}
\left(1-\exp\left(-\frac{U_{r-1}-1}{2|\Psi|}\right)\right).
\end{align*}


Set $x:=(U_{r-1}-1)/(2|\Psi|)$. Since $U_{r-1}\le n$, $f(n)\ge1$, and $|\Psi|\ge c_\Psi nf(n)$, we have $0\le x\le M:=1/(2c_\Psi)$. By convexity of $s\mapsto e^{-s}$ and $e^M\ge1+M$, we have $1-e^{-x}\ge e^{-M}x$. Hence, whenever $U_{r-1}\ge2$,
\begin{align*}
\mu_r
\ge
\frac{e^{-M}(U_{r-1}-1)^2}{4|\Psi|}
\ge
\frac{e^{-M}}{16}\frac{U_{r-1}^2}{|\Psi|}
=
C_0 \cdot \frac{U_{r-1}^2}{|\Psi|},
\end{align*}

where $C_0:=e^{-M}/16$. Choose $c_0=\min\{1,2c_\Psi/C_0\}$. Then $a_0=c_0n\le U_0$ and, since $|\Psi|\ge c_\Psi n$, also $a_0\le2|\Psi|/C_0$. Inductively, if $a_{r-1}\le2|\Psi|/C_0$, then
\begin{align*}
a_r
=
a_{r-1}\frac{C_0a_{r-1}}{2|\Psi|}
\le
a_{r-1}
\le
\frac{2|\Psi|}{C_0}.
\end{align*}

Thus $(a_r)$ is non-increasing. Set $D:=2C_\Psi/(C_0c_0)\ge1$, using $c_0\le2C_\Psi/C_0$. Unfolding the recursion gives
\begin{align*}
a_r
&=
\left(\frac{C_0}{2}\right)^{2^r-1}
\frac{a_0^{2^r}}{|\Psi|^{2^r-1}}
\ge
c_0
\left(\frac{C_0c_0}{2C_\Psi}\right)^{2^r-1}
\frac{n}{f(n)^{2^r-1}}
=
c_0\frac{n}{(Df(n))^{2^r-1}}.
\end{align*}

Next, we relate $a_r$ to the conditional mean $\mu_r$. Condition on $U_{r-1}\ge a_{r-1}$. Since $a_{t-1}\ge2$ and $(a_r)$ is non-increasing, we have $a_{r-1}\ge2$, and hence
\begin{align*}
\mu_r
\ge
C_0\frac{U_{r-1}^2}{|\Psi|}
\ge
C_0\frac{a_{r-1}^2}{|\Psi|}
=
2a_r.
\end{align*}

Using $\mu_r \geq 2a_r$, Lemma~\ref{lem:clique-round-lower-tail} implies $\Pr(U_r<a_r\mid U_{r-1}\ge a_{r-1})\le\exp(-3a_r/26)$, and consequently $\Pr(U_r\ge a_r\mid U_{r-1}\ge a_{r-1})\ge1-\exp(-3a_r/26)$. Since $U_0=n\ge a_0$, chaining these estimates over $r=1,\dots,t$ yields the desired probability bound.
\end{proof} 




The preceding lemma gives a general lower bound on the number of uncoloured vertices after any number of rounds. We next apply this bound in two palette-slack regimes. First, polylogarithmic palette slack still requires essentially $\log\log n$ rounds. Second, subpolynomial palette slack leaves $\omega(1)$ vertices uncoloured after every constant number of rounds a.a.s.

\subsubsection*{Polylogarithmic palette slack}

For a clique on $n$ vertices, the following lemma shows that even polylogarithmic multiplicative slack still requires $\Omega(\log\log n)$ rounds a.a.s. In particular, Corollary~\ref{cor:tight-round-complexity-near-optimal-palette} records the matching lower bound for $|\Psi|=\Theta(n)$. Thus, in the clique case, the general-graph upper bound from Corollary~\ref{cor:near-tight-general-graph-colouring} is tight up to constant factors.

\begin{lemma}[Round lower bound with polylogarithmic palettes]
\label{lem:round-lower-bound-polylogarithmic-palette}
Let $\mathcal C$ be a clique on $n$ vertices and let $|\Psi|=n(\log n)^{O(1)}$. Then RCTDEG requires $\Omega(\log\log n)$ rounds a.a.s.
\end{lemma}

\begin{proof}
Write $|\Psi|\le C_1n(\log n)^{C_2}$ for constants $C_1,C_2>0$, and set $f(n):=|\Psi|/n$. Let $D,c_0>0$ be the constants from Lemma~\ref{lem:clique-general-lower-bound}. Choose a sufficiently large constant $C_3>0$, and set $r:=\lfloor\log_2\log n-\log_2\log\log n\rfloor-C_3$. Then $2^r\le 2^{-C_3}\log n/\log\log n$. Moreover, for all sufficiently large $n$, $\log(Df(n))\le(C_2+1)\log\log n$. Hence
\begin{align*}
(Df(n))^{2^r-1}
\le
\exp\left(2^r\log(Df(n))\right)
\le
\exp\left(2^{-C_3}(C_2+1)\log n\right)
\le
n^{1/2},
\end{align*}
where the last inequality holds by the choice of $C_3$. 

Lemma~\ref{lem:clique-general-lower-bound} therefore gives
$a_r\ge c_0n/(Df(n))^{2^r-1}\ge c_0\sqrt n$. Since $(a_j)$ is non-increasing, this also implies $a_{r-1}\ge a_r\ge2$ for all sufficiently large $n$, and hence $r\exp(-3a_r/26)=o(1)$. Thus, with probability $1-o(1)$, we have $U_r\ge a_r=\omega(1)$. Consequently, RCTDEG cannot have coloured $\mathcal C$ before round $r$, which proves the claim.
\end{proof}

The near-optimal palette case follows immediately.

\begin{corollary}[Tight round complexity with near-optimal palettes]
\label{cor:tight-round-complexity-near-optimal-palette}
Let $\mathcal C$ be a clique on $n$ vertices and let $|\Psi|=\Theta(n)$. Then RCTDEG requires $\Omega(\log\log n)$ rounds a.a.s.
\end{corollary}

\begin{proof}
This is the special case $|\Psi|=n(\log n)^{O(1)}$ of Lemma~\ref{lem:round-lower-bound-polylogarithmic-palette}.
\end{proof}

Together with Corollary~\ref{cor:near-tight-general-graph-colouring}, this gives a tight a.a.s. round bound of $\Theta(\log\log n)$ for RCTDEG on cliques with $|\Psi|=(1+o(1))n$ colours.

\subsubsection*{Subpolynomial palette slack}

Theorem~\ref{thm:constant-round-colouring} shows that Sequential Radial Colouring terminates within a constant number of rounds when $|\Psi|=\lceil\chi^{1+\varepsilon}\rceil$ for a fixed constant $\varepsilon>0$. The following result shows that no constant number of rounds suffices with subpolynomial multiplicative slack: for every constant $t\ge1$, $\omega(1)$ vertices of the inner clique remain uncoloured after $t$ rounds a.a.s.

\begin{lemma}[Clique lower bound with multiplicative slack]
\label{lem:clique-multiplicative-lower-bound}
Let $\mathcal C$ be a clique on $n$ vertices, let $\chi \leq |\Psi|\le n^{1+\varepsilon_n}$ with $n^{\varepsilon_n}\to\infty$, and let $t_n$ satisfy $\varepsilon_n(2^{t_n}-1)\le1-\delta$ for some fixed $\delta\in (0,1)$. Then $U_{t_n}\ge n^{\delta-o(1)}$ a.a.s.; in particular, $U_{t_n}=\omega(1)$ a.a.s.
\end{lemma}

\begin{proof}
Since $\chi \leq |\Psi|$, we have that $|\Psi|\geq n$. We set $f_\Psi(n):=|\Psi|/n$. Then $1\le f_\Psi(n)\le n^{\varepsilon_n}$, and Lemma~\ref{lem:clique-general-lower-bound} applies with $f=f_\Psi$. Hence there are constants $c_0,D>0$ such that the lower-bound sequence satisfies
\begin{align*}
a_{t_n}
\ge
c_0\frac{n}{(D f_\Psi(n))^{2^{t_n}-1}}.
\end{align*}

By assumption, $\varepsilon_n(2^{t_n}-1)\le1-\delta$. Moreover, $2^{t_n}\le1+(1-\delta)/\varepsilon_n=o(\log n)$, since $\varepsilon_n\log n\to\infty$. Hence
\begin{align*}
(D f_\Psi(n))^{2^{t_n}-1}
\le
D^{2^{t_n}-1}n^{\varepsilon_n(2^{t_n}-1)}
\le
n^{o(1)}n^{1-\delta}
=
n^{1-\delta+o(1)}
\end{align*}
and $a_{t_n}\ge n^{\delta-o(1)}=\omega(1)$. Since $(a_r)$ is non-increasing, the hypothesis $a_{t_n-1}\ge2$ of Lemma~\ref{lem:clique-general-lower-bound} holds for all sufficiently large $n$. The lemma then gives $\Pr(U_{t_n}\ge a_{t_n})\ge1-t_n\exp(-3a_{t_n}/26)=1-o(1)$, and hence $U_{t_n}\ge n^{\delta-o(1)}$ a.a.s.
\end{proof}

\begin{corollary}[No constant-round colouring with subpolynomial slack]
\label{cor:no-constant-round-subpoly-slack}
For a clique $\mathcal C$ on $n$ vertices with $|\Psi|=n^{1+o(1)}$, every constant number of rounds $t\ge1$ leaves $U_t=\omega(1)$ a.a.s.; in particular, RCTDEG does not colour $\mathcal C$ in $O(1)$ rounds a.a.s.
\end{corollary}

\begin{proof}
Write $|\Psi|=n^{1+\rho_n}$ with $\rho_n=o(1)$, and set $\varepsilon_n:=\max\{\rho_n,1/\sqrt{\log n}\}$. Then $\varepsilon_n=o(1)$, $n^{\varepsilon_n}\to\infty$, and $|\Psi|\le n^{1+\varepsilon_n}$. Fix a constant $t\ge1$. Since $\varepsilon_n=o(1)$, we have $\varepsilon_n(2^t-1)\le1/2$ for all sufficiently large $n$. Applying Lemma~\ref{lem:clique-multiplicative-lower-bound} with $\delta=1/2$ and $t_n=t$ gives $U_t\ge n^{1/2-o(1)}=\omega(1)$ a.a.s.
\end{proof}

Thus, a polynomial multiplicative increase in the palette size is necessary for RCTDEG to terminate within a constant number of rounds. Table~\ref{tab:clique-lower-bounds} summarises the resulting lower bounds for the clique. 

\begin{table}[t]
\centering
\renewcommand{\arraystretch}{1.35}
\begin{tabular}{l@{\hspace{1.5cm}}l}
\toprule
\textbf{Palette size} & \textbf{Lower bound} \\
\midrule
$\left(1+\frac{1}{\log n}\right)n$
&
$\Omega(\log\log n)$
\\
\addlinespace
$n^{1+\varepsilon_n}$, $n^{\varepsilon_n}\to\infty$
&
$\Omega\!\left(\log\left(1+\frac{1}{\varepsilon_n}\right)\right)$
\\
\addlinespace
$n^{1+o(1)}$
&
$\omega(1)$
\\
\bottomrule
\end{tabular}
\caption{Lower bounds for colouring a clique on $n$ vertices.}
\label{tab:clique-lower-bounds}
\end{table}

\bibliographystyle{plain}
\bibliography{references} 

\appendix
\section{Omitted proofs}\label{sec:proofs}

In this appendix we give proofs of some technical lemmas, which were omitted from the main document. 

\subsection{Degree properties of HRG}
\printProofs[degree]

\subsection{One-round domination tool}
\printProofs[oneRoundDom]



\end{document}